\documentclass[12pt,a4paper]{article}

\usepackage[T1]{fontenc}
\usepackage{lmodern}[lmr]

\usepackage{amsfonts, amsmath, amssymb, amsthm, enumitem, float, mathrsfs, mathtools, multirow, nccmath, rotating, tocloft, setspace, subfig, booktabs, caption}

\usepackage[mathcal]{euscript}
\usepackage[dvipsnames]{xcolor}
\usepackage{hyperref}
\hypersetup{colorlinks,
linkcolor={MidnightBlue},
citecolor={RoyalBlue},
urlcolor={Blue}}
\usepackage[title]{appendix}

\usepackage[nameinlink, noabbrev, capitalize]{cleveref}
\usepackage[hmargin=0.6in, vmargin=0.6in]{geometry}
\usepackage[round]{natbib}
\def\be{\begin{equation}}
\def\ee{\end{equation}}
\def\bea{\begin{eqnarray}}
\def\eea{\end{eqnarray}}

\author{}
\title{}
 
\DeclareMathOperator*{\argmin}{\arg\!\min}

\DeclareMathOperator*{\sgn}{\normalfont\textrm{sgn}}
\DeclareMathOperator*{\diag}{\normalfont\textrm{diag}}

\DeclareMathOperator*{\VEC}{\normalfont\textrm{vec}}

\newtheorem{assumption}{Assumption}
\newtheorem{corollary}{Corollary}

\newtheorem{example}{Example}
\newtheorem{lemma}{Lemma}
\newtheorem{proposition}{Proposition}

\newtheorem{theorem}{Theorem}

\numberwithin{equation}{section}

\renewcommand{\arraystretch}{1.2}

\allowdisplaybreaks[4]

\begin{document}

\setlength{\abovedisplayskip}{8pt}
\setlength{\belowdisplayskip}{8pt}

\begin{titlepage}

\begin{center}

\begin{spacing}{1.5}
{\Large Inference for High-Dimensional Local Projection}
\end{spacing}

\medskip

{\sc Jiti Gao$^{\dag}$,  Fei Liu$^\sharp$ and Bin Peng$^{\dag}$}

\medskip

$^\dag$Monash University and $^\sharp$University of Bath

\bigskip

\today
\end{center}

\begin{abstract}
    This paper rigorously analyzes the properties of the local projection (LP) methodology within a high-dimensional (HD) framework, with a central focus on achieving robust long-horizon inference.  We integrate a general dependence structure into $h$-step ahead forecasting models via a flexible specification of the residual terms. Additionally, we study the corresponding HD covariance matrix estimation, explicitly addressing the complexity arising from the long-horizon setting. Extensive Monte Carlo simulations are conducted to substantiate the derived theoretical findings. In the empirical study, we utilize the proposed HD LP framework to study the impact of business news attention on U.S. industry-level stock volatility.

    \medskip

    
    \medskip

\end{abstract}

\end{titlepage}

\section{Introduction}\label{Sec1}

The local projection (LP) method, introduced in the seminal work of \cite{Oscar_2005}, has garnered significant attention in macroeconomics and econometrics. We refer interested readers to \cite{JT2025} for the latest comprehensive literature review.

A key advantage of the LP method, which has led to its widespread adoption and numerous extensions, is its simplicity: impulse responses can be recovered via a set of simple least squares regressions, bypassing traditional short-run, long-run, or sign restrictions. Consequently, researchers have intuitively applied different generalizations of the least squares regression to the LP method. The literature has thus expanded to include extensions concerning nonlinear structures, robustness estimation, and such models associated with high-dimensional (HD) regressors (covariates), just to name a few. Given the vastness of this literature, we acknowledge the limitation of the following literature review, and summarize only the key articles that directly motivate our research, after which we situate our specific contribution within this field. 

The first stream we review concerns the robustness and inference of the LP method. For example, \cite{montiel2021local} establish long-horizon inference and provide a practical bootstrap method for conducting it. \cite{xu2023local} extends this work, demonstrating the efficiency of the LP method within a vector autoregression of infinite order (VAR($\infty$)) framework. Furthermore, both \cite{herbst2024bias} and \cite{mei2023nickell} investigate the biases associated with the LP method under different model specifications, drawing attention to its application in panel data studies. Given the inherent connection between panel and HD data, a group of studies focuses on the application of the LP method to HD settings.  Recently, \cite{ASW2024} apply the desparsified least absolute shrinkage and selection operator (LASSO) to carry on the HD LP approach, while leaving the impulse response parameter of interest unpenalized; \cite{cha2024} considers a similar setting without specifically relying on sparsity assumptions, and establishes some useful concentration inequalities. Methodologically, \cite{basu2015regularized} and \cite{MIAO2023155} provide useful insights by analyzing specific HD--VAR models; these serve as essential benchmarks for investigating the LP approach in certain HD settings.

Up to this point, we emphasize the necessity of integrating the LP approach within a HD framework, recognizing the intrinsic relationship between LP and VAR models. To illustrate this, consider $N$ time series observed over $T$ time periods. A critical challenge arises in the HD--LP setup: the effective sample size is only of the order $T$, yet the number of unknown parameters can scale as $N$ or even $N^2$ (e.g., as shown in \eqref{def.xth_2} of Example \ref{Ex3} and \eqref{def.xth} of Section \ref{Sec2} respectively).   In the latter case, this means even a modest cross-sectional dimension, such as $N=10$, yields at least $100\cdot p$ unknown parameters in which $p$ is the number of lags to be specified. It then calls for a toolkit to understand and implement the LP approach robustly in HD settings, thereby enabling accurate estimation of impulse response functions when the cross-sectional dimension is large -- a common scenario in empirical macroeconomics and finance. To the best of our knowledge, very limited research has investigated long-horizon inference for the LP approach within the HD setting. In this regard, we believe that we make a significant contribution by deriving a set of basic results (e.g., concentration inequality and HD covariance estimation) for HD $h$-step ahead forecasting models in the context of long-horizon analysis.
 
The second main stream that we review applies the LP method to nonlinear models. For instance, \cite{BB2019} consider the LP approach based on a varying coefficient framework. \cite{GONCALVES2021107} introduce a semi-parametric model and examine the validity of the LP method, while \cite{INOUE2024105726} further introduce a nonparametric time-varying framework in order to model local instabilities. In sharp contrast to this existing literature, we motivate and model the nonlinearity via the residual terms as proposed in the relevant literature, such as in Assumption 2 of \cite{WP2009}, and Assumption A.3 of \cite{DG2018}. This approach has direct and critical impacts on the long-horizon inference established in our paper, thereby offering a distinct path to further generalize the aforementioned nonlinear models.

The third stream that is relevant to our work considers simultaneous equations and IV approaches. For a comprehensive survey of the simultaneous equations, we refer interested readers to \cite{GONCALVES2024105702} for details, wherein they also refer to this line of research as ``state-dependent local projection''.  \cite{PW2021} investigate the IV based approach, and further bridge the LP method and the VAR literature.  Although we do not have any specific IV included, our framework can be considered as a set of simultaneous equations. While one could impose additional identification restrictions, similar to those often applied in the VAR literature, these are generally application-driven; hence, we do not pursue this avenue further in the current paper.

\medskip

Having reviewed the relevant literature, we summarize our contributions below:

\begin{enumerate}[leftmargin=24pt, parsep=2pt, topsep=2pt] 
    \item We establish long-horizon inference specifically for the LP approach in a class of HD moving average infinity--HDMA($\infty$) models, which is considerably different from such models associated with high--dimensional regressors (covariants) discussed in the relevant literature. Our approach relies only on the underlying DGP, circumventing the need for restrictive assumptions like mixing conditions or near-epoch dependence imposed in the relevant literature.
    
    \item While \cite{BL2008} lay the foundation for HD covariance matrix estimation, our study offers the first application of this approach to long--horizon inference within the literature. Based on this result, we establish some asymptotic normality extending the long-horizon inference analysis (\citealp{montiel2021local}) to our class of HD--MA($\infty$) models. 
    
    \item Motivated by Assumption 2 of \cite{WP2009} and Assumption A.3 of \cite{DG2018}, we model the nonlinearity via the residual terms, and offer a distinct path to further generalize the existing literature of nonlinear models. The main challenge is how to integrate this general structure into the $h$-step ahead forecasting models and work out its theoretical properties for such a class of HDMA($\infty$) models.
    
    \item In addition, we conduct extensive simulations to examine the theoretical results. In the empirical study, we utilize the proposed HD--LP framework to revisit the impact of business news attention on U.S. industry-level stock volatility. Though the structural VAR models, as a powerful toolkit, have been widely used in the relevant literature \citep[e.g.,][]{diebold2009measuring, diebold2014network} to trace how volatility shocks propagate across markets or industries, such methods impose rigid assumptions on system dynamics. Our method, on the other hand, allows for more flexibility in capturing the effect of news shocks and the volatility spillover effects among industries. Finally, we compute impulse-response functions at different horizons, and summarize how news-driven volatility shocks evolve over time.
\end{enumerate}

The remainder of the paper is organized as follows. Section \ref{Sec2} introduces the model setup and methodology, and establishes the asymptotic results. In Section \ref{Sec3}, we conduct extensive simulations to corroborate the theoretical findings. Section \ref{Sec4} presents the empirical study on the impact of business news attention on U.S. industry-level stock volatility. Section \ref{Sec5} concludes. Appendix \ref{AP.A} provides supplementary results, while Appendix \ref{AP.B} contains the technical proofs and preliminary lemmas.

\medskip

Before proceeding further, we introduce some notations which will be repeatedly used in the paper. For a positive integer $L$, we let $[L]\coloneqq \{1,\ldots,N\}$. Let $\mathbf{e}_i $ be a $N\times 1$ selection vector with the $i^{th}$ element being 1 and the others being 0, and $i\in [N]$. For a vector $\mathbf{v} =(v_1,\ldots, v_N)^\top$, we let $\|\mathbf{v} \|_1\coloneqq \sum_{i=1}^N |v_i|$, and $\|\mathbf{v} \|_2\coloneqq \{\sum_{i=1}^N |v_i|^2\}^{1/2}$. For a matrix $\mathbf{B}\coloneqq \{b_{ij}\}$, we let $\|\mathbf{B}\|$, $\|\mathbf{B}\|_2$, $\|\mathbf{B}\|_{\max}$, $\|\mathbf{B}\|_1$, and $\|\mathbf{B}\|_\infty$ define its Frobenius norm, spectral norm, max norm, column norm, and row norm respectively. Additionally, we define the following operator for a symmetric square matrix $\mathbf{B}$:

\begin{eqnarray}\label{Def.GM}
    \mathscr{G}_\eta(\mathbf{B}) \coloneqq \{ b_{ij}I(|b_{ij}| \ge \eta) \quad \text{for}\quad i\ne j\},
\end{eqnarray}
where $\eta$ is a user-specified threshold parameter. For a random variable $z$, let $\|z\|_p\coloneqq (E|z|_{p}^p )^{1/p}$. For $a\in \mathbb{R}$, let $\lfloor a \rfloor$ be the largest positive integer less than or equal to $a$. For two numbers, we let $a\lesssim b$ stand for $a=O(b)$; let $a\asymp b$ stands for $a\lesssim b$ and $b\lesssim a$;  $a\wedge b$ and $a\vee b$ stand for $\min(a,b)$ and $\max(a,b)$ respectively. For two random variables (say, $a$ and $b$ again for simplicity), we let $a\lesssim b$ stand for $a=O_P(b)$, and accordingly define $a\asymp b$.

\section{Setup and Asymptotic Properties}\label{Sec2}

In this section, we first present the DGP and justify its generality (Section \ref{Sec2.0}). Next, Section \ref{Sec2.1} introduces the setup for $h$-step ahead forecasting models and their useful properties. We then detail the estimation procedure and establish its corresponding asymptotics under minimal conditions in Section \ref{Sec2.2}. Finally, Section \ref{Sec2.3} imposes more structural assumptions, specifically considering a HD covariance matrix estimation, and establishes the asymptotic distribution for practical analysis. Secondary results are relegated to Appendices \ref{AP.A.1}-\ref{AP.A.3} for the sake of space.

\subsection{Data Generating Process}\label{Sec2.0}

Suppose that the following panel dataset is observable.

\begin{eqnarray}\label{def.xit}
    \{x_{it}\mid i\in [N], \, t =2-p,3-p, \ldots, T\},
\end{eqnarray}
where $i$ indexes the number of individuals, $t$ indexes the time periods, and $p \, (\ge 1)$ is introduced to accommodate the lag terms in the dynamic structure to be specified and estimated later. In an AR(1) model, $p=1$ ensures the observable dataset is $\{x_{it}\mid i\in [N], \, t \in [T]\}$. The parameter $p$ is preserved in \eqref{def.xit} solely to simplify the notation during theoretical derivations.

Additionally, suppose the data generating process (DGP) of $\{x_{it}\}$ admits the following HDMA($\infty$) process:

\begin{eqnarray}\label{def.xt}
    \mathbf{x}_t =\sum_{\ell = 0}^{\infty}\mathbf{B}_\ell \pmb{\varepsilon}_{t-\ell},
\end{eqnarray}
where, for $\forall \ell\ge 0$,  

\begin{eqnarray*}
    \mathbf{B}_\ell =(\mathbf{B}_{\ell,1},\ldots, \mathbf{B}_{\ell,N}) =(\pmb{\beta}_{\ell,1},\ldots,\pmb{\beta}_{\ell,N})^\top =\{b_{\ell,ij} \}_{N\times N}
\end{eqnarray*}
is an $N\times N$ matrix, $\mathbf{x}_t=(x_{1t},\ldots, x_{Nt})^\top$, and $\pmb{\varepsilon}_{t} =(\varepsilon_{1t},\ldots, \varepsilon_{Nt})^\top$ with $\varepsilon_{it}$ being independent and identically distributed (i.i.d.) over both dimensions. 

The popularity and importance of MA($\infty$) processes is governed by the Wold decomposition (\citealp[p. 15]{fan2003nonlinear}). Its recent extensions to fixed multi--dimensional settings include \cite{gpy2024a} and \cite{ygp2025}. Equation \eqref{def.xt} is a substantial extension of such MA($\infty$) processes to a HDMA($\infty$) setting. 

Throughout, we suppose that $E[\pmb{\varepsilon}_{t}]=\mathbf{0}$ and $E[\pmb{\varepsilon}_{t}\pmb{\varepsilon}_{t}^\top]=\mathbf{I}_N$ without loss of generality. Otherwise, one always encounters the identification issue due to the fact that 

\begin{eqnarray*}
   \mathbf{B}_\ell \pmb{\varepsilon}_{t-\ell} \equiv \mathbf{B}_\ell\mathbf{W} \mathbf{W}^{-1}\pmb{\varepsilon}_{t-\ell}
\end{eqnarray*}
for any conformable matrix $\mathbf{W}$ with invertibility. 

Our goal is to infer the impulse responses defined below via the LP approach.

\begin{eqnarray}\label{def.IR}
    \textbf{IR}_{h,j} = E[\mathbf{x}_{t+h}\mid \pmb{\varepsilon}_{ t\mid j}(1), \pmb{\varepsilon}_{t-1},\ldots   ]  -E[\mathbf{x}_{t+h}\mid \pmb{\varepsilon}_{ t\mid j}(0), \pmb{\varepsilon}_{t-1},\ldots],
\end{eqnarray}
where $h\ge 1$, $j\in [N]$, and $\pmb{\varepsilon}_{t\mid j}(a) \coloneqq (\varepsilon_{1t},\ldots,\varepsilon_{j-1,t}, a, \varepsilon_{j+1,t},\ldots, \varepsilon_{Nt})^\top$. Using \eqref{def.xt} and \eqref{def.IR}, simple algebra shows that

\begin{eqnarray}\label{def.Bh}
    \textbf{IR}_{h,j} =\mathbf{B}_{h,j},
\end{eqnarray}
so the key question is how to recover $\mathbf{B}_h$ for $h\ge 1$. 

Before proceeding further, we show the flexibility of the above setup. 

\begin{example}\label{Ex1}
\normalfont
If the DGP of \eqref{def.xt} follows a VAR(1) process $\mathbf{x}_t = \mathbf{a}_1\mathbf{x}_{t-1}+\pmb{\varepsilon}_t$, it then yields the following $h$-step ahead forecasting model for $h\ge 1$

\begin{eqnarray}\label{eq.Var1}
\mathbf{x}_{t+h} =\mathbf{a}_1^{h} \mathbf{x}_t+\mathbf{u}_{t,h},
\end{eqnarray}
where $\mathbf{u}_{t,h} \coloneqq \mathbf{a}_1^{h-1}\pmb{\varepsilon}_{t+1}+ \cdots+\mathbf{a}_1\pmb{\varepsilon}_{t+h-1}+\pmb{\varepsilon}_{t+h}.$ Notably, the following two facts hold: (i). $\mathbf{u}_{t,h}$ is $(h-1)$-dependent along the time dimension (i.e., $\mathbf{u}_{t,h}$ and $\mathbf{u}_{s,h}$ are independent for $|t-s|\ge h$); (ii). for $\forall t$, $\{\mathbf{u}_{t,h}\mid h\ge 1 \}$ and $\{\mathbf{x}_s \mid s\le t\}$ are independent of each other. 

One can further extend the above DGP to include (conditional) heteroskedasticity. For example, consider an autoregressive conditional heteroskedasticity (ARCH) structure:

\begin{eqnarray*}
\mathbf{x}_t = \mathbf{a}_1\mathbf{x}_{t-1}+\tilde{\pmb{\varepsilon}}_t,\quad \tilde{\pmb{\varepsilon}}_t =\pmb{\sigma}_t^{1/2}\pmb{\varepsilon}_t, \quad\text{and}\quad \pmb{\sigma}_t  = \mathbf{b}_0+ \mathbf{b}_1^\top \pmb{\sigma}_{t-1} \mathbf{b}_1,
\end{eqnarray*}
where  $\mathbf{b}_0$ and $\mathbf{b}_1$ need to fulfil certain conditions. In this case, the $h$-step ahead forecasting model is almost the same as \eqref{eq.Var1} with minor modification:

\begin{eqnarray}\label{eq.Var12}
\mathbf{x}_{t+h} =\mathbf{a}_1^{h} \mathbf{x}_t+\tilde{\mathbf{u}}_{t,h},
\end{eqnarray}
where $\tilde{\mathbf{u}}_{t,h}  \coloneqq   (\mathbf{a}_1^{h-1}\pmb{\sigma}_{t+1}^{1/2},\ldots, \mathbf{a}_1 \pmb{\sigma}_{t+h-1}^{1/2}, \pmb{\sigma}_{t+h}^{1/2})  \mathbf{v}_{t,h}$ with $\mathbf{v}_{t,h}\coloneqq (\pmb{\varepsilon}_{t+1}^\top,\ldots,\pmb{\varepsilon}_{t+h}^\top)^\top$. In \eqref{eq.Var12}, $\mathbf{v}_{t,h}$ is $(h-1)$-dependent, and for $\forall t$, $\{\mathbf{v}_{t,h}\mid h\ge 1 \}$ and $\{\mathbf{x}_s \mid s\le t\}$ are still independent of each other. 
\end{example}

\begin{example}\label{Ex2}
\normalfont

If the DGP of \eqref{def.xt} follows a VAR($p$) process, we have

\begin{eqnarray}\label{eq.Varp}
    \begin{pmatrix}
    \mathbf{x}_{t} \\ \vdots \\
    \mathbf{x}_{t-p+1}
    \end{pmatrix}  =\begin{pmatrix}
    \mathbf{a}_{1:(p-1)}&\mathbf{a}_p\\
    \begin{matrix}
        \mathbf{I}
    \end{matrix}& \mathbf{0}
    \end{pmatrix}
    \begin{pmatrix}
    \mathbf{x}_{t-1}\\ \vdots \\
    \mathbf{x}_{t-p}
    \end{pmatrix} + \begin{pmatrix}
    \pmb{\varepsilon}_t \\
    \mathbf{0} 
    \end{pmatrix}\eqqcolon  \tilde{\mathbf{A}}_p \mathbf{X}_{t-1} + \pmb{\mathcal{E}}_t,
\end{eqnarray}
where $\mathbf{a}_{1:(p-1)}=(\mathbf{a}_1 , \ldots , \mathbf{a}_{p-1})$, $\mathbf{I}$ and $\mathbf{0}$ are conformable matrix and vector respectively, and the definitions of $\tilde{\mathbf{A}}_p$, $\mathbf{X}_t$ and $\pmb{\mathcal{E}}_t$ are evident. According to Appendix \ref{AP.A}, we can further obtain that 

\begin{eqnarray*}
    \mathbf{x}_{t+h}=\mathbf{S}_{p}\tilde{\mathbf{A}}_p^{h}\mathbf{X}_t + \mathbf{u}_{t,h},
\end{eqnarray*}
where $\mathbf{S}_{p} =(\mathbf{I}_N,  \mathbf{0}_{N\times N(p-1)} )$, $\mathbf{u}_{t,h} \coloneqq \tilde{\mathbf{a}}_p^{h-1}\pmb{\varepsilon}_{t+1}+ \cdots+ \tilde{\mathbf{a}}_p\pmb{\varepsilon}_{t+h-1}+\pmb{\varepsilon}_{t+h}$, and $\tilde{\mathbf{a}}_p^\ell\coloneqq\mathbf{S}_{p} \tilde{\mathbf{A}}_p^{\ell}\mathbf{S}_{p}^\top$. Still, $\mathbf{u}_{t,h}$ is $(h-1)$-dependent, and for $\forall t$, $\{\mathbf{u}_{t,h}\mid h\ge 1 \}$ and $\{\mathbf{x}_s \mid s\le t\}$ are independent of each other.

Again, one can extend \eqref{eq.Varp} to capture (conditional) hetroskedasticity, and we will no longer discuss it due to similarity.
\end{example}

\begin{example}\label{Ex3}
\normalfont

We now provide an example to connect our study with some popular examples of the literature. For notational simplicity  we focus on the case with one lag only, i.e., Example \ref{Ex1} again. The extension with multiple lags such as Example \ref{Ex2} should be straightforward in view of the following justification.

Firstly, we partition $\mathbf{x}_t$ of Example \ref{Ex1} into two parts ($\mathbf{x}_t^*$ and $\mathbf{x}_t^\dag$) as follows:

\begin{eqnarray}\label{eq.ex3}
    \begin{pmatrix}
    \mathbf{x}_t^* \\
    \mathbf{x}_t^\dag
    \end{pmatrix} = \begin{pmatrix}
    \mathbf{a}_1^* & \mathbf{a}_1^\dag\\
    \mathbf{a}_2^* & \mathbf{a}_2^\dag
    \end{pmatrix}\begin{pmatrix}
    \mathbf{x}_{t-1}^* \\
    \mathbf{x}_{t-1}^\dag
    \end{pmatrix}+\begin{pmatrix}
    \pmb{\varepsilon}_t^* \\
    \pmb{\varepsilon}_t^\dag
    \end{pmatrix},
\end{eqnarray}
wherein  $\{\pmb{\varepsilon}_t^*\}$ and $\{\pmb{\varepsilon}_t^\dag\}$ are independent of each other. It yields that 

\begin{eqnarray*}
    \left\{ \begin{array}{l}
        \mathbf{x}_t^* = \mathbf{a}_1^* \mathbf{x}_{t-1}^* + \mathbf{a}_1^\dag \mathbf{x}_{t-1}^\dag +\pmb{\varepsilon}_t^*\\
        \mathbf{x}_t^\dag =\mathbf{a}_2^* \mathbf{x}_{t-1}^* + \mathbf{a}_2^\dag \mathbf{x}_{t-1}^\dag +\pmb{\varepsilon}_t^\dag
    \end{array} \right. ,
\end{eqnarray*}
which is then similar to Eq. (2) of \cite{GONCALVES2024105702} with constant parameters.

One may further regulate $\{\mathbf{x}_t^\dag \}$ to be exogenous regressors (i.e., $\mathbf{a}_2^*\equiv \mathbf{0}$), so \eqref{eq.ex3} reduces to

\begin{eqnarray*}
    \begin{pmatrix}
    \mathbf{x}_t^* \\
    \mathbf{x}_t^\dag
    \end{pmatrix} = \begin{pmatrix}
    \mathbf{a}_1^* & \mathbf{a}_1^\dag\\
    \mathbf{0} & \mathbf{a}_2^\dag
    \end{pmatrix}\begin{pmatrix}
    \mathbf{x}_{t-1}^* \\
    \mathbf{x}_{t-1}^\dag
    \end{pmatrix}+\begin{pmatrix}
    \pmb{\varepsilon}_t^* \\
    \pmb{\varepsilon}_t^\dag
    \end{pmatrix},
\end{eqnarray*}
In this case, if we let $\mathbf{a}_2^\dag \equiv \mathbf{0}$, $\{\mathbf{x}_t^\dag \ (\equiv \pmb{\varepsilon}_t^\dag) \}$ become purely white noises. The setting therefore is in the same sprit as in \citet[Eq. (3)]{INOUE2024105726} and \citet[Eq. (3)]{GONCALVES2024105702}. 

Additionally, we may focus exclusively on modeling $\mathbf{x}_t^*$. The first equation of \eqref{eq.ex3} then represents a standard dynamic model incorporating a HD set of exogenous regressors: 

\begin{eqnarray}\label{def.xth_2}
    \mathbf{x}_t^* = \mathbf{a}_1^* \mathbf{x}_{t-1}^* +\mathbf{a}_1^\dag\mathbf{x}_{t-1}^{\dag} + \pmb{\varepsilon}_t^*.
\end{eqnarray}
In the limiting case where $\mathbf{x}_t^*$ is a scalar, $\mathbf{a}_1^\dag$ becomes a HD row vector. In this configuration, \eqref{def.xth_2} aligns with the frameworks established by \citet[Eq. (2.4)]{ASW2024} and \citet[Eq. (2.2)]{cha2024}. 
\end{example}

\subsection{The Setup}\label{Sec2.1}

Up to this point, we have discussed our goal and the flexibility of the DGP process. However, \eqref{def.xt} has infinite parameters, so we have to further impose certain structure in order to get some meaningful results practically while maintaining the assumptions as flexible as possible. In Appendix \ref{AP.A.1}, we provide a Proposition \ref{AP.Prop1} to briefly answer the question that how far we can go in approximating \eqref{def.xt} without knowing the underlying DGP, which however is not the main focus of this paper.

Having said that, we suppose that $\{\mathbf{x}_t\}$ also admit the following set of regression models:

\begin{eqnarray}\label{def.xth}
    \mathbf{x}_{t+h}=\mathbf{A}_1\mathbf{x}_{t} + \cdots + \mathbf{A}_{p}\mathbf{x}_{t-p+1} +\mathbf{u}_{t,h}= (\mathbf{X}_t^\top \otimes \mathbf{I}_N) \VEC(\tilde{\mathbf{A}})+\mathbf{u}_{t,h},
\end{eqnarray}
where  $p$ is a fixed positive constant, $h \ (\ge 1)$ may diverge, $\tilde{\mathbf{A}}=\{\tilde{A}_{ij} \}_{N\times Np}\coloneqq (\mathbf{A}_1,\ldots, \mathbf{A}_p)$, and $\mathbf{u}_{t,h}=(u_{1t,h},\ldots, u_{Nt,h})^\top$ has an $(h-1)$-dependent structure to be specified below. It should be understood that $(\mathbf{A}_1,\ldots, \mathbf{A}_p)$ vary with respect to $h$, which is suppressed in the sub-indices for notational simplicity when no misunderstanding arises.  The set of regression models given in \eqref{def.xth} is similar to Eq. (2) of \cite{Oscar_2005} with a focus on long-horizon inference under the HD setting. 

To facilitate development, we impose the following assumptions.

\begin{assumption}\label{AS1}
$\{\varepsilon_{it}\}$ are  i.i.d.  over both $i$ and $t$, and satisfy that $E[\varepsilon_{11}]=0$, $E[\varepsilon_{11}^2]=1$, and $E|\varepsilon_{11}|^J<\infty$ with a fixed $J \, (\ge 4)$.
\end{assumption}

Assumption \ref{AS1} only regulates the random components of \eqref{def.xt}. We will impose more restrictions on the deterministic matrices involved whenever necessary.
 
\begin{assumption}\label{AS2}
\item
\begin{enumerate}[leftmargin=24pt, parsep=2pt, topsep=2pt]
    \item Assume that $g(\cdot)$ is a smooth function such that $\mathbf{u}_{t,h}=\mathbf{g}(\pmb{\varepsilon}_{t+h},\ldots, \pmb{\varepsilon}_{t+1})$ satisfies that  $E[\mathbf{u}_{t,h}]=\mathbf{0}$, $E[\mathbf{u}_{t,h}\mathbf{u}_{t,h}^\top]=\pmb{\Sigma}_h$, and $\max_{i,h}\|\mathbf{e}_i^\top\mathbf{u}_{t,h} \|_J<\infty$, where $J$ is given in Assumption \ref{AS1}.
    
    
    \item Suppose that (i). $\max_{j}\sum_{\ell = 0}^{\infty} \ell^{\frac{J}{2(J+1)}} \|\pmb{\beta}_{\ell, j} \|_2^{\frac{J}{J+1}}<\infty$ with $\{\pmb{\beta}_{\ell, j}\}$ being defined below (\ref{def.xt}); (ii). $\frac{d_{\pmb{\beta}}^4\log N}{T}\to 0$ with $d_{\pmb{\beta}}\coloneqq  \sum_{\ell=0}^\infty \|\mathbf{B}_\ell\|_{\infty} $; (iii). $h/T\to 0$ as $(h, T)\to (\infty,\infty)$.
\end{enumerate}
\end{assumption}

Assumption \ref{AS2} is readily justified in light of Examples \ref{Ex1} and \ref{Ex2}. Specifically, Assumption \ref{AS2}.2.(i) requires a reasonably slow decay rate for the norm $\|\pmb{\beta}_{\ell, j} \|_2$. At this stage, the sample size $N$ can be large in view of the condition in Assumption \ref{AS2}.2.(ii). Assumption \ref{AS2}.2.(ii) allows that $N$ and $T$ can be proportional to each other. Assumption \ref{AS2}.2.(iii) permits that $h$ can go to $\infty$ as long as $\frac{h}{T}\rightarrow 0$.

With the above setup, we present the first result of this paper.

\begin{proposition}\label{PP.Main1}
    Suppose that $\mathbf{x}_t$ admits both representations \eqref{def.IR} and \eqref{def.xth}, Assumptions \ref{AS1} and \ref{AS2}.1 hold, and $\mathbf{B}_0 =\mathbf{I}_N$. For $\forall h\ge 1$, we obtain that $\mathbf{A}_{1} = \mathbf{B}_{h}  = ({\normalfont \textbf{IR}_{h,1}},\ldots, {\normalfont \textbf{IR}_{h,N}})$.
\end{proposition}

The condition $\mathbf{B}_0 =\mathbf{I}$ serves a standard identification purpose, aligning with  \citet[Eq. (3)]{Oscar_2005}. Without this normalization, we can only determine the relationship $\mathbf{A}_1 \mathbf{B}_0 = \mathbf{B}_h$. Given the definition of the impulse responses in \eqref{def.IR}, this lack of identification means the LP  method would only recover the impulse responses up to a rotation matrix. The matrices $\mathbf{A}_j$ for $j\ge 2$, while being essential components of the underlying model \eqref{def.xth}, are not directly connected to the impulse responses.

To conclude our model setup, we establish the following concentration inequality for the $h$-step ahead forecasting models based on \eqref{def.xt} and \eqref{def.xth}, which will facilitate the estimation of the parameters of interest in the next subsection.

\begin{theorem}\label{TM.Main1}
Under Assumptions \ref{AS1}-\ref{AS2}, there exit positive constants $c_1,c_2,c_3, c_4$ such that for $\forall i,j\in [N]$, $\forall\ell \in 0\cup [p-1]$, and $\delta >c_4\sqrt{T}\mu_{h}^{1+1/J}$

\begin{eqnarray*}
    \Pr\left(\max_{t\in [T-h]}\left|\sum_{s=1}^{t}u_{is,h} x_{j,s-\ell}\right|\ge \delta\right) \le c_1 \frac{T}{\delta^J} \mu_h^{J+1}+c_2\exp\left(-\frac{c_3\delta^2}{T \mu_h^{2+2/J}}\right),
\end{eqnarray*}
where $\displaystyle\mu_h= \max_{j\in [N]}\sum_{\ell=0}^\infty  [(\ell+h)^{\frac{J}{2}-1} \|\pmb{\beta}_{\ell, j} \|_2^J]^{\frac{1}{J+1}}\lesssim h^{\frac{J-2}{2(J+1)}}$.
\end{theorem}

Theorem \ref{TM.Main1} establishes a concentration inequality for the $h$-step ahead forecasting models of \eqref{def.xth}, contributing to the literature on long-horizon inference in a HD setting. Notably, the bound contains the term $\mu_h$, which diverges as the horizon $h$ increases. The term $\mu_h$ is simply due to lack of structure in $\mathbf{u}_{t,h}$ and \eqref{def.xt}. Since $J\ge 4$ as stated in Assumption \ref{AS1}, we obtain that $h^{\frac{J-2}{2(J+1)}}\in [h^{1/5}, h^{1/2}]$. Therefore, in the worst case scenario, $\mu_h$ diverges at the rate $h^{1/2}$. From a modeling perspective, this result explains the inherent loss of accuracy in forecasting models as $h$ becomes large. 

Furthermore, the derivation of this theorem relies only on the underlying DGP, circumventing the need for restrictive assumptions like mixing conditions or near-epoch dependence. In this regard, it substantially improves upon the findings in Section 4 of \cite{cha2024}.

\subsection{Estimation}\label{Sec2.2}

We are now ready to consider the estimation in this subsection. Firstly, we assume $p$ is known, and discuss how to select the optimal lag later.

According to \eqref{def.xth}, we define the following objective function:

\begin{eqnarray}\label{def.obj_fun}
    Q_0(\tilde{\mathbf{a}}) =\frac{1}{T-h}\sum_{t=1}^{T-h}\|\mathbf{x}_{t+h}-(\mathbf{X}_t^\top \otimes \mathbf{I}_N) \VEC(\tilde{\mathbf{a}})\|_2^2,
\end{eqnarray}
where $\tilde{\mathbf{a}}$ is a generic $N\times Np$ matrix. However, due to the large amount of elements included in $\tilde{\mathbf{A}}=\{\tilde{A}_{ij} \}_{N\times Np}$, one cannot fully recover everything from \eqref{def.obj_fun} unless the effective sample size $T-h$ is extremely large practically. Consequently, it might not always be feasible for real data analysis. In order to find a balance between feasibility and generality, we impose sparsity on $\tilde{\mathbf{A}}$ of \eqref{def.xth}, and define some new notation to facilitate the investigation. Let

\begin{eqnarray*}
    &&\mathbf{S}_{\mathscr{A}} \coloneqq \{  I(|\tilde{A}_{ij}| >0)\}_{N\times Np}\quad \text{and} \quad \mathbf{S}_{\bar{\mathscr{A}}} \coloneqq \{ I(\tilde{A}_{ij} =0)\}_{N\times Np}.
\end{eqnarray*}
Apparently, $d_{\mathscr{A}} \coloneqq \|\mathbf{S}_{\mathscr{A}}\|^2$ gives the total number of nonzero elements in $\tilde{\mathbf{A}}$. It is noteworthy that we allow for the case that $\tilde{\mathbf{A}}=\mathbf{0}$, i.e., $d_{\mathscr{A}}=0$. For this extreme case, $\mathbf{x}_t$ reduces to a HD white noise. 

Using \eqref{def.obj_fun}, we then introduce the following estimation procedure to recover $\tilde{\mathbf{A}}$ and its sparsity.

\begin{enumerate}[leftmargin=48pt, parsep=2pt, topsep=2pt]
    \item[\it Step 1]  Initial estimation via LASSO:
\begin{eqnarray}\label{def.argmin0}
    \hat{\tilde{\mathbf{a}}} =\argmin_{\tilde{\mathbf{a}}} (Q_0(\tilde{\mathbf{a}})+\gamma\|\VEC(\tilde{\mathbf{a}})\|_1),
\end{eqnarray}
where $\gamma$ is a tuning parameter. 

\item[\it Step 2] Refined estimation via adaptive LASSO: 
\begin{eqnarray}\label{def.argmin}
     \hat{\tilde{\mathbf{a}}}_{\phi} =\argmin_{\tilde{\mathbf{a}}}( Q_0(\tilde{\mathbf{a}})+ \gamma\|\VEC(\tilde{\mathbf{a}})\circ \pmb{\phi}\|_1),
\end{eqnarray}
where $\pmb{\phi}=(\phi_1,\cdots, \phi_{N^2p})^\top$ is a vector of predetermined weights. 
\end{enumerate}

The above estimation procedure connects the adaptive LASSO developed in \cite{Zou2006} with the LP approach in a HD framework. For the adaptive LASSO, we refer interested readers to \cite{Zou2006} for extensive theoretical investigation and to \cite{mcilhagga2016penalized} for detailed numerical implementation. In the second step, we may let $\phi_\ell=|\hat{\tilde{a}}_{\ell}|^{-\zeta}$ with $\hat{\tilde{a}}_{\ell}$ being the $\ell^{th}$ element of $\hat{\tilde{\mathbf{a}}}$, in which $\zeta>0$ is an arbitrary positive constant, as long as certain conditions to be specified below are fulfilled. 

To proceed, we introduce some additional notation and assumptions. Let $\phi_{\mathscr{A},\ell}$ and $\phi_{\bar{\mathscr{A}},\ell}$ denote the $\ell^{th}$ elements of $\pmb{\phi}_{\mathscr{A}}$ and $\pmb{\phi}_{\bar{\mathscr{A}}}$ respectively, where $\pmb{\phi}_{\mathscr{A}}$ and $\pmb{\phi}_{\bar{\mathscr{A}}}$ contain the elements of $\pmb{\phi}$ that correspond to the non--zero elements of $\VEC(\mathbf{S}_{\mathscr{A}})$ and $\VEC(\mathbf{S}_{\bar{\mathscr{A}}})$. Let $\pmb{\Sigma}_{\mathbf{B}} \coloneqq \{\pmb{\Sigma}_{\mathbf{B}, ji}\}$ with $0\le i,j\le p-1$ and 

    \begin{eqnarray*}
        \pmb{\Sigma}_{\mathbf{B}, ji} =\pmb{\Sigma}_{\mathbf{B}, ij}^\top =\left\{\begin{array}{ll}
            \sum_{\ell = 0}^{\infty}  \mathbf{B}_\ell \mathbf{B}_\ell^\top & \text{if } i=j\\
            \sum_{\ell = 0}^{\infty}  \mathbf{B}_{\ell+j-i} \mathbf{B}_\ell^\top & \text{if } j>i
        \end{array} \right. .
    \end{eqnarray*}

\begin{assumption} \label{AS3}
\item 
\begin{enumerate}[leftmargin=24pt, parsep=2pt, topsep=2pt]
    \item Suppose that $\VEC(\mathbf{a})^\top (\pmb{\Sigma}_{\mathbf{B}} \otimes \mathbf{I}_N)\VEC(\mathbf{a})\ge \alpha \|\VEC(\mathbf{a})\|_2^2$ for $\forall\mathbf{a}\in \mathbb{A}(\mathbf{S}_{\mathscr{A}})$ and $\mathbf{a}\ne \mathbf{0}$ uniformly in $N$, where $\alpha>0$ is a fixed positive constant, and 
    \begin{eqnarray*}
    \mathbb{A}(\mathbf{S}_{\mathscr{A}} ) &\coloneqq &\{\mathbf{a}=\{a_{ij}\}_{N\times Np}\mid \|\VEC(\mathbf{S}_{\bar{\mathscr{A}}}\circ \mathbf{a} ) \|_1\le 3 \|\VEC(\mathbf{S}_{\mathscr{A}}\circ \mathbf{a})\|_1\}.
    \end{eqnarray*}

    \item Suppose that $d_{\mathscr{A}} d_{\pmb{\beta}}^2\frac{\sqrt{\log N}}{\sqrt{T}}\to 0$ and $\frac{N^2 }{T^{J-1}\log N}\to 0$, where $J$ is given in Assumption \ref{AS1}.
\end{enumerate}
\end{assumption}

The first condition of Assumption \ref{AS3}.1 is the so-called restricted eigenvalue condition, which has been fully discussed in the literature. See \citet[pp. 1709-1710]{BRT2009} and \citet[p. 2245]{raskutti10a} for example. In the definition of $\mathbb{A}(\mathbf{S}_{\mathscr{A}} )$, the targeted set $\|\VEC(\mathbf{S}_{\bar{\mathscr{A}}}\circ \mathbf{a} ) \|_1\le 3 \|\VEC(\mathbf{S}_{\mathscr{A}}\circ \mathbf{a})\|_1$ is obvious in view of the development given in \eqref{eq.lasso2_1} of the appendix. The second condition further regulates the sparsity and the sample size involved.

With these in hand, we are able to achieve the following results for the 2-step procedure.

\begin{theorem}\label{TM.Main2}
Let Assumptions \ref{AS1}-\ref{AS3} hold and $\gamma \asymp \frac{\mu_{h}^{1+1/J}\sqrt{\log (N)}}{\sqrt{T}}$, where $\mu_h$ is the same as in Theorem 1. 

\begin{enumerate}[leftmargin=24pt, parsep=2pt, topsep=2pt]
    \item For \eqref{def.argmin0}, the following results hold:
\begin{enumerate}[leftmargin=24pt, parsep=2pt, topsep=2pt]
    \item $\|\VEC(\tilde{\mathbf{A}}-\hat{\tilde{\mathbf{a}}}) \|_2=O_P \big( \frac{\mu_{h}^{1+1/J}\sqrt{d_{\mathscr{A}}\log N}}{\sqrt{T}} \big);$
    \item $\|\VEC(\tilde{\mathbf{A}}-\hat{\tilde{\mathbf{a}}}) \|_1 =O_P \big( \frac{\mu_{h}^{1+1/J}d_{\mathscr{A}}\sqrt{\log N}}{\sqrt{T}} \big)$.
\end{enumerate}
    \item For \eqref{def.argmin}, let the elements of $\pmb{\phi}$ satisfy 

\begin{enumerate}[leftmargin=24pt, parsep=2pt, topsep=2pt]
    \item[i.] $\min_{\ell \in [d_\mathscr{A}]}|\beta_{0,\mathscr{A},\ell}| \gtrsim   \frac{\mu_{h}^{1+1/J}\sqrt{d_{\mathscr{A}}\log N}}{\sqrt{T}} \cdot \max_{\ell\in[d_{\mathscr{A}}]}\phi_{\mathscr{A},\ell}$;
    
    \item[ii.] $ \min_{\ell\in[d_{\bar{\mathscr{A}}}]}\phi_{\bar{\mathscr{A}},\ell} \gtrsim d_{\mathscr{A}}   \max_{\ell\in[d_{\mathscr{A}}]}\phi_{\mathscr{A},\ell}$.
\end{enumerate}
Then $\sgn(\hat{\tilde{\mathbf{a}}}_{\phi}) =\sgn (\tilde{\mathbf{A}})$ with probability approaching one.
\end{enumerate}
\end{theorem}

The first result of Theorem \ref{TM.Main2} explains how the $h$-step ahead forecasting framework affects the estimation results when $h$ diverges. The second result proves the sign consistency, i.e., the identification of 0's of $\tilde{\mathbf{A}}$, wherein the additional restrictions can be easily fulfilled in view of the discussion under \eqref{def.argmin}.

\medskip

\noindent \textbf{Selection of the lags} \ With Theorem \ref{TM.Main2} in hand, we can then choose the lag terms involved. The following selection criterion is based on the first step of the estimation procedure (i.e., \eqref{def.argmin0}).

Define the following information criterion:

\begin{eqnarray*}
    \text{IC}(\mathfrak{p})=\frac{1}{T-h}\sum_{t=1}^{T-h} \|\mathbf{x}_{t+h}-(\mathbf{X}_t^\top \otimes \mathbf{I}_N) \VEC(\hat{\tilde{\mathbf{a}}}_{\mathfrak{p}})\|_2^2 + \mathfrak{p}\cdot \xi,
\end{eqnarray*}
where $\hat{\tilde{\mathbf{a}}}_{\mathfrak{p}}$ is obtained via \eqref{def.argmin0} using $\mathfrak{p}$ lags, and $\xi$ is a tuning parameter satisfying certain condition to be specified. The optimal lag is then estimated by 

\begin{eqnarray}\label{def.phat}
    \hat{p} =\argmin_{\mathfrak{p}\le \mathfrak{p}^*} \text{IC}(\mathfrak{p}),
\end{eqnarray}
where $\mathfrak{p}^*$ is a user-specified large and fixed integer.

\begin{theorem}\label{TM.Main3}
    Let Assumptions \ref{AS1}-\ref{AS3} hold, $\frac{\mu_h^{2+2/J} d_{\mathscr{A}} \log N}{T\xi}\to 0$, and $\xi\to 0$. Then $\Pr(\widehat{p} = p)\to 1$.
\end{theorem}

In order to select the optimal number of lags practically, we make two comments:

\begin{enumerate}[leftmargin=24pt, parsep=2pt, topsep=2pt]
    \item An intuitive choice of $\xi$ is to let $\xi\asymp \gamma$, so the condition $\frac{\mu_h^{2+2/J} d_{\mathscr{A}} \log N}{T\xi}\to 0$ reduces to $\frac{\mu_h^{1+1/J} d_{\mathscr{A}} \sqrt{\log N}}{\sqrt{T} }\to 0$, which is required by Theorem \ref{TM.Main2}.1 already. Therefore, such a choice does not create any additional restriction.
    
    \item In fact, one can further simplify the selection process. It is worth pointing out that \cite{Oscar_2005} defines the LP approach as a series of regressions such as those in \eqref{def.xth}, where the number of lags is independent of the horizon $h$. Theorem \ref{TM.Main3} establishes a result that holds for $\forall h$. Therefore, from a practical standpoint, we suggest that using a small horizon (e.g., $h=1$ or $2$) is preferable for selecting the optimal lag, as this significantly simplifies the required restrictions. For instance, the term $\mu_h$ no longer plays a role in such cases. In the simulation, we shall further examine this point.
\end{enumerate}

\subsection{Asymptotic Distribution}\label{Sec2.3} 

It is worth mentioning that we have not imposed too many restrictions on $\mathbf{u}_{t,h}$ so far. However, to achieve the HD covariance estimation and establish the asymptotic normality in what follows, we need to add more structures to the residuals (i.e., $\mathbf{u}_{t,h}$'s).

\medskip

Firstly, we consider the estimation of the covariance matrix. To be precise, our goal is to estimate $\pmb{\Omega}_h \coloneqq E[\tilde{\pmb{\Omega}}_h]$, where

\begin{eqnarray}\label{def.Cov}
    \tilde{\pmb{\Omega}}_h \coloneqq \frac{1}{T-h}\sum_{|t-k|< h}^{T-h} (\mathbf{X}_t  \otimes \mathbf{I}_N)\mathbf{u}_{t,h} \mathbf{u}_{k,h}^\top(\mathbf{X}_k^\top \otimes \mathbf{I}_N).
\end{eqnarray}

Note that if $\tilde{\mathbf{A}}$ is sparse, it will naturally pass the sparsity to $\pmb{\Sigma}_h$, and consequently pass the sparsity to $\pmb{\Omega}_h$. To see this point, consider an extreme case using Example \ref{Ex1}. 

\medskip

\noindent \textbf{Example 1 (Cont.)} If $\mathbf{a}_1=\mathbf{0}$, then $\pmb{\Sigma}_h=\mathbf{I}_N$ and $\pmb{\Omega}_h =\mathbf{I}_{N^2}$ for $\forall h\ge 1$. 

\medskip

\noindent The justification of the above statement should be obvious, so we omit the details. Therefore, without loss of generality, we suppose that $\pmb{\Omega}_h\in \mathbb{U}(c_a, c_N, \bar{c})$, where

\begin{eqnarray}\label{def.dense}
    \mathbb{U}(c_a, c_N, \bar{c}) =\big\{\pmb{\Omega}=\{\omega_{ij}\} \mid \omega_{ii}\le \bar{c}\text{ and } \sum_{j} |\omega_{ij}|^{c_a}\le c_N \text{ for all }i\big\},
\end{eqnarray}
$0<c_a<1$, and $\bar{c}$ is a fixed number. The restriction \eqref{def.dense} is the so-called densest sparse condition. We refer interested readers to \cite{BL2008} and extensive extensions since then for  discussion on \eqref{def.dense}.

We now move to estimate $\pmb{\Omega}_h$. Recall the the operator $\mathscr{G}_\eta(\cdot)$ defined in Section \ref{Sec1}, and let 

\begin{eqnarray}\label{def.Omegahat}
    \hat{\pmb{\Omega}}_h \coloneqq \frac{1}{T-h}\sum_{|t-k|< h}^{T-h} (\mathbf{X}_t  \otimes \mathbf{I}_N)\hat{\mathbf{u}}_{t,h} \hat{\mathbf{u}}_{k,h}^\top(\mathbf{X}_k^\top \otimes \mathbf{I}_N),
\end{eqnarray}
where $\hat{\mathbf{u}}_{t,h}\coloneqq \mathbf{x}_{t+h}-(\mathbf{X}_t^\top \otimes \mathbf{I}_N) \VEC(\hat{\tilde{\mathbf{a}}}_{\phi})$. Putting these together, the estimator of $\pmb{\Omega}_h$ is given by 

\begin{eqnarray}\label{def.Omegahat_Eta}
    \mathscr{G}_\eta(\hat{\pmb{\Omega}}_h),
\end{eqnarray}
where $\eta  \asymp \sqrt{h\log N /T}$. 

In order to investigate \eqref{def.Omegahat_Eta}, we impose some extra conditions in Assumption \ref{AS4}.

\begin{assumption}\label{AS4}
\item 
\begin{enumerate}[leftmargin=24pt, parsep=2pt, topsep=2pt]
\item There exists a constant $c_\delta>2$ such that for $l=1,\ldots, h-1$ with $h>1$
\begin{eqnarray*}
    \max_i\|\mathbf{e}_i^\top (\mathbf{u}_{t,h} - \mathbf{u}_{t,h,l}^*)\|_{J} \eqqcolon \max_{i}\delta_{J,i}(h,l) =O(l^{-c_\delta}),
\end{eqnarray*}
where $\mathbf{u}_{t,h,l}^*$ is the coupled version of $\mathbf{u}_{t,h}$ by replacing $\pmb{\varepsilon}_{t+h-l}$ with an independent copy $\pmb{\varepsilon}_{t+h-l}^*$.
\item There exists a constant $c_b>2$, such that $\max_j\sum_{\ell=m+1}^{\infty} \|\pmb{\beta}_{\ell,j}\|_2^{2}=O (m^{1-2c_b})$  as $m\to \infty$. 
\item Suppose that for $J>4$, 
\begin{enumerate}
    \item $\mu_{h}^{2+2/J}d_{\mathscr{A}}^2\sqrt{h\log N}/\sqrt{T}\rightarrow 0$,
    \item $N^4T^{1-J/4}(\log N)^{-J/4}(\log T) ( T^{(-J/4)/c_0} + h^{- J/4-1} )\rightarrow 0,$
\end{enumerate}
where $c_0 = \big(\frac{2c_b-1}{2}\big)\big(\frac{c_\delta\wedge c_b-1}{c_\delta\wedge c_b}\big)$.
\end{enumerate}
\end{assumption}

Note that $h$ is involved in a few places of Assumption \ref{AS4}. The first condition essentially regulates the decay rate of the time series dependence of $\{\mathbf{u}_{t,h}\}$, while the second condition imposes a further restriction on the matrices involved in \eqref{def.xt}. To ensure the third condition holds, $J$ has to be larger than 4.

\begin{theorem}\label{TM.Main4}
Under Assumptions \ref{AS1}-\ref{AS4}, 

\begin{enumerate}[leftmargin=24pt, parsep=2pt, topsep=2pt]
    \item $\max_{i,j}|\hat{\Omega}_{ij}-\Omega_{ij}|=O_P (\sqrt{h\log N /T} )$,
    \item $ \|\mathscr{G}_\eta(\hat{\pmb{\Omega}}_h)- \pmb{\Omega}_h \|_2=O_P((h\log N /T)^{(1-c_a)/2} c_{N})$,
\end{enumerate}
where $\hat{\Omega}_{ij}$ and $\Omega_{ij}$ are the $(i,j)^{th}$ elements of $\hat{\pmb{\Omega}}_h$ and $\pmb{\Omega}_h$, respectively.
\end{theorem}

Based on Theorem \ref{TM.Main4}, we can then derive a central limit theory for the purpose of inference, which further extends the long--horizon inference of \cite{montiel2021local} to a HD setting. 

Specifically, we adopt the node-wise LASSO method to construct the debiased estimator. For notational simplicity, define $\mathbf{Z}_t\coloneqq \mathbf{X}_t\otimes \mathbf{I}_N$. Accordingly,  for $i\in[N^2p]$, let $\mathbf{Z}_{i,t}$ and $\mathbf{Z}_{-i,t}$ be the $i^{th}$ row of $\mathbf{Z}_{t}$ and the submatrix of $\mathbf{Z}_{t}$ constructed by removing its $i^{th}$ row, respectively. Additionally, define $\pmb{\rho}_{\mathbf{a}}=(\rho_1,\cdots, \rho_{N^2p})^\top$ as the selection vector such that $\|\pmb{\rho}_{\mathbf{a}}\|_1<\infty$. For each $i$, let  $s_{\mathscr{A},i} \coloneqq \|\{I(|\Sigma_{s,ij}|>0)\}_{Np\times 1}\|^2$ denote the sparsity of the $i^{th}$ row of $\pmb{\Sigma}_{\mathbf{B}}^{-1} $, where $\Sigma_{s,ij}$ denotes its $(i,j)^{th}$ element of  $\pmb{\Sigma}_{\mathbf{B}}^{-1} $.

Then, we conduct the node-wise LASSO estimation:

\begin{eqnarray}\label{def.node}
    \hat{\mathbf{b}}_i=\argmin_{\mathbf{b}_i}\Big\{\frac{1}{T-h}\sum_{t=1}^{T-h}\|\mathbf{Z}_{i,t}^\top-\mathbf{Z}_{-i,t}^\top\mathbf{b}_i \|_2^2+2\tilde{\gamma}_i\|\mathbf{b}_i\|_1\Big\},
\end{eqnarray}
where $\tilde{\gamma}_i\asymp  \frac{\mu_{h}^{1+1/J}\sqrt{\log (N)}}{\sqrt{T}}$. Let $\hat{\tau}_i^2=\frac{1}{T-h}\sum_{t=1}^{T-h}\|\mathbf{Z}_{i,t}^\top-\mathbf{Z}_{-i,t}^\top\hat{\mathbf{b}}_i \|_2^2+\tilde{\gamma}_i\|\hat{\mathbf{b}}_i\|_1$, and we then define the debiased estimator \(\hat{\tilde{\mathbf{a}}}_{\mathrm{bc}}\) as

\begin{eqnarray}\label{def.node2}
    \VEC(\hat{\tilde{\mathbf{a}}}_{\mathrm{bc}})
    &=& \VEC(\hat{\tilde{\mathbf{a}}})
    +\frac{1}{T-h}\,\hat{\pmb{\Omega}}_z\,\sum_{t=1}^{T-h}\mathbf{Z}_t\big(\mathbf{x}_{t+h}-\mathbf{Z}_t^\top \VEC(\hat{\tilde{\mathbf{a}}})\big),
\end{eqnarray}
where \(\hat{\pmb{\Omega}}_z=\hat{\mathbf{T}}_z^{-1}\hat{\mathbf{C}}_z\), \(\hat{\mathbf{T}}_z=\mathrm{diag}(\hat{\tau}_1^2,\ldots,\hat{\tau}_{N^2p}^2)\), and \(\hat{\mathbf{C}}_z=(\hat{\mathbf{C}}_{z,1},\ldots,\hat{\mathbf{C}}_{z,N^2p})^\top\) is an \(N^2p\times N^2p\) matrix.  Each row \(\hat{\mathbf{C}}_{z,i}^\top\) is constructed by placing a 1 in the \(i^{th}\) position and filling the remaining \(N^2p-1\) entries with the elements of \(-\hat{\mathbf{b}}_i\) in their natural order.  

Finally, it is noteworthy that under the HD--LP framework, the node-wise LASSO estimation can be much simplified numerically in practical analysis. For the sake of space, we provide the details in Appendix \ref{AP.A.5}.

We are now ready to establish the following CLT for the debiased estimator.

\begin{theorem}\label{TM.Main5}
Let Assumptions \ref{AS1}-\ref{AS4} hold and $ (d_{\mathscr{A}}+\max_i s_{\mathscr{A},i})(\log N) T^{-1/2}\mu_{h}^{2+2/J}\rightarrow0 $. Then 

\begin{eqnarray*}
    \sqrt{T}\,(\pmb{\rho}_{\mathbf{a}}^\top\hat{\pmb{\Omega}}_z\mathscr{G}_\eta(\hat{\pmb{\Omega}}_h)\hat{\pmb{\Omega}}_z\pmb{\rho}_{\mathbf{a}})^{-1/2}\pmb{\rho}_{\mathbf{a}}^\top\VEC(\hat{\tilde{\mathbf{a}}}_{\rm bc}-\tilde{\mathbf{A}})\xrightarrow{~D~}\mathcal{N}(0,1).
\end{eqnarray*}
\end{theorem}

We have emphasized that the preceding framework is not limited to the large-sample $N\to \infty$ regime; it also encompasses such settings associated with fixed $N$ as special cases. A corollary detailing the corresponding asymptotic distribution for the fixed---$N$ case is provided in Appendix \ref{AP.A.0} below.

Up to this point, we have completed our theoretical investigation.  It is worth mentioning again that some additional secondary results are also presented in Appendix \ref{AP.A} for the sake of space. In what follows, we examine these theoretical findings via numerical analyses.

\section{Simulation}\label{Sec3}

Without loss of generality, let the true DGP be as follows: 

\begin{eqnarray}\label{sim.1}
    \mathbf{x}_t = \mathbf{a}_1\mathbf{x}_{t-1} + \mathbf{a}_2\mathbf{x}_{t-2}+\pmb{\varepsilon}_t,
\end{eqnarray}
where $\pmb{\varepsilon}_t\sim N(\mathbf{0},\mathbf{I}_N)$. By \eqref{sim.1}, it is obvious that $p=2$. The coefficient matrices $\mathbf{a}_1=\{a_{1,ij}\}_{N\times N}$ and $\mathbf{a}_2=\{a_{2,ij}\}_{N\times N}$ are designed as follows:

\begin{eqnarray*}
    a_{1, ij}=\left\{ \begin{array}{ll}
        0.25 & \text{if $i=j$ and $i\le N/2$} \\
        0.35 & \text{if $i-j=1$} \\
        0 & \text{otherwise}
    \end{array} \right. \quad\text{and}\quad a_{2, ij}=\left\{ \begin{array}{ll}
        0.35 & \text{if $i=j$ and $i> N/2$} \\
        -0.25 & \text{if $j-i=1$} \\
        0 & \text{otherwise}
    \end{array} \right. .
\end{eqnarray*}
Also, we let $t\in \{-T_B, -T_B+1,\ldots, T \}$, in which $T_B$ is a sufficiently large positive number corresponding to a burn-in period. According to Example \ref{Ex2}, the model \eqref{sim.1} admits an HDMA($\infty$) representation $\mathbf{x}_t = \sum_{\ell =0}^{\infty}\mathbf{B}_\ell \pmb{\varepsilon}_{t-\ell},$ in which 

\begin{eqnarray*}
    \mathbf{B}_\ell=\mathbf{S} \mathbf{C}^\ell\mathbf{S}^\top,\quad \mathbf{S} =(\mathbf{I}_N, \mathbf{0}_N),\quad\text{and}\quad \mathbf{C}=\begin{pmatrix}
    \mathbf{a}_1 & \mathbf{a}_2 \\
    \mathbf{I}_n & \mathbf{0}
\end{pmatrix}.
\end{eqnarray*}

For each generated dataset, we fit the observations $\{\mathbf{x}_t\mid, t=2-p, 3-p,\ldots, T\}$ to the following set of $h$-step ahead forecasting models:

\begin{eqnarray}\label{sim.2}
    \mathbf{x}_{t+h} = \mathbf{A}_1\mathbf{x}_{t} + \mathbf{A}_2\mathbf{x}_{t-1}+ \mathbf{u}_{t,h},
\end{eqnarray}
where $h\ge 1$. We suppress the subindex $h$ in $\mathbf{A}_1$ and $\mathbf{A}_2$ when no misunderstanding arises. According to Proposition \ref{PP.Main1}, we have $\mathbf{A}_1=\mathbf{B}_h$. Additionally, for $h=1$, \eqref{sim.1} and \eqref{sim.2} infer that $\mathbf{a}_1=\mathbf{A}_1=\mathbf{B}_1$.

In what follows, we consider the cases with $N\in \{20, 30, 40\}$ and $T\in \{300, 400, 500\}$, so the number of parameters involved ranges from $800$ to $3200$. Based on the above settings, we firstly present Figure \ref{Fig.Sparse} to demonstrate the sparsity involved in $\mathbf{B}_h$. It is clear that the sparsity exists in all $\mathbf{B}_h$'s, and also the pattern varies with respect to $h$. For example, when $h=10$, $\mathbf{B}_h$ has more non-zero elements, but many of them are close to zero and may become negligible. 

\begin{figure}[htb!]
    \centering 
    \caption{Sparsity of $\mathbf{B}_h$}\label{Fig.Sparse}
    \includegraphics[scale = 0.68]{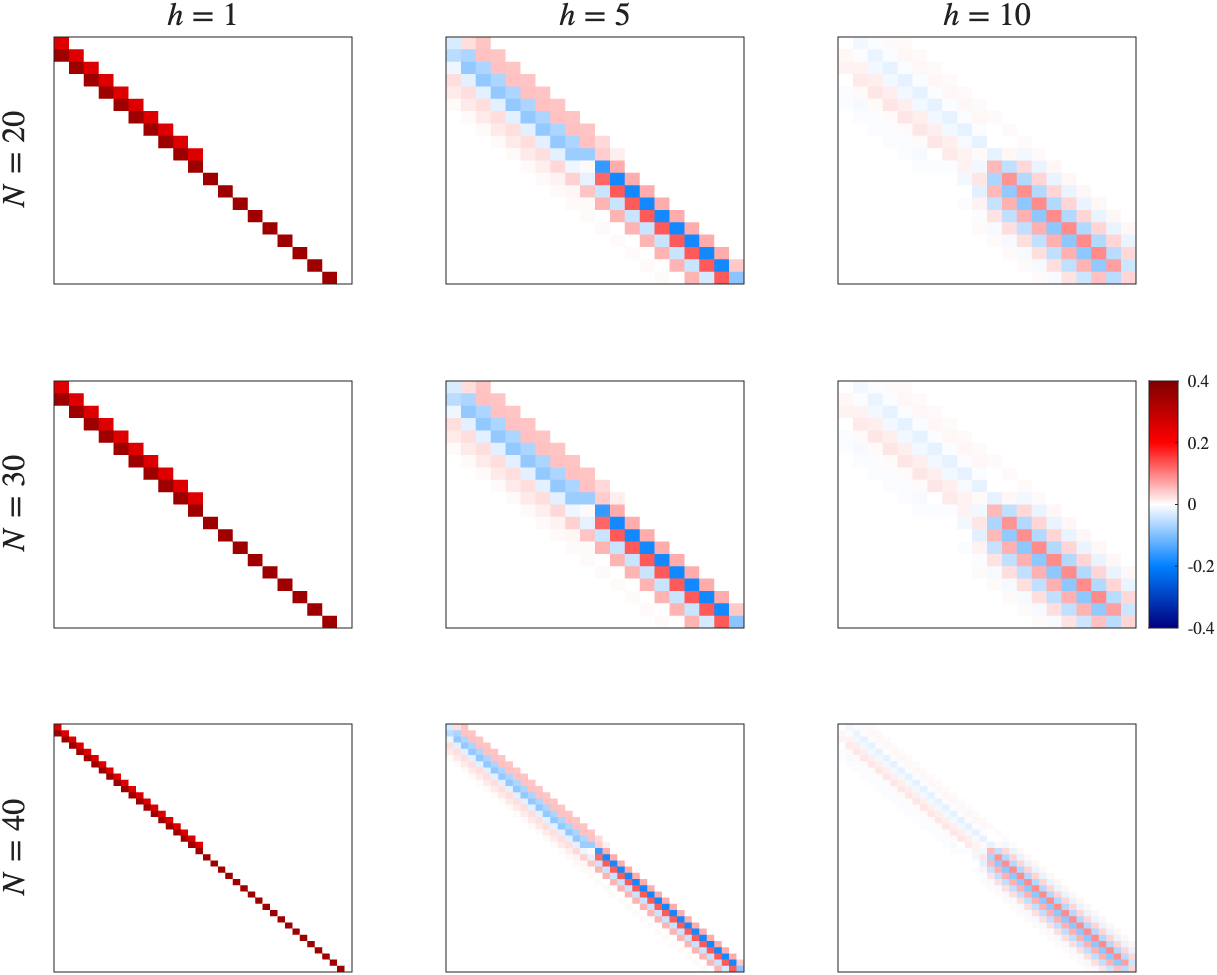}
\end{figure}
 
Based on the specified DGP, we perform $R$ replications to evaluate our proposed estimation procedure. Without loss of generality, we let $R=500$ throughout. The complete numerical algorithm is detailed in Appendix \ref{AP.A.5}. We assess the estimation accuracy of $\mathbf{B}_h =\{ b_{h,ij}\}$ after identifying (1). $\hat{p}$, and (2). the sparsity. 

We start from examining the selection of the optimal lag, and emphasize that two facts:

\begin{enumerate}[leftmargin=24pt, parsep=2pt, topsep=2pt]
    \item Over-selection is generally permissible, as it still yields reasonably reliable (consistent) estimates, though it may result in a decrease in estimation efficiency;
    \item As explained under Theorem \ref{TM.Main3}, we only need to select the optimal $p$ once for the group of regressions such as those in \eqref{def.xth} practically. To minimize the restrictions required, we therefore consider $h=1,2$ only.
\end{enumerate}
To quantify performance, we introduce the following measures: 
\begin{equation*} 
    S_{\ell}^- =\frac{1}{R}\sum_{r=1}^R I(\hat{p}_{\ell, r}<p), \quad S_{\ell} =\frac{1}{R}\sum_{r=1}^R I(\hat{p}_{\ell, r}=p),\quad S_{\ell}^+ =\frac{1}{R}\sum_{r=1}^R I(\hat{p}_{\ell, r}>p),
\end{equation*} 
where $\ell\in \{1,2\}$ corresponds to the lag orders $h\in \{1,2\}$ when implementing \eqref{def.phat}. These indices represent the empirical frequencies of under-selection ($S_{\ell}^-$), correct selection ($S_{\ell}$), and over-selection ($S_{\ell}^+$), where $\hat{p}_{\ell, r}$ denotes the estimated lag in the $r^{th}$ replication for $\forall \ell\in \{1,2\}$.

The selection results are reported in Panel A of Table \ref{Tb.1}.   For $\forall \ell\in \{1,2 \}$, we observe that the values of $S_\ell$ are close to 1, which aligns with theoretical expectations. When $\ell =2$, we have a small proportion over-section (i.e., $S_{\ell}^+ >0$), and $S_{\ell}$ moves towards 1 as $T$ increases. As explained above, over-section is not ideal but permissible. It is not surprising that $S_2$ decreases slightly as $N$ goes up. As we have seen in Figure \ref{Fig.Sparse}, the coefficients have the strongest signal for $h=1$ when implementing LASSO procedure. Hence, it is reasonable to expect $h=1$ offers better finite sample performance from different perspectives such as those presented by Panel A of Table \ref{Tb.1}.

For each $\hat{p}_\ell$ with $\ell\in \{ 1,2\}$, we carry on the second step of the numerical algorithm of Appendix \ref{AP.A.5}. To measure the performance, we quantify the accuracy of sparsity detection as:

\begin{eqnarray*}
    \text{SL}_{\mid \hat{p}_\ell} &=& \frac{1}{R}\sum_{r=1}^R\frac{1}{N^2}\sum_{i,j}|I(\hat{a}_{a,ij,r}=0) - I(b_{h,ij}=0) |, 
\end{eqnarray*}
where SL stands for the selection, the subscript $\mid \hat{p}_\ell$ indicates the selection results are calculated based on the estimated lag order $\hat{p}_\ell$, $\hat{\mathbf{A}}_{a} =\{ \hat{a}_{a,ij}\}$ defines the adaptive LASSO estimate, and $r$ indexes the $r^{th}$ replication. The SL measure focuses on the overall accuracy of zero/non-zero element classification rather than distinguishing between false negative (incorrectly identified zero elements) and false positive (incorrectly identified non-zero elements). 

We observe in Panel A of Table \ref{Tb.1} that the SL value tends to become relatively large as the horizon $h$ increases. This phenomenon is primarily attributable to a source of false negative, as illustrated in Figure \ref{Fig.Sparse}. Specifically, as $h$ increases, the number of non-zero elements in the true matrix $\mathbf{B}_h$ may rise; however, many of these elements become negligible (i.e., close to zero) due to weak signal propagation. The adaptive LASSO procedure may incorrectly detect these negligible elements as zero, thereby leading to a higher overall SL score.

Finally, based on $\hat{p}_\ell$ and the estimated sparsity, we  introduce the following measure:

\begin{eqnarray*}
    \text{AD}_{\star\mid\hat{p}_\ell} &=& \frac{1}{R}\sum_{r=1}^R\|\hat{\mathbf{A}}_{\star,r} \circ \hat{\mathbf{S}}_r-\mathbf{B}_h\|_2 ,
\end{eqnarray*}
where AD stands for the averaged distance, $\star \in \{a,d\}$, the subscript $\mid \hat{p}_\ell$ again indicates the results are calculated based on the estimated lag order $\hat{p}_\ell$,  and $\hat{\mathbf{S}}_r\coloneqq\{ I( \hat{a}_{a,ij, r}\ne 0)\}$. The definition of $\text{AD}_{\mathbf{B}_h}$ infers that for the debiased LASSO estimator, we only consider those elements which are not identified as 0 by the adaptive LASSO. 

Panel B of Table \ref{Tb.1} summarizes the relevant results. The results conditional on $\hat{p}_1$ and $\hat{p}_2$ are very similar, which should be expected. Despite the relatively large number of parameters, which ranges from 800 to 3200, both the adaptive LASSO estimation error ($\text{AD}_{a}$) and the debiased LASSO estimation error ($\text{AD}_{d}$) are observed to be small. Also, both $\text{AD}_{a}$ and $\text{AD}_{d}$ move towards 0 as $T$ increases for all $(N,h)$. Overall, the debiased LASSO exhibits superior finite sample performance, which is consistent with the theoretical expectation that debiasing mitigates the shrinkage bias inherent in the standard LASSO procedure.

\begin{table}[htb!]
\centering \footnotesize
\caption{Results of Selection and Estimation}\label{Tb.1}
\setlength{\tabcolsep}{4pt}
\renewcommand{\arraystretch}{1.0}
\begin{tabular}{llcccccclcccccc}
\toprule
 &  & \multicolumn{6}{c}{Panel A} &  & \multicolumn{6}{c}{Panel B} \\ \cline{3-8} \cline{10-15} 
 &  & \multirow{2}{*}{$S_1^-$} & \multirow{2}{*}{$S_{\ell} $} & \multirow{2}{*}{$S_1^+$} & \multicolumn{3}{c}{$\text{SL}_{\mid \hat{p}_1}$} &  & \multicolumn{3}{c}{$\text{AD}_{a \mid \hat{p}_1}$} & \multicolumn{3}{c}{$\text{AD}_{d \mid \hat{p}_1}$} \\
$N$ & $T$ &  &  &  & $h=1$ & $h=5$ & $h=10$ &  & $h=1$ & $h=5$ & $h=10$ & $h=1$ & $h=5$ & $h=10$ \\
20 & 300 & 0.000 & 1.000 & \multicolumn{1}{c|}{0.000} & 0.019 & 0.308 & 0.620 &  & 0.465 & 0.410 & \multicolumn{1}{c|}{0.376} & 0.406 & 0.395 & 0.368 \\
 & 400 & 0.000 & 1.000 & \multicolumn{1}{c|}{0.000} & 0.013 & 0.304 & 0.619 &  & 0.408 & 0.393 & \multicolumn{1}{c|}{0.373} & 0.342 & 0.375 & 0.367 \\
 & 500 & 0.000 & 1.000 & \multicolumn{1}{c|}{0.000} & 0.009 & 0.301 & 0.618 &  & 0.364 & 0.379 & \multicolumn{1}{c|}{0.370} & 0.307 & 0.356 & 0.363 \\
30 & 300 & 0.000 & 1.000 & \multicolumn{1}{c|}{0.000} & 0.023 & 0.216 & 0.448 &  & 0.563 & 0.453 & \multicolumn{1}{c|}{0.411} & 0.542 & 0.446 & 0.410 \\
 & 400 & 0.000 & 1.000 & \multicolumn{1}{c|}{0.000} & 0.018 & 0.214 & 0.448 &  & 0.536 & 0.445 & \multicolumn{1}{c|}{0.412} & 0.499 & 0.435 & 0.410 \\
 & 500 & 0.000 & 1.000 & \multicolumn{1}{c|}{0.000} & 0.015 & 0.213 & 0.448 &  & 0.497 & 0.437 & \multicolumn{1}{c|}{0.411} & 0.426 & 0.424 & 0.409 \\
40 & 300 & 0.000 & 1.000 & \multicolumn{1}{c|}{0.000} & 0.024 & 0.166 & 0.349 &  & 0.591 & 0.471 & \multicolumn{1}{c|}{0.428} & 0.587 & 0.468 & 0.427 \\
 & 400 & 0.000 & 1.000 & \multicolumn{1}{c|}{0.000} & 0.020 & 0.165 & 0.349 &  & 0.583 & 0.467 & \multicolumn{1}{c|}{0.428} & 0.573 & 0.463 & 0.427 \\
 & 500 & 0.000 & 1.000 & \multicolumn{1}{c|}{0.000} & 0.017 & 0.165 & 0.349 &  & 0.568 & 0.465 & \multicolumn{1}{c|}{0.428} & 0.546 & 0.458 & 0.427 \\
 &  &  &  &  &  &  &  &  &  &  &  &  &  &  \\
 &  & \multirow{2}{*}{$S_2^-$} & \multirow{2}{*}{$S_2$} & \multirow{2}{*}{$S_2^+$} & \multicolumn{3}{c}{$\text{SL}_{\mid \hat{p}_2}$} &  & \multicolumn{3}{c}{$\text{AD}_{a \mid \hat{p}_2}$} & \multicolumn{3}{c}{$\text{AD}_{d \mid \hat{p}_2}$} \\
 &  &  &  &  & $h=1$ & $h=5$ & $h=10$ &  & $h=1$ & $h=5$ & $h=10$ & $h=1$ & $h=5$ & $h=10$ \\
20 & 300 & 0.000 & 0.932 & \multicolumn{1}{c|}{0.068} & 0.019 & 0.308 & 0.620 &  & 0.465 & 0.410 & \multicolumn{1}{c|}{0.376} & 0.406 & 0.396 & 0.369 \\
 & 400 & 0.000 & 0.978 & \multicolumn{1}{c|}{0.022} & 0.013 & 0.304 & 0.619 &  & 0.408 & 0.393 & \multicolumn{1}{c|}{0.373} & 0.343 & 0.375 & 0.367 \\
 & 500 & 0.000 & 0.994 & \multicolumn{1}{c|}{0.006} & 0.009 & 0.301 & 0.618 &  & 0.364 & 0.379 & \multicolumn{1}{c|}{0.370} & 0.307 & 0.357 & 0.363 \\
30 & 300 & 0.000 & 0.966 & \multicolumn{1}{c|}{0.034} & 0.023 & 0.216 & 0.448 &  & 0.563 & 0.453 & \multicolumn{1}{c|}{0.412} & 0.543 & 0.446 & 0.410 \\
 & 400 & 0.000 & 0.992 & \multicolumn{1}{c|}{0.008} & 0.018 & 0.214 & 0.448 &  & 0.536 & 0.445 & \multicolumn{1}{c|}{0.412} & 0.499 & 0.435 & 0.410 \\
 & 500 & 0.000 & 0.998 & \multicolumn{1}{c|}{0.002} & 0.015 & 0.213 & 0.448 &  & 0.497 & 0.437 & \multicolumn{1}{c|}{0.411} & 0.426 & 0.424 & 0.409 \\
40 & 300 & 0.000 & 0.974 & \multicolumn{1}{c|}{0.026} & 0.024 & 0.166 & 0.349 &  & 0.591 & 0.471 & \multicolumn{1}{c|}{0.428} & 0.587 & 0.468 & 0.427 \\
 & 400 & 0.000 & 0.996 & \multicolumn{1}{c|}{0.004} & 0.020 & 0.165 & 0.349 &  & 0.583 & 0.467 & \multicolumn{1}{c|}{0.428} & 0.573 & 0.463 & 0.427\\
 & 500 & 0.000 & 0.992 & \multicolumn{1}{c|}{0.002} & 0.017 & 0.165 & 0.349 &  & 0.568 & 0.465 & \multicolumn{1}{c|}{0.428} & 0.546 & 0.458 & 0.427 \\
 \bottomrule
\end{tabular}
\end{table}


\section{Empirical Study}\label{Sec4}

In the era of big data, textual sources such as news and policy announcements are increasingly used in economic research. For example, \cite{baker2016measuring} develop an economic policy uncertainty index from newspaper articles and find that higher news-based uncertainty is associated with greater stock price volatility. Similarly, \cite{bybee2024business} construct 180 news attention indices from thousands of Wall Street Journal (WSJ) articles and show that attention topics like recession news have economically large predictive power for future output variables. Their news-attention measures closely track a wide range of macroeconomic and financial series. In particular, it is shown that shifts in WSJ coverage of specific themes coincide closely with observed industry-level  market volatility. 

In this paper, we utilize the proposed high-dimensional local projection framework to revisit the impact of business news attention on U.S. industry-level stock volatility. Additionally, we explicitly include lagged volatilities from all industries as regressors to model cross-sector spillovers.  The spillover of volatility across assets and sectors has long been a focus in finance \citep[e.g.,][]{diebold2009measuring, diebold2014network}. Though the structural VAR models, as a powerful toolkit, have been widely used in the relevant literature to trace how volatility shocks propagate across markets or industries, such methods impose rigid assumptions on system dynamics. Our method, on the other hand, allows for more flexibility in capturing the effect of news shocks and the volatility spillover effects among industries. Finally, we compute impulse-response functions at different horizons to trace the short-, medium-, and long-run effects of a news-attention shock on industry volatility. By deriving these responses, we summarize how news-driven volatility shocks evolve over time.

\subsection{Data and Variables}

To construct U.S. industry-level volatility measures, we collect daily industry portfolio returns for 37 industry categories from Kenneth French's data library.\footnote{The data are available at \url{https://mba.tuck.dartmouth.edu/pages/faculty/ken.french/data_library.html}. We exclude the sector labeled ``Other'' in the original dataset, as it lacks a clear economic interpretation.} Monthly industry volatility is computed as the standard deviation of daily returns within each calendar month, which is a standard approach in the empirical asset pricing and volatility literature. Details on the industry classifications are provided in Table \ref{t_Emp_1} and the descriptive statistics for the resulting industry-level volatility series are reported in Table \ref{t_Emp_2}. To assess the time-series properties of the data, we conduct augmented Dickey--Fuller (ADF) unit root tests. As shown in Table \ref{t_Emp_2}, all industry volatility series are found to be stationary at conventional significance levels. 

For the news attention data, we rely on the monthly topic-level attention indices constructed by \citet{bybee2024business}, which cover 180 distinct news topics over the period 1984-2017.\footnote{The data are available at \url{www.structureofnews.com}.} From this universe, we select 23 topics that are most closely related to economic growth and financial markets. Prior evidence in \citet{bybee2024business} highlights the {\it Recession}  topic (REC) as particularly informative for macro-financial dynamics. Accordingly, in our empirical analysis, we progressively incorporate the recession-related news attention series, followed by the remaining seven topics in the economic growth category and fifteen topics in the financial market category, into the local projection framework to examine their effects on industry-level volatilities. A detailed description of the selected topics is provided in Table \ref{t_Emp_1}. 

By merging the industry volatility dataset with the news attention indices, we obtain a multivariate time-series dataset spanning January 1984 to June 2017 ($T=402$), comprising a total of 60 variables: 37 industry-level volatility series and 23 news topic attention measures. All volatility series and news attention indices are subsequently standardized to have zero mean and a variance of one. We consider three local projection specifications that differ in the set of regressors included. These models are summarized as follows:
\begin{itemize}
    \item \textbf{Model 1}: 37 industry-level volatility variables and the recession-related news attention index;
    \item \textbf{Model 2}: 37 industry-level volatility variables and 8 news attention indices related to economic growth;
    \item \textbf{Model 3}: 37 industry-level volatility variables, 8 news attention indices related to economic growth, and 15 news attention indices related to financial markets.
\end{itemize}

\begin{table}
\centering\footnotesize
\caption{Industry Classification and News Topic Information}
\label{t_Emp_1}
\renewcommand{\arraystretch}{1}
\setlength{\tabcolsep}{2pt}
\begin{tabular}{lllll}
\toprule
\multicolumn{2}{c}{\textbf{Panel A: Industry Classification}} 
& \multicolumn{3}{c}{\textbf{Panel B: News Topics}} \\
\cmidrule(lr){1-2} \cmidrule(lr){3-5} 
 & Industry                                   &  & News Topic              & Category         \\
AG           & Agriculture, forestry, and fishing         & REC         & Recession               & Economic Growth  \\
MN           & Mining                                     & RHI         & Record High             & Economic Growth  \\
OG           & Oil and Gas Extraction                     & EGR         & Economic Growth         & Economic Growth  \\
ST           & Nonmetallic Minerals Except Fuels          & FED         & Federal Reserve         & Economic Growth  \\
CN           & Construction                               & ESD         & European Sovereign Debt & Economic Growth  \\
FD           & Food and Kindred Products                  & PPR         & Product Prices          & Economic Growth  \\
TB           & Tobacco Products                           & OPT         & Optimism                & Economic Growth  \\
TX           & Textile Mill Products                      & MAC         & Macroeconomic Data      & Economic Growth  \\
AP           & Apparel and other Textile Products         & IPO         & IPOs                    & Financial Market \\
WD           & Lumber and Wood Products                   & BYD         & Bond Yields             & Financial Market \\
FU           & Furniture and Fixtures                     & SRT         & Short Sales             & Financial Market \\
PA           & Paper and Allied Products                  & SML         & Small Caps              & Financial Market \\
PR           & Printing and Publishing                    & TBY         & Treasury Bonds          & Financial Market \\
CH           & Chemicals and Allied Products              & INT         & International Exchanges & Financial Market \\
PE           & Petroleum and Coal Products                & CMP         & Exchanges Composites    & Financial Market \\
RB           & Rubber and Miscellaneous Plastics Products & FXM         & Currencies Metals       & Financial Market \\
LE           & Leather and Leather Products               & PAY         & Share Payouts           & Financial Market \\
GL           & Stone, Clay and Glass Products             & BBM         & Bear Bull Market        & Financial Market \\
MT           & Primary Metal Industries                   & CMD         & Commodities             & Financial Market \\
FM           & Fabricated Metal Products                  & TRD         & Trading Activity        & Financial Market \\
MC           & Machinery, Except Electrical               & VIX         & Options VIX             & Financial Market \\
EL           & Electrical and Electronic Equipment        &             &                         &                  \\
CA           & Transportation Equipment                   &             &                         &                  \\
IN           & Instruments and Related Products           &             &                         &                  \\
MF           & Miscellaneous Manufacturing Industries     &             &                         &                  \\
TS           & Transportation                             &             &                         &                  \\
PH           & Telephone and Telegraph Communication      &             &                         &                  \\
TV           & Radio and Television Broadcasting          &             &                         &                  \\
UT           & Electric, Gas, and Water Supply            &             &                         &                  \\
GB           & Sanitary Services                          &             &                         &                  \\
SM           & Steam Supply                               &             &                         &                  \\
WT           & Irrigation Systems                         &             &                         &                  \\
WS           & Wholesale                                  &             &                         &                  \\
RT           & Retail Stores                              &             &                         &                  \\
FI           & Finance, Insurance, and Real Estate        &             &                         &                  \\
SV           & Services                                   &             &                         &                  \\
GV           & Public Administration                      &             &                         &                  \\
\bottomrule
\end{tabular}
\begin{minipage}{1\linewidth}
\footnotesize
\textit{Notes:} Panel A reports the industry classification and corresponding abbreviations used in the empirical analysis. Panel B lists news topic indices used in our study and their broad  categories according to \cite{bybee2024business}'s classification. To distinguish between these two categories of variables, we use two-letter codes for industries and three-letter codes for news-attention indices.
\end{minipage}
\end{table}

\begin{table}
\centering
\caption{Descriptive Statistics}
\label{t_Emp_2}
\footnotesize 
\renewcommand{\arraystretch}{1.1}
\begin{tabular}{lllll lllll}
\toprule
Industry & Mean & StD  & $t$-stats & $p$ value & Industry & Mean & StD  & $t$-stats & $p$ value \\
\midrule
AG       & 1.34 & 0.64 & -2.40     & 0.016     & FM       & 1.00 & 0.58 & -2.92     & 0.004     \\
MN       & 1.81 & 0.94 & -2.52     & 0.012     & MC       & 1.22 & 0.72 & -2.83     & 0.005     \\
OG       & 1.45 & 0.94 & -2.57     & 0.010     & EL       & 1.27 & 0.75 & -2.66     & 0.008     \\
ST       & 1.65 & 0.70 & -2.70     & 0.007     & CA       & 1.07 & 0.61 & -3.03     & 0.003     \\
CN       & 1.37 & 0.86 & -2.66     & 0.008     & IN       & 0.92 & 0.49 & -3.19     & 0.002     \\
FD       & 0.80 & 0.40 & -3.04     & 0.003     & MF       & 1.11 & 0.57 & -2.71     & 0.007     \\
TB       & 1.43 & 0.71 & -2.82     & 0.005     & TS       & 1.07 & 0.58 & -2.89     & 0.004     \\
TX       & 1.19 & 0.71 & -2.70     & 0.007     & PH       & 1.16 & 0.64 & -2.80     & 0.005     \\
AP       & 1.19 & 0.66 & -2.70     & 0.007     & TV       & 1.17 & 0.70 & -3.02     & 0.003     \\
WD       & 1.42 & 0.71 & -2.32     & 0.020     & UT       & 0.74 & 0.50 & -3.61     & 0.001     \\
FU       & 1.13 & 0.68 & -2.61     & 0.009     & GB       & 1.29 & 0.58 & -2.88     & 0.004     \\
PA       & 1.16 & 0.67 & -2.91     & 0.004     & SM       & 1.52 & 3.81 & -7.32     & 0.001     \\
PR       & 1.03 & 0.70 & -2.82     & 0.005     & WT       & 0.15 & 0.88 & -4.37     & 0.001     \\
CH       & 1.00 & 0.53 & -3.10     & 0.003     & WS       & 0.89 & 0.50 & -3.03     & 0.003     \\
PE       & 1.24 & 0.71 & -3.00     & 0.003     & RT       & 0.99 & 0.55 & -2.80     & 0.005     \\
RB       & 1.04 & 0.59 & -2.87     & 0.005     & FI       & 0.92 & 0.71 & -2.69     & 0.007     \\
LE       & 1.38 & 0.64 & -2.35     & 0.018     & SV       & 1.02 & 0.61 & -3.14     & 0.002     \\
GL       & 1.25 & 0.70 & -2.62     & 0.009     & GV       & 2.09 & 4.34 & -7.06     & 0.001     \\
MT       & 1.39 & 0.92 & -2.76     & 0.006     &          &      &      &           &          \\
\bottomrule
\end{tabular}
\end{table}

\subsection{Estimation Results}


The numerical implementation is the same as that documented in Appendix \ref{AP.A.5}. In constructing the asymptotic covariance estimator, we set the threshold parameter to \(\eta = 2\sqrt{h \log N / T}\). To determine the optimal lag order in the local projection regressions, we employ the information criterion proposed in Section \ref{Sec2}. As discussed earlier, it is sufficient to consider relatively small forecasting horizons when selecting the lag length \(p\). Accordingly, we set \(h=1,2,3\) and, for each horizon, compute the values of \(\text{IC}(p)\) for lag orders up to a maximum of ten.  The resulting information criterion values  are reported in Table \ref{t_Emp_3}. The results indicate that the criterion favors \(\widehat{p}=5\) when \(h=1\), and \(\widehat{p}=3\) for both \(h=2\) and \(h=3\). Based on these findings, we adopt \(\widehat{p}=3\) as the common lag order in the subsequent analysis.

The estimated local projection coefficients for different horizons ($h=1,6,12,24$) using Model 1 are reported in Figure \ref{f_Emp_1}. In the short run, the recession-attention news index exerts a broadly positive effect on industry-level return volatility.  This result coincides with the expectation that recession-related news can raise perceived macroeconomic uncertainty in the market, which can in turn increase the cross-sectional dispersion of firms' economic activities (e.g., investment) and eventually increases asset-return volatility. This pattern is consistent with a large literature confirming  that positive contribution of the market uncertainty to the asset return volatility \citep[see][for example]{Bloom2009}.  
The effect turns weak at medium horizons (e.g., $h=6$) and becomes negligible at longer horizons (e.g., $h=24$), indicating that the volatility impulse triggered by recession news is predominantly transitory.  Similar patterns are observed in Figures \ref{f_Emp_2} and \ref{f_Emp_3}, which present the estimated coefficients for Models 2 and 3. This indicates that the positive impact of recession-related news attention on volatility remains robust to the inclusion of additional news indices.
To document cross-industry heterogeneity more precisely, Table \ref{t_Emp_4} reports the estimated short-run ($h=1$) impulse responses and their standard errors for each industry; several industries exhibit estimates exceeding 0.2, whereas many others remain indistinguishable from zero.  These heterogeneous patterns are economically plausible: the short-run responses are strongest in cyclical and finance-sensitive sectors (for example, construction, transportation equipment, and finance sectors), whereas defensive sectors  (such as mining, food or tobacco products) display statistically insignificant responses.  Table \ref{t_Emp_4} also presents the short-run impulse-response estimates  from Models 2 and 3. Even after adding attention series for other news topics that are related to economic growth or financial markets, the recession topic maintains positive impacts on volatility across most industries, though several estimated effects decline in magnitude.

To assess the temporal evolution of recession-driven volatility shocks, we compute impulse responses at horizons spanning one month to two years for each industry. Owing to page constraints, Figure \ref{f_Emp_4} displays the estimated impulse responses and their 90\% confidence intervals for six representative industries (AG, OG, ST, CN, MF, and FI). The results confirm positive, statistically significant short-run effects of recession-related news attention. For example, the value for ST industry remains significantly positive through horizon $h=4$, while the responses for CN and OG industries are significantly positive at horizons up to 3.  At medium and long horizons, on the other hand,  the effects become statistically insignificant as confidence bands widen. These findings corroborate those reported in Figure \ref{f_Emp_1}.

In addition to analyzing the effects of recession-related news on volatility, Models 2 and 3 allow us to examine the influence of other economics- and finance-related news topics. As shown in Figures \ref{f_Emp_2} and \ref{f_Emp_3}, these topics exert divisive effects on industry-level volatilities. In particular, the topic {\it European Sovereign Debt} (ESD) exhibits a strong short-run effect in increasing volatility, and this effect is broadly uniform across industries. This finding is not surprising, as news related to sovereign debt crises, which is similar to the recession-related news, tends to heighten market perceptions of macroeconomic and financial uncertainty, thereby amplifying volatility. By contrast, the {\it Federal Reserve} (FED) topic displays a pervasive volatility-reducing effect across industries. One possible explanation is that Federal Reserve–related news often conveys policy guidance and information about the future path of monetary policy, which can reduce uncertainty by anchoring market expectations.

We then turn to examine volatility spillovers across industries. In our estimates, the Finance, Insurance, and Real Estate (FI) sector emerges as a prominent source of volatility spillovers to other industries at both short and medium horizons, as shown in Figure \ref{f_Emp_1}. This result is consistent with the broader literature, which finds that financial sectors often act as net transmitters of volatility due to their central role in credit intermediation and risk sharing \citep[e.g.,][]{diebold2015financial}. It also extends the existing literature that documents risk spillovers among financial institutions themselves \citep[e.g.,][]{tobias2016covar}, and further underscores the pivotal role of the financial sector in propagating volatility shocks across the entire economic network. By contrast, the transportation industry appears to exert a restraining effect on volatility spillovers, which indicates sectors that are more closely tied to final demand and production may absorb volatility rather than amplify it.

As a robustness check, we replace historical volatility with realized volatility. Specifically, monthly realized volatility is constructed as the square root of the sum of squared daily returns for each industry. We then re-estimate the model using the same local projection framework. Due to page constraints, the corresponding results are reported in Figures \ref{f_Emp_1_RV}, \ref{f_Emp_2_RV}, and \ref{f_Emp_3_RV} in Appendix \ref{AP.A.4}.
The main empirical findings remain largely unchanged. Recession-related news attention continues to exert a positive effect on volatility across most industries, and the presence of volatility spillovers from the financial sector to other sectors is again confirmed.

\begin{table}
\centering
\caption{IC Values for Lag Order Selection}
\label{t_Emp_3}
\footnotesize 
\begin{tabular}{lllllllllll}
\toprule
$h/p$ & 1      & 2      & 3               & 4      & 5               & 6      & 7      & 8      & 9      & 10     \\
\midrule
1     & 17.085 & 17.235 & 17.073          & 16.829 & \textbf{16.772} & 16.914 & 17.207 & 17.925 & 17.986 & 18.050 \\
2     & 19.611 & 19.946 & \textbf{19.793} & 20.103 & 20.694          & 21.180 & 21.989 & 20.446 & 21.711 & 21.753 \\
3     & 24.085 & 23.431 & \textbf{23.178} & 23.650 & 23.689          & 24.039 & 24.610 & 24.461 & 25.006 & 24.878\\
\bottomrule
\end{tabular}
\end{table}

\begin{table}
\centering\footnotesize
\caption{Estimated Impulse Responses of Volatilities to Recession Attention}
\label{t_Emp_4}
\setlength{\tabcolsep}{2pt}
\renewcommand{\arraystretch}{1.2}
\begin{tabular}{lrrlrrlrrlrrlrrlrr}
 \toprule
    \multicolumn{6}{c}{\textbf{Model 1}}                                 & \multicolumn{6}{c}{\textbf{Model 2}} &           \multicolumn{6}{c}{\textbf{Model 3}}       \\
    \cmidrule(lr){1-6} \cmidrule(lr){7-12} \cmidrule(lr){13-18}
Ind. & IR           & StD          & Ind. & IR    & StD   & Ind. & IR           & StD          & Ind. & IR    & StD   & Ind. & IR           & StD          & Ind. & IR    & StD   \\
AG       & 0.167        & 0.081        & FM       &       &       & FM       & 0.188        & 0.096        & FM       & 0.102 & 0.095 & AG       & 0.203        & 0.104        & FM       & 0.124 & 0.125 \\
MN       &              &              & MC       & 0.191 & 0.101 & MC       &              &              & MC       & 0.198 & 0.097 & MN       &              &              & MC       & 0.195 & 0.120 \\
OG       & 0.191        & 0.075        & EL       & 0.184 & 0.108 & EL       & 0.222        & 0.088        & EL       &       &       & OG       & 0.234        & 0.133        & EL       & 0.151 & 0.119 \\
ST       & 0.224        & 0.067        & CA       & 0.211 & 0.079 & CA       & 0.188        & 0.046        & CA       & 0.194 & 0.043 & ST       & 0.217        & 0.093        & CA       & 0.197 & 0.115 \\
CN       & 0.212        & 0.072        & IN       &       &       & IN       & 0.188        & 0.024        & IN       &       &       & CN       & 0.265        & 0.055        & IN       &       &       \\
FD       &              &              & MF       & 0.154 & 0.084 & MF       & 0.098        & 0.104        & MF       & 0.120 & 0.074 & FD       & 0.092        & 0.134        & MF       &       &       \\
TB       &              &              & TS       & 0.169 & 0.091 & TS       &              &              & TS       & 0.150 & 0.071 & TB       &              &              & TS       & 0.143 & 0.050 \\
TX       & 0.193        & 0.078        & PH       & 0.173 & 0.081 & PH       & 0.182        & 0.067        & PH       & 0.148 & 0.065 & TX       & 0.189        & 0.109        & PH       & 0.126 & 0.115 \\
AP       & 0.164        & 0.086        & TV       & 0.192 & 0.079 & TV       & 0.132        & 0.076        & TV       & 0.150 & 0.068 & AP       &              &              & TV       & 0.146 & 0.136 \\
WD       & 0.108        & 0.072        & UT       & 0.227 & 0.090 & UT       & 0.074        & 0.065        & UT       & 0.216 & 0.094 & WD       & 0.127        & 0.109        & UT       & 0.224 & 0.149 \\
FU       & 0.092        & 0.067        & GB       &       &       & GB       & 0.083        & 0.074        & GB       &       &       & FU       & 0.106        & 0.141        & GB       &       &       \\
PA       & 0.089        & 0.077        & SM       &       &       & SM       &              &              & SM       &       &       & PA       & 0.123        & 0.123        & SM       &       &       \\
PR       & 0.136        & 0.063        & WT       &       &       & WT       & 0.115        & 0.059        & WT       &       &       & PR       & 0.152        & 0.106        & WT       &       &       \\
CH       & 0.166        & 0.096        & WS       &       &       & WS       & 0.174        & 0.092        & WS       & 0.102 & 0.074 & CH       & 0.143        & 0.086        & WS       &       &       \\
PE       &              &              & RT       & 0.235 & 0.070 & RT       &              &              & RT       & 0.185 & 0.064 & PE       &              &              & RT       & 0.235 & 0.141 \\
RB       & 0.162        & 0.091        & FI       & 0.183 & 0.072 & FI       & 0.112        & 0.089        & FI       & 0.167 & 0.077 & RB       & 0.128        & 0.139        & FI       & 0.162 & 0.116 \\
LE       & 0.231        & 0.111        & SV       & 0.173 & 0.100 & SV       & 0.213        & 0.125        & SV       &       &       & LE       & 0.186        & 0.165        & SV       & 0.149 & 0.103 \\
GL       & 0.152        & 0.075        & GV       &       &       & GV       & 0.126        & 0.055        & GV       &       &       & GL       & 0.121        & 0.106        & GV       &       &       \\
MT       & 0.162        & 0.068        &          &       &       & MT       & 0.161        & 0.061        &          &       &       & MT       & 0.165        & 0.104        &          &       &    \\  
\bottomrule
\end{tabular}
\begin{minipage}[t]{1\textwidth}
\footnotesize 
\textit{Notes:} This table reports estimated impulse responses (IR) of industry-level volatilities to a recession attention shock, together with their corresponding standard deviations (StD).  All reported results correspond to horizon $h=1$.  Entries left blank indicate coefficients that are not selected by the LASSO procedure and are therefore statistically insignificant.
\end{minipage}

\end{table}

\begin{sidewaysfigure}[htp!]
	\centering
	\subfloat[$h=1$ ]
	{\includegraphics[width=0.4\textwidth]{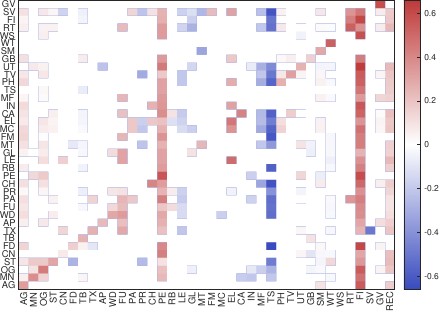}}
	\subfloat[$h=6$ ]
    {\includegraphics[width=0.4\textwidth]{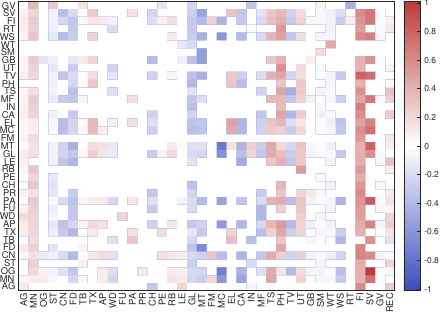}}\\
	\subfloat[$h=12$ ]
	{\includegraphics[width=0.4\textwidth]{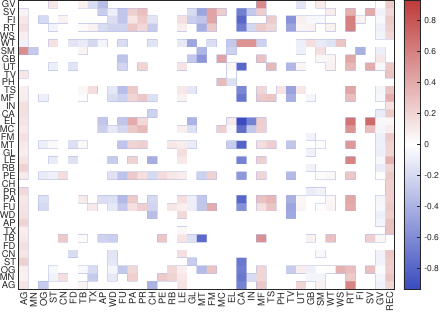}}
	\subfloat[$h=24$ ]
	{\includegraphics[width=0.4\textwidth]{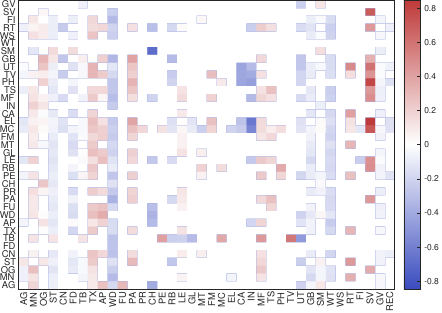}}	\\
	
    \caption{\textbf{Estimated High-Dimensional Local Projection Coefficients for Model 1.}
  \textit{Notes:} Vertical axis shows the local projection response of the dependent variable at horizon \(t+h\); horizontal axis lists the predictor variable.  Responses are shown separately for horizons \(h\in\{1,6,12,24\}\).
    }
	\label{f_Emp_1}
\end{sidewaysfigure}

\begin{sidewaysfigure}[htp!]
	\centering
	\subfloat[$h=1$ ]
	{\includegraphics[width=0.4\textwidth]{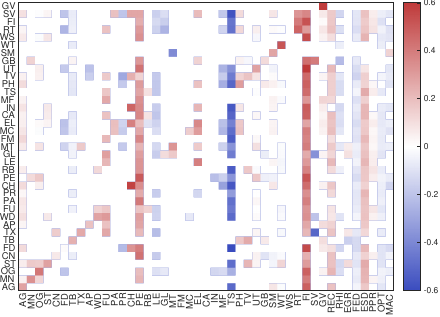}}
	\subfloat[$h=6$ ]
    {\includegraphics[width=0.4\textwidth]{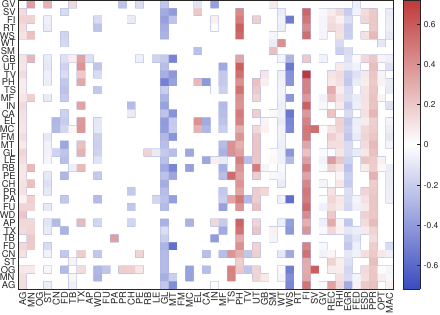}}\\
	\subfloat[$h=12$ ]
	{\includegraphics[width=0.4\textwidth]{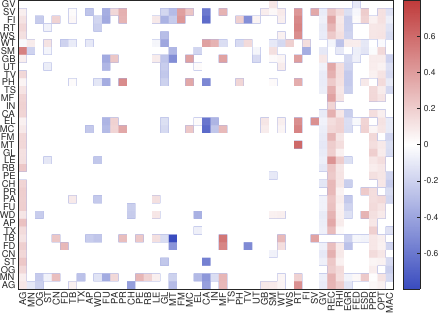}}
	\subfloat[$h=24$ ]
	{\includegraphics[width=0.4\textwidth]{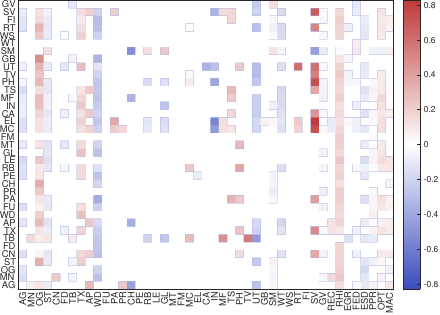}}	\\
	
    \caption{\textbf{Estimated High-Dimensional Local Projection Coefficients for Model 2.}
  \textit{Notes:} Vertical axis shows the local projection response of the dependent variable at horizon \(t+h\); horizontal axis lists the predictor variable.  Responses are shown separately for horizons \(h\in\{1,6,12,24\}\).
    }
	\label{f_Emp_2}
\end{sidewaysfigure}

\begin{sidewaysfigure}[htp!]
	\centering
	\subfloat[$h=1$ ]
	{\includegraphics[width=0.4\textwidth]{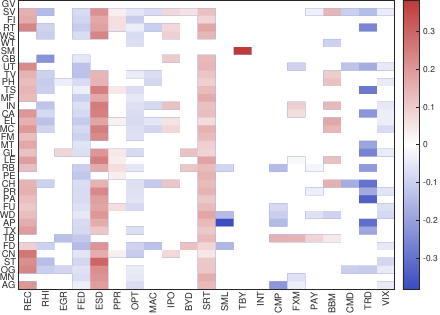}}
	\subfloat[$h=6$ ]
    {\includegraphics[width=0.4\textwidth]{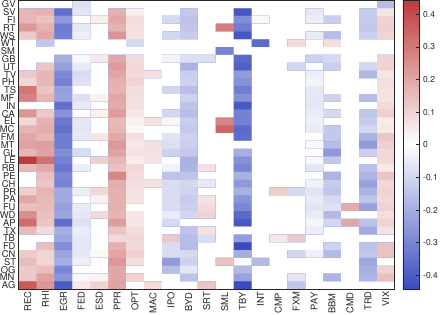}}\\
	\subfloat[$h=12$ ]
	{\includegraphics[width=0.4\textwidth]{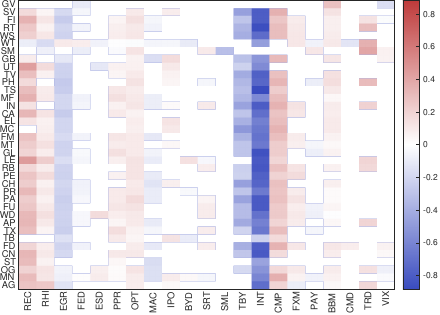}}
	\subfloat[$h=24$ ]
	{\includegraphics[width=0.4\textwidth]{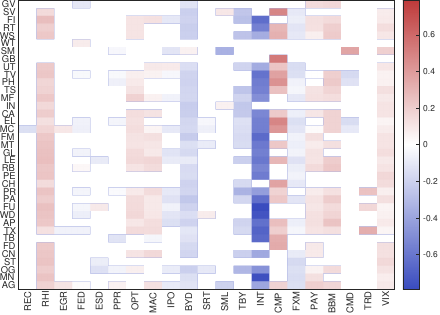}}	\\
	
    \caption{\textbf{Estimated High-Dimensional Local Projection Coefficients for Model 3.}
  \textit{Notes:} Vertical axis shows the local projection response of the dependent variable at horizon \(t+h\); horizontal axis lists the predictor variable.  Responses are shown separately for horizons \(h\in\{1,6,12,24\}\). For brevity, estimates for spillovers within the volatility series are omitted; only the responses of volatility to news-attention indices are reported.
    }
	\label{f_Emp_3}
\end{sidewaysfigure}

\begin{figure}[htp!]
	\centering
	\subfloat[AG ]
	{\includegraphics[width=0.4\textwidth]{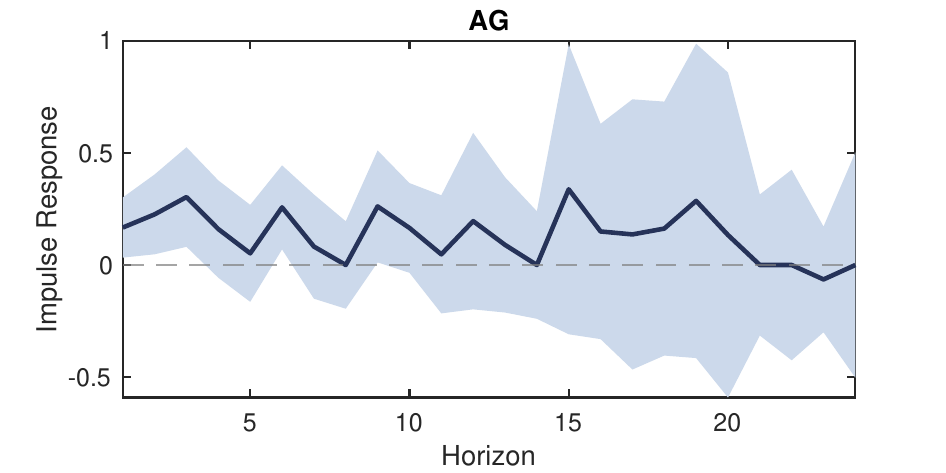}}
	\subfloat[OG]
    {\includegraphics[width=0.4\textwidth]{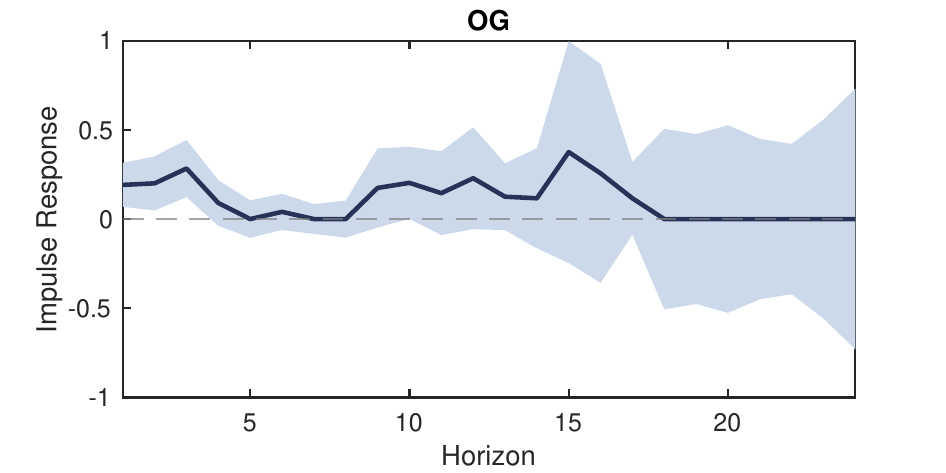}}\\
	\subfloat[ST]
	{\includegraphics[width=0.4\textwidth]{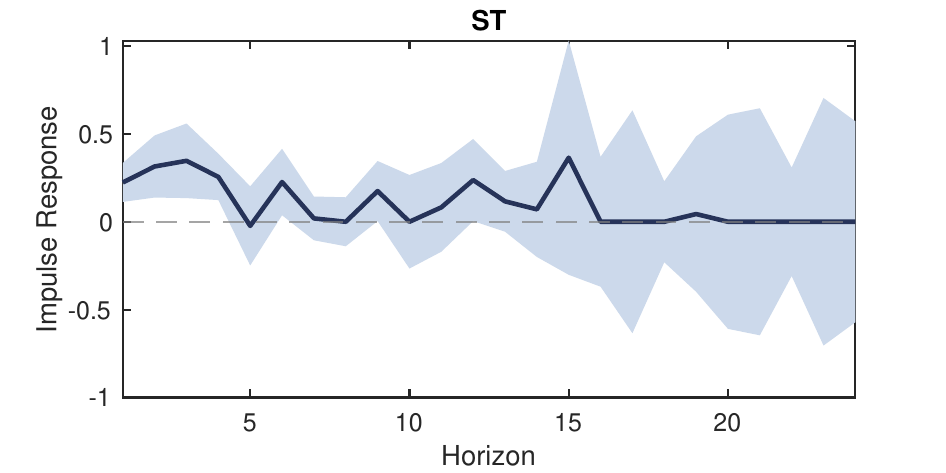}}
	\subfloat[CN ]
	{\includegraphics[width=0.4\textwidth]{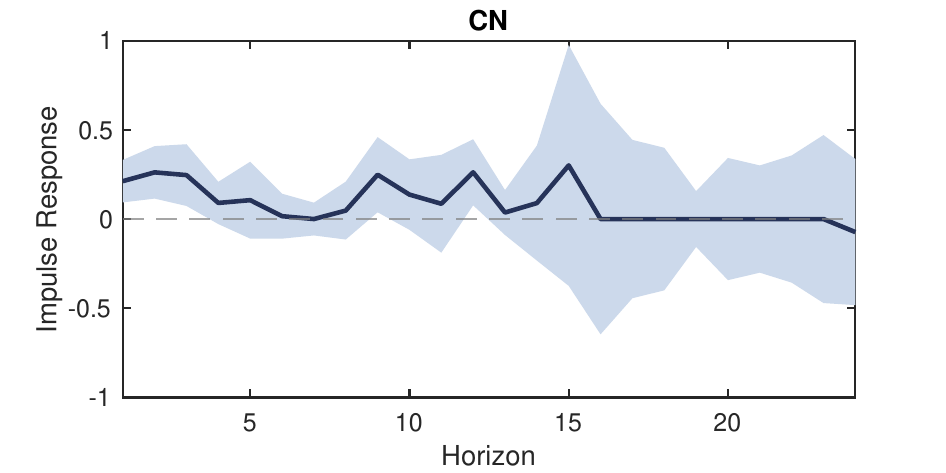}}	\\
    \subfloat[MF ]
    {\includegraphics[width=0.4\textwidth]{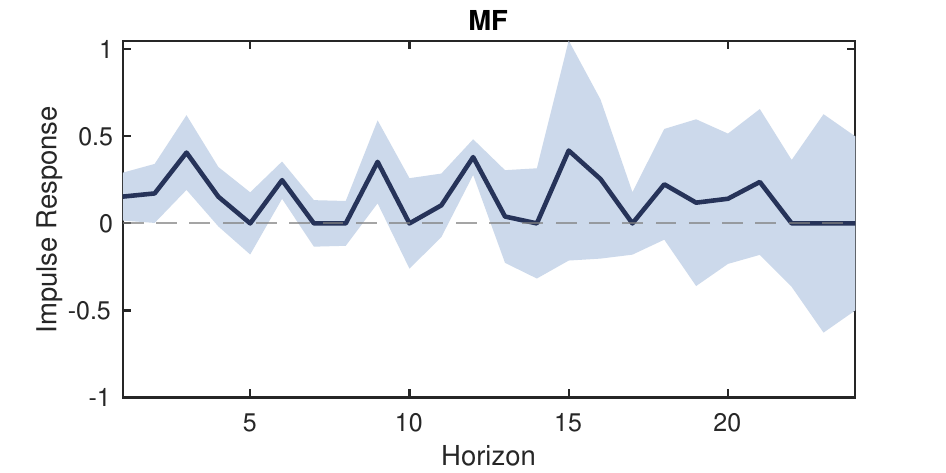}}
	\subfloat[FI]
	{\includegraphics[width=0.4\textwidth]{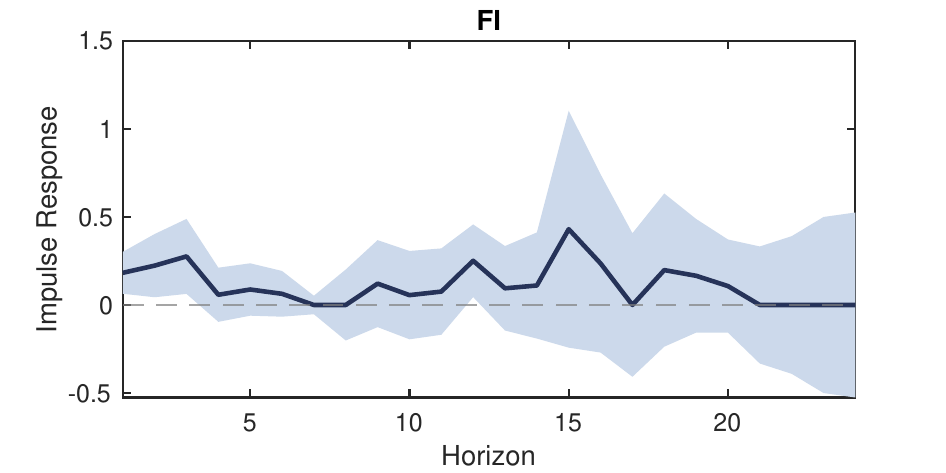}}	\\
    \caption{\textbf{Estimated Impulse Responses of Industry Volatilities to Recession News Attention.} \textit{Notes:} The solid black line represents the estimated impulse response of volatility to the recession-related news attention index, while the light blue shaded area denotes the associated 90\% confidence intervals.
    }
	\label{f_Emp_4}
\end{figure}

\section{Conclusion}\label{Sec5}

In this paper, we rigorously analyze the properties of the LP methodology within a HD framework, with a central focus on achieving robust long-horizon inference.  We integrate a general dependence structure into $h$-step ahead forecasting models via a flexible specification of the residual terms. Additionally, we study the corresponding HD covariance matrix estimation, explicitly addressing the complexity arising from the long-horizon setting. Extensive Monte Carlo simulations are conducted to substantiate the derived theoretical findings. 

In the empirical study, we utilize the proposed HD LP framework to revisit the impact of business news attention on U.S. industry-level stock volatility. Though the structural VAR models, as a powerful toolkit, have been widely used in the relevant literature \citep[e.g.,][]{diebold2009measuring, diebold2014network} to trace how volatility shocks propagate across markets or industries, such methods impose rigid assumptions on system dynamics. Our method, on the other hand, allows for more flexibility in capturing the effect of news shocks and the volatility spillover effects among industries. Finally, we compute impulse-response functions at different horizons, and summarize how news-driven volatility shocks evolve over time. 

{\footnotesize
\bibliography{Ref.bib}
}

{\small

\begin{appendices}

\section{Additional Results and Information}\label{AP.A}

\renewcommand{\thelemma}{A\arabic{lemma}}
\setcounter{lemma}{0} 
\renewcommand{\theproposition}{A\arabic{proposition}}
\setcounter{proposition}{0} 
\renewcommand{\theassumption}{A\arabic{assumption}}
\setcounter{assumption}{0} 

Appendix \ref{AP.A.0} provides a corollary to state the asymptotic distribution for the case with $N$ being fixed. In Appendix \ref{AP.A.1}, we briefly answer that how far we can go in approximating \eqref{def.xt} without knowing the underlying DGP. In Appendix  \ref{AP.A.2}, we explain how our study is connected with some of the existing literature. We justify Example \ref{Ex2} in Appendix \ref{AP.A.3}, provide additional empirical results in Appendix \ref{AP.A.4},  and document the detailed numerical implementation in Appendix \ref{AP.A.5}.

\subsection{Finite-Dimensional Setting}\label{AP.A.0}

In this appendix, we establish an asymptotic theory for the bias--corrected estimator $\hat{\tilde{\mathbf{a}}}_{\rm bc}$ under a finite-dimensional (fixed $N$) setting. In this case, the optimization problem in \eqref{def.argmin0} reduces to a classical LASSO estimation problem with a fixed number of parameters. Consequently, the restricted eigenvalue condition imposed in Assumption \ref{AS3}, which is essential in high-dimensional settings, is no longer required. Instead, standard regularity conditions that ensure the positive definiteness of the population covariance matrix \(\pmb{\Sigma}_{\mathbf{B}}\) are sufficient to establish the asymptotic normality of \(\hat{\tilde{\mathbf{a}}}_{\rm bc}\).

Moreover, in the finite--dimensional framework, the estimation of the asymptotic covariance matrix does not require the thresholding step introduced in \eqref{Def.GM}. Considering all together, these simplifications allow us to derive an asymptotic normality for \(\hat{\tilde{\mathbf{a}}}_{\rm bc}\) as a counterpart of the high--dimensional result in Theorem \ref{TM.Main5}. 

\begin{corollary}\label{TM.Main5_B}
Let Assumptions \ref{AS1}, \ref{AS2}, and \ref{AS4} hold. In addition, suppose that 
\(\lambda_{\min}(\pmb{\Sigma}_{\mathbf{B}}) > 0\) and 
\((d_{\pmb{\beta}}^{2} + \mu_{h}^{2 + 2/J})T^{-1/2} \to 0\). Then,
\begin{eqnarray*}
    \sqrt{T}\,
    (\hat{\pmb{\Omega}}_z \hat{\pmb{\Omega}}_h \hat{\pmb{\Omega}}_z)^{-1/2}
    \VEC(\hat{\tilde{\mathbf{a}}}_{\rm bc} - \tilde{\mathbf{A}})
    \xrightarrow{~D~} \mathcal{N}(\mathbf{0},\mathbf{I}),
\end{eqnarray*}
where $\mathbf{I}$ and $\mathbf{0}$ denote conformable identity matrix and zero vector, respectively.
\end{corollary}

Since this corollary follows directly from Theorem~\ref{TM.Main5} under the finite-dimensional specialization, its proof is omitted.

\subsection{Approximation of \texorpdfstring{HDMA($\infty$)}{HDMA}}\label{AP.A.1}

In this section, we seek to answer that how far we can go in approximating \eqref{def.xt} without knowing the underlying DGP such as \eqref{def.xth}. 

As in \cite{DZ2023}, one may adopt an autoregressive (AR) model to approximate \eqref{def.xt}. Without loss of generality, we examine the suitability of AR($1$) model, i.e.,  finding $\{\tilde{\mathbf{x}}_t\}$ which are close to $\{\mathbf{x}_t \}$ in some sense. Specifically, we consider

\begin{eqnarray}\label{def.regression1}
    \tilde{\mathbf{x}}_t = \mathbf{A} \tilde{\mathbf{x}}_{t-1} + \pmb{\varepsilon}_t,
\end{eqnarray}
where $\pmb{\varepsilon}_t$ is the same as that defined in \eqref{def.xt}, and $\mathbf{A}$ is an $N\times N$ matrix.

Therefore, the question becomes whether there exists an $\mathbf{A}$ such that the difference between $\{\mathbf{x}_t\}$ and $\{\tilde{\mathbf{x}}_t\}$ can be relatively small. We answer this question in the following proposition. Firstly, we introduce an additional assumption.

\begin{assumption}\label{AS.A1}
    Suppose that 
    \begin{enumerate}[leftmargin=24pt, parsep=2pt, topsep=2pt]
    \item $\limsup_{N}\sum_{\ell =0}^{\infty}\ell \cdot\|\mathbf{B}_\ell\|_2<\infty$,
    \item $ \liminf_{N}\lambda_{\min}(\bar{\mathbf{B}}\bar{\mathbf{B}}^\top) >0$, where $\bar{\mathbf{B}} = \sum_{\ell=0}^\infty \mathbf{B}_\ell$. 
    \end{enumerate}
\end{assumption}

This fundamental assumption helps us establish some basic results such as the HD BN decomposition in Lemma \ref{LMB4}, which might be used elsewhere.

\begin{proposition}\label{AP.Prop1}

Suppose Assumptions \ref{AS1} and \ref{AS.A1} hold. There exists a matrix $\mathbf{A}$ satisfying that

\begin{enumerate}[leftmargin=24pt, parsep=2pt, topsep=2pt]
    \item $\limsup_{N}\|\mathbf{A}\|_2<1$,
    \item $\liminf_{N}\lambda_{\min}[(\mathbf{I}-\mathbf{A})(\mathbf{I}-\mathbf{A})^\top]>0$,
\end{enumerate}
such that $\max_{t\in [T]}  \| \sum_{s=1}^t( \mathbf{x}_s - \tilde{\mathbf{x}}_s ) \|_2 =O_P(\sqrt{N}).$
\end{proposition}
The condition $\limsup_{N}\|\mathbf{A}\|_2<1$ is consistent with Assumption A.1.vi of \cite{MIAO2023155}. We provide some remarks below.

\medskip

\noindent\textbf{Remark}:

(a). As shown in the proof, we have  

\begin{eqnarray}\label{def.AB}
    \bar{\mathbf{B}}=(\mathbf{I}-\mathbf{A})^{-1}\quad \text{or}\quad\mathbf{A} = \mathbf{I} - \bar{\mathbf{B}}^{-1}.
\end{eqnarray}
In view of \eqref{def.AB}, if the true DGP of $\{\mathbf{x}_t\}$ is indeed a HD VAR(1) process such as

\begin{eqnarray*}
    \mathbf{x}_t = \mathbf{a} \mathbf{x}_{t-1} + \pmb{\varepsilon}_t,
\end{eqnarray*}
we can further obtain that

\begin{eqnarray}\label{def.Aa}
    \mathbf{A} = \mathbf{I} - \bar{\mathbf{B}}^{-1} = \mathbf{I} - \left(\sum_{\ell=1}^\infty\mathbf{a}^\ell\right)^{-1}= \mathbf{I} - \left((\mathbf{I}-\mathbf{a})^{-1}\right)^{-1} =\mathbf{a},
\end{eqnarray}
which is exactly what we are hoping for when the model is correctly specified. It is worth mentioning that in \eqref{def.Aa}, it is possible to allow for  $\mathbf{A} =\mathbf{a}=\mathbf{0}_{N}$.

(b). For the case with $N$ being fixed, Proposition \eqref{AP.Prop1} yields

\begin{eqnarray*}
    \frac{1}{T}\max_{t\in [T]} \big\| \sum_{s=1}^t( \mathbf{x}_s - \tilde{\mathbf{x}}_s ) \big\|_2 =O_P\left(\frac{1}{T}\right).
\end{eqnarray*}
It then infers that, for the typical fixed dimensional data, one does not need to worry about model misspecification too much.  

(c). For the high dimensional case without knowing the DGP of \eqref{def.xt}, the rate given in Proposition \ref{AP.Prop1} may not be further reduced without additional information. Moreover, the non-negligible bias will generate more accumulated errors when one tries to implement regression with $\{\mathbf{x}_t\}$ practically.

\subsection{Connection with the Existing Literature}\label{AP.A.2}

To close our theoretical investigation, we discuss how our study connects with the existing literature. It is  recognized in the LP literature that under certain conditions the HAC estimator of the covariance matrix can suffer from  small-sample bias \citep[][]{herbst2024bias,xu2023local}, and may even perform worse than the simpler Eicker-Huber-White estimator, which ignores autocorrelation in the errors. To understand this issue, we rewrite the LP regression as follows:
\begin{eqnarray}
    \mathbf{x}_{t+h}&=&\mathbf{A}_1\mathbf{x}_{t} + \cdots + \mathbf{A}_{p}\mathbf{x}_{t-p+1} +\mathbf{u}_{t,h}\notag \\
    &=& (\mathbf{x}_{t}^\top  \otimes \mathbf{I}_N) \VEC(\mathbf{A}_1)+(\mathbf{X}_{t-1}^{\ast\top} \otimes \mathbf{I}_N) \VEC(\tilde{\mathbf{A}}^\ast)+\mathbf{u}_{t,h},
\end{eqnarray}
where $\tilde{\mathbf{A}}^\ast\coloneqq (\mathbf{A}_2,\ldots, \mathbf{A}_{p})$ and $\mathbf{X}_{t-1}^{\ast\top}=(\mathbf{x}_{t-1}^\top, \cdots, \mathbf{x}_{t-p+1}^\top)^\top$. Since we are primarily interested in inference for $\VEC(\mathbf{A}_1)$, standard multivariate regression theory implies that the relevant covariance is obtained by projecting out $\mathbf{X}_{t-1}^{\ast}$ from $\mathbf{x}_{t}$: $\pmb{\varepsilon}^\ast_t=\mathbf{x}_{t}-E[\mathbf{x}_{t}|\mathbf{X}_{t-1}^{\ast}] =\pmb{\varepsilon}_t+ \mathbf{A}_p\mathbf{x}_{t-p}$. Eventually, it requires estimation of the following covariance matrix:
\begin{eqnarray}
    \pmb{\Omega}_1 &\coloneqq&  \frac{1}{T-h}\sum_{|t-k|\le h}^{T-h} E\!\left[(\pmb{\varepsilon}^\ast_t  \otimes \mathbf{I}_N)\mathbf{u}_{t,h} \mathbf{u}_{k,h}^\top(\pmb{\varepsilon}_k^{\ast\top} \otimes \mathbf{I}_N)\right].
\end{eqnarray}
Notably, if the second term in $\pmb{\varepsilon}^\ast_t$ were absent, we would obtain
\begin{eqnarray}\label{discuss1}
    \pmb{\Omega}_1 &=&  \frac{1}{T-h}\sum_{|t-k|\le h}^{T-h} E\!\left[(\pmb{\varepsilon}_t  \otimes \mathbf{I}_N)\mathbf{u}_{t,h} \mathbf{u}_{k,h}^\top(\pmb{\varepsilon}_k^\top \otimes \mathbf{I}_N)\right] \notag \\
    &=&E\!\left[(\pmb{\varepsilon}_1  \otimes \mathbf{I}_N)\mathbf{u}_{1,h} \mathbf{u}_{1,h}^\top(\pmb{\varepsilon}_1^\top \otimes \mathbf{I}_N)\right],
\end{eqnarray}
which implies that no autocorrelation adjustment would be required for estimating $\pmb{\Omega}_1$.

To address this issue in practice, the lag-augmented LP regression proposed by \cite{montiel2021local} is often employed. In this specification, we include one additional lag:

\begin{eqnarray}\label{def.xth_new}
    \mathbf{x}_{t+h}&=&\mathbf{A}_1\mathbf{x}_{t} + \cdots + \mathbf{A}_{p}\mathbf{x}_{t-p+1} + \mathbf{A}_{p+1}\mathbf{x}_{t-p} +\mathbf{u}_{t,h}\notag \\
    &=& (\mathbf{x}_{t}^\top  \otimes \mathbf{I}_N) \VEC(\mathbf{A}_1)+(\mathbf{X}_{t-1}^\top \otimes \mathbf{I}_N) \VEC(\bar{\mathbf{A}}^\ast)+\mathbf{u}_{t,h},
\end{eqnarray}
where $\bar{\mathbf{A}}^\ast\coloneqq (\mathbf{A}_2,\ldots, \mathbf{A}_{p+1})$. Then the inference of $\VEC(\mathbf{A}_1)$ only requires the estimation of \eqref{discuss1} in which the autocorrelation is not involved. While this approach works reasonably well for low dimensional models, in the HD setting, such a strategy may fail. Loosely speaking, adding one additional lag term will automatically include $N^2$ unknown parameters in the regression, as a result it may lead to significant efficiency loss. Not to mention additional restrictions need to be imposed in this scenario.

\subsection{Justification of Example \ref{Ex2}}\label{AP.A.3}

\begin{proof}[Justification of Example \ref{Ex2}]
    \item 
    
Similar to Example \ref{Ex1}, we have 

\begin{eqnarray}\label{Eq.B1}
    \mathbf{X}_{t+h}=\tilde{\mathbf{A}}_p^{h}\mathbf{X}_{t} +\mathbf{U}_{t,h},
\end{eqnarray}
where $\mathbf{U}_{t,h} \coloneqq \tilde{\mathbf{A}}_p^{h-1}\pmb{\mathcal{E}}_{t+1}+ \cdots+\tilde{\mathbf{A}}_p\pmb{\mathcal{E}}_{t+h-1}+\pmb{\mathcal{E}}_{t+h}.$

Note that $\mathbf{S}_p$ satisfies $\mathbf{S}_{p}^\top \mathbf{S}_{p} =\diag\{\mathbf{I}_N,  \mathbf{0}_{N(p-1)} \}$ and $ \mathbf{S}_{p}\mathbf{S}_{p}^\top = \mathbf{I}_N$. Using $\mathbf{S}_p$ and \eqref{Eq.B1}, we write further

\begin{eqnarray*} 
     \mathbf{x}_{t+h} &=& \mathbf{S}_{p}\tilde{\mathbf{A}}_p^{h}\mathbf{X}_t + \mathbf{S}_{p}(\tilde{\mathbf{A}}_p^{h-1}\pmb{\mathcal{E}}_{t+1}+ \cdots+\tilde{\mathbf{A}}_p\pmb{\mathcal{E}}_{t+h-1}+\pmb{\mathcal{E}}_{t+h})\notag \\
     &=& \mathbf{S}_{p}\tilde{\mathbf{A}}_p^{h}\mathbf{X}_t + \mathbf{S}_{p}(\tilde{\mathbf{A}}_p^{h-1}\mathbf{S}_{p}^\top \mathbf{S}_{p}\pmb{\mathcal{E}}_{t+1}+ \cdots+\tilde{\mathbf{A}}_p\mathbf{S}_{p}^\top \mathbf{S}_{p}\pmb{\mathcal{E}}_{t+h-1}+\mathbf{S}_{p}^\top \mathbf{S}_{p}\pmb{\mathcal{E}}_{t+h})\notag \\
     &=& \mathbf{S}_{p}\tilde{\mathbf{A}}_p^{h}\mathbf{X}_t + \mathbf{S}_{p}(\tilde{\mathbf{A}}_p^{h-1}\mathbf{S}_{p}^\top \pmb{\varepsilon}_{t+1}+ \cdots+\tilde{\mathbf{A}}_p\mathbf{S}_{p}^\top \pmb{\varepsilon}_{t+h-1}+\mathbf{S}_{p}^\top  \pmb{\varepsilon}_{t+h})\notag \\
     &=& \mathbf{S}_{p}\tilde{\mathbf{A}}_p^{h}\mathbf{X}_t + \mathbf{u}_{t,h},
\end{eqnarray*}
where $\mathbf{u}_{t,h} \coloneqq \tilde{\mathbf{a}}_p^{h-1}\pmb{\varepsilon}_{t+1}+ \cdots+ \tilde{\mathbf{a}}_p\pmb{\varepsilon}_{t+h-1}+\pmb{\varepsilon}_{t+h}$ and $\tilde{\mathbf{a}}_p^\ell\coloneqq\mathbf{S}_{p} \tilde{\mathbf{A}}_p^{\ell}\mathbf{S}_{p}^\top$. Then the statements provided in Example \ref{Ex2} become obvious.
\end{proof}

\subsection{Additional Empirical Results}\label{AP.A.4}

In this appendix, we report robustness checks for the empirical results. Specifically, we replace historical volatility with realized volatility. The re-estimated local projection coefficients for Models 1, 2, and 3 are presented in Figures \ref{f_Emp_1_RV}, \ref{f_Emp_2_RV}, and \ref{f_Emp_3_RV}, respectively.

\begin{sidewaysfigure}[htp!]
	\centering
	\subfloat[$h=1$ ]
	{\includegraphics[width=0.4\textwidth]{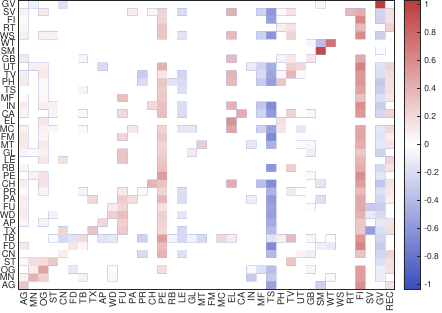}}
	\subfloat[$h=6$ ]
    {\includegraphics[width=0.4\textwidth]{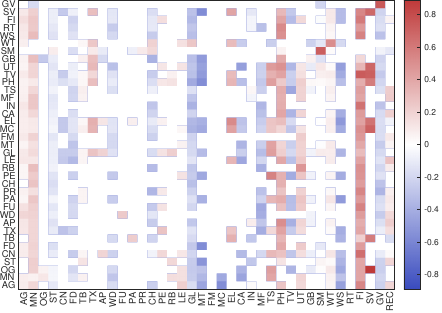}}\\
	\subfloat[$h=12$ ]
	{\includegraphics[width=0.4\textwidth]{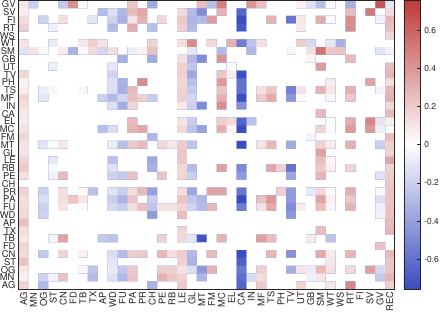}}
	\subfloat[$h=24$ ]
	{\includegraphics[width=0.4\textwidth]{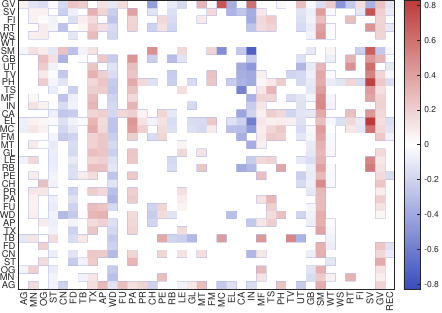}}	\\
	
    \caption{\textbf{Estimated High-Dimensional Local Projection Coefficients for Model 1 (Realized Volatility).}
  \textit{Notes:} Vertical axis shows the local projection response of the dependent variable at horizon \(t+h\); horizontal axis lists the predictor variable.  Responses are shown separately for horizons \(h\in\{1,6,12,24\}\).
    }
	\label{f_Emp_1_RV}
\end{sidewaysfigure}

\begin{sidewaysfigure}[htp!]
	\centering
	\subfloat[$h=1$ ]
	{\includegraphics[width=0.4\textwidth]{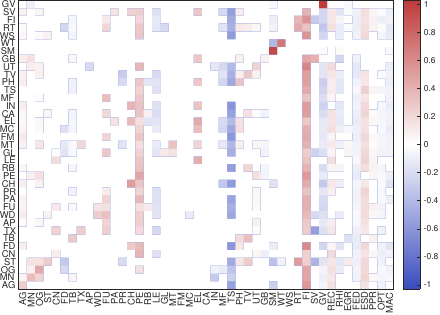}}
	\subfloat[$h=6$ ]
    {\includegraphics[width=0.4\textwidth]{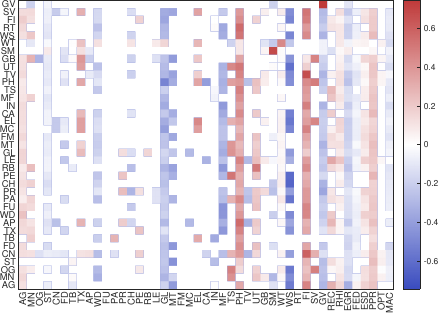}}\\
	\subfloat[$h=12$ ]
	{\includegraphics[width=0.4\textwidth]{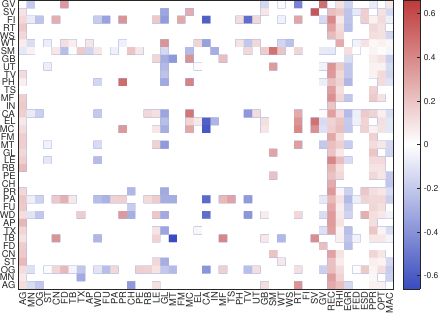}}
	\subfloat[$h=24$ ]
	{\includegraphics[width=0.4\textwidth]{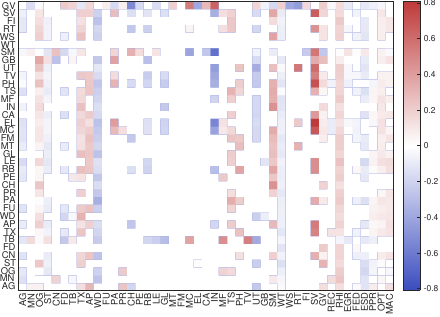}}	\\
	
    \caption{\textbf{Estimated High-Dimensional Local Projection Coefficients for Model 2 (Realized Volatility).}
  \textit{Notes:} Vertical axis shows the local projection response of the dependent variable at horizon \(t+h\); horizontal axis lists the predictor variable.  Responses are shown separately for horizons \(h\in\{1,6,12,24\}\).
    }
	\label{f_Emp_2_RV}
\end{sidewaysfigure}

\begin{sidewaysfigure}[htp!]
	\centering
	\subfloat[$h=1$ ]
	{\includegraphics[width=0.4\textwidth]{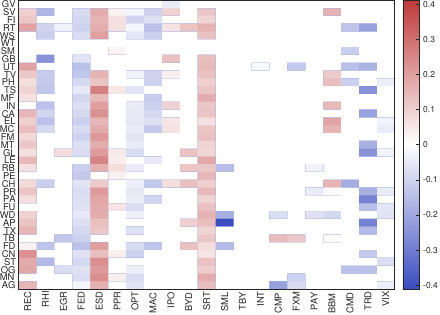}}
	\subfloat[$h=6$ ]
    {\includegraphics[width=0.4\textwidth]{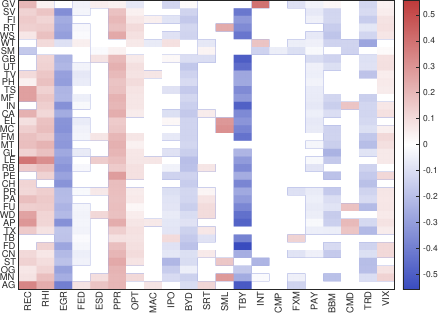}}\\
	\subfloat[$h=12$ ]
	{\includegraphics[width=0.4\textwidth]{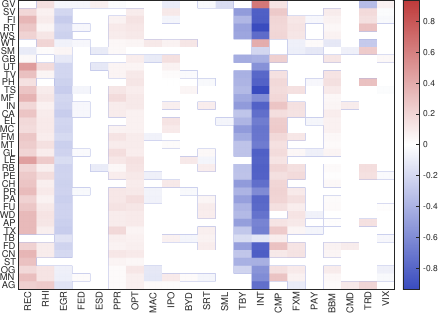}}
	\subfloat[$h=24$ ]
	{\includegraphics[width=0.4\textwidth]{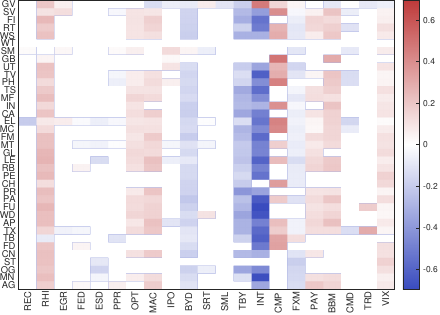}}	\\
	
    \caption{\textbf{Estimated High-Dimensional Local Projection Coefficients for Model 3 (Realized Volatility).}
  \textit{Notes:} Vertical axis shows the local projection response of the dependent variable at horizon \(t+h\); horizontal axis lists the predictor variable.  Responses are shown separately for horizons \(h\in\{1,6,12,24\}\). For brevity, estimates for spillovers within the volatility series are omitted; only the responses of volatility to news-attention indices are reported.
    }
	\label{f_Emp_3_RV}
\end{sidewaysfigure}

\subsection{Numerical Algorithm}\label{AP.A.5}

Before providing the numerical algorithm, we firstly simplify the node-wise LASSO numerically for the LP approach in the HD environment. 

Firstly, we introduce the following notations to facilitate the development. For $\forall \mathbf{b}_i$ which is a $(N^2p-1)\times 1 $ vector, let $\tilde{\mathbf{b}}_i$ be defined by placing a 1 in the \(i^{th}\) position and filling the remaining \(N^2p-1\) entries with the elements of \(- \mathbf{b}_i\) in their natural order.

We can then rewrite the objective function of \eqref{def.node} as follows:

\begin{eqnarray}\label{def.node3}
    && \frac{1}{T-h}\sum_{t=1}^{T-h}\|\mathbf{Z}_{i,t}^\top-\mathbf{Z}_{-i,t}^\top\mathbf{b}_i \|_2^2+2\tilde{\gamma}_i\|\mathbf{b}_i\|_1  \notag \\
    &=& \frac{1}{T-h}\sum_{t=1}^{T-h}\|(\mathbf{X}_{t}^\top\otimes \mathbf{I}_N)\tilde{\mathbf{b}}_i \|_2^2+2\tilde{\gamma}_i\|\mathbf{b}_i\|_1 \notag \\
    &=& \frac{1}{T-h}\sum_{t=1}^{T-h}\|\mathbf{X}_{t}^\top \tilde{\mathbf{b}}_i^*\|_2^2+\frac{1}{T-h}\sum_{t=1}^{T-h}\| (\mathbf{X}_{t}^\top\otimes \mathbf{I}_{N-1})^\top \tilde{\mathbf{b}}_i^+\|_2^2+2\tilde{\gamma}_i(\|\tilde{\mathbf{b}}_i^*\|_1+ \|\tilde{\mathbf{b}}_i^+\|_1 ),
\end{eqnarray}
where $\tilde{\mathbf{b}}_i^*$ is $Np\times 1$, and it includes the 1 and the elements of $\mathbf{b}_i$ corresponding to the non-zero elements of the $i^{th}$ row of $(\mathbf{X}_{t}^\top\otimes \mathbf{I}_N)$, and $\tilde{\mathbf{b}}_i^+$ includes the rest elements of  $\mathbf{b}_i $. Note that the interaction term of $\mathbf{X}_{t}^\top \tilde{\mathbf{b}}_i^*$ and $(\mathbf{X}_{t}^\top\otimes \mathbf{I}_{N-1})^\top \tilde{\mathbf{b}}_i^+$ is 0 due to the structure of Kronecker product, so it automatically disappears from the second equality of \eqref{def.node3}.
 
Note that 

\begin{eqnarray}
    \frac{1}{T-h}\sum_{t=1}^{T-h}\| (\mathbf{X}_{t}^\top\otimes \mathbf{I}_{N-1})^\top \tilde{\mathbf{b}}_i^+\|_2^2+2\tilde{\gamma}_i  \|\tilde{\mathbf{b}}_i^+\|_1 
\end{eqnarray}
reaches the minimum (i.e., 0) at $\tilde{\mathbf{b}}_i^+=\mathbf{0}$ due to the quadratic form. Therefore, minimizing \eqref{def.node3} is equivalent to minimizing 

\begin{eqnarray}\label{def.node4}
    \frac{1}{T-h}\sum_{t=1}^{T-h}\|\mathbf{X}_{t}^\top \tilde{\mathbf{b}}_i^*\|_2^2+2\tilde{\gamma}_i \|\tilde{\mathbf{b}}_i^*\|_1
\end{eqnarray}
with respect to $\tilde{\mathbf{b}}_i^*$ only. 

We can now conclude that the node-wise LASSO only needs to be implemented $Np$ times before constructing \eqref{def.node2} accordingly.

\paragraph{Algorithm} We perform the numerical implementation using the \textit{lasso} function\footnote{\url{https://www.mathworks.com/help/stats/lasso.html}} in MATLAB. For further technical details regarding the (adaptive) LASSO algorithm, we refer interested readers to \cite{mcilhagga2016penalized}. Based on the results in Section \ref{Sec2}, we set $J=5$ and $\gamma \asymp h^{1/5}(\log(N)/T)^{1/2}$. Following the discussion of Theorem \ref{TM.Main3}, we restrict our analysis to $h=1,2$ to minimize the constraints involved in selecting the optimal lag order; for simplicity, we then let $\xi \asymp  (\log(N)/T)^{1/2}$.

That said, the detailed algorithm is as follows:

\begin{enumerate}[leftmargin=48pt, parsep=2pt, topsep=2pt]
    \item[\it Step 1] Select the optimal lag using \eqref{def.phat} and \eqref{def.argmin0} by letting $h$ be a small integer such as $h=1$ or 2. 
    \item[\it Step 2] With $\hat{p}$ in hand, implement \eqref{def.argmin0} and \eqref{def.argmin}, and identify the positions of 0's.
    \item[\it Step 3] Conduct the node-wise LASSO estimation using \eqref{def.node} and \eqref{def.node2}, and set those positions identified in \textit{Step 2} to 0's.
\end{enumerate}

\newpage

\section{Proofs}\label{AP.B}

The following notations will be used throughout the proofs. 

Firstly, recall that we have defined $\mathbf{Z}_t\coloneqq \mathbf{X}_t\otimes \mathbf{I}_N$, $\pmb{\Sigma}_{\mathbf{B}}$, $\pmb{\phi}$ and $\phi_\ell$ in the main text. 

Secondly, let $\pmb{\beta}_0\coloneqq\VEC(\tilde{\mathbf{A}})$, and let $\pmb{\beta}_{0,{\mathscr{A}}}$ include the elements of $\pmb{\beta}_{0}$ corresponding to the non-zero elements of $\VEC(\mathbf{S}_{\mathscr{A}})$, wherein the $\ell^{th}$ element of $\pmb{\beta}_{0,\mathscr{A}}$ is denoted as $\beta_{0,\mathscr{A}, \ell}$. Similarly, let $\hat{\pmb{\beta}}_{\phi}\coloneqq \VEC(\hat{\tilde{\mathbf{a}}}_{\phi})$, and let $\hat{\pmb{\beta}}_{\phi,\mathscr{A}}$ include the elements of $\hat{\pmb{\beta}}_{\phi}$ corresponding to the non-zero elements of $\VEC(\mathbf{S}_{\mathscr{A}})$. 

In the same fashion, we define $\hat{\pmb{\beta}}_{\phi,\bar{\mathscr{A}}}$, $\pmb{\beta}_{0,\bar{\mathscr{A}}}$, $\pmb{\phi}_{\mathscr{A}}$, $\pmb{\phi}_{\bar{\mathscr{A}}}$, $\mathbf{Z}_{\mathscr{A},t}$, and $\mathbf{Z}_{\bar{\mathscr{A}},t}$. Let  $\phi_{\mathscr{A},\ell}$, $\phi_{\bar{\mathscr{A}},\ell}$, $\hat{\beta}_{\phi,\ell}$, and $\hat{\beta}_{\phi,\mathscr{A},\ell}$ be the $\ell^{th}$ elements of  $\pmb{\phi}_{\mathscr{A}}$, $\pmb{\phi}_{\bar{\mathscr{A}}}$, $\hat{\pmb{\beta}}_{\phi}$, and $\hat{\pmb{\beta}}_{\phi,\mathscr{A}}$ respectively. Finally, let $ \mathbf{g}_{\mathscr{A}} \coloneqq \sgn (\pmb{\beta}_{0,{\mathscr{A}}})\circ\pmb{\phi}_{\mathscr{A}}$,  $\pmb{\Sigma}_{\mathbf{Z}} \coloneqq \pmb{\Sigma}_{\mathbf{B}}\otimes \mathbf{I}_N$, and let $\pmb{\Sigma}_{\mathbf{Z},\mathscr{A}}$ be the covariance matrix associated with $\mathbf{Z}_{\mathscr{A},t}$.

\renewcommand{\thelemma}{B\arabic{lemma}}
\setcounter{lemma}{0} 
\renewcommand{\theproposition}{B\arabic{proposition}}
\setcounter{proposition}{0} 

\subsection{Proofs of the Main Results}

\begin{proof}[Proof of Proposition \ref{AP.Prop1}]
\item

Firstly, it is easy to see that \eqref{def.regression1} also admits a HDMA($\infty$) representation:

\begin{eqnarray}\label{def.xtilde} 
    \tilde{\mathbf{x}}_t = \sum_{\ell = 0}^{\infty}\mathbf{A}^\ell \pmb{\varepsilon}_{t-\ell}.
\end{eqnarray}

Secondly, we note that 

\begin{eqnarray*}
\left(\sum_{\ell=0}^L \mathbf{A}^\ell\right)  (\mathbf{I}-\mathbf{A})= (\mathbf{I}-\mathbf{A})\sum_{\ell=0}^L \mathbf{A}^\ell =\mathbf{I} - \mathbf{A}^{L+1}.
\end{eqnarray*}
Provided the conditions $\|\mathbf{A}\|_2<1$ and $\lambda_{\min}(\mathbf{I}-\mathbf{A})>0$, we can write further

\begin{eqnarray*}
\mathbf{A}^{L+1}&=&(\mathbf{I}-\mathbf{A})^{-1}(\mathbf{I}-\mathbf{A}) -\left(\sum_{\ell=0}^L \mathbf{A}^\ell\right)  (\mathbf{I}-\mathbf{A}) \notag\\
&=&\left[(\mathbf{I}-\mathbf{A})^{-1} -\left(\sum_{\ell=0}^L \mathbf{A}^\ell\right)\right] (\mathbf{I}-\mathbf{A}).
\end{eqnarray*}
Thus, as $L\to \infty$,

\begin{eqnarray*}
\left\|(\mathbf{I}-\mathbf{A})^{-1} - \sum_{\ell=0}^L \mathbf{A}^\ell \right\|_2\asymp \|\mathbf{A}^{L+1}\|_2\le \|\mathbf{A}\|_2^{L+1}\to 0.
\end{eqnarray*}

\medskip

We are not ready to proceed. Note that by Lemma \ref{LMB4},

\begin{eqnarray*}
\sum_{s=1}^t( \mathbf{x}_s - \tilde{\mathbf{x}}_s ) =(\bar{\mathbf{B}} - \bar{\mathbf{A}}) \sum_{s=1}^t\pmb{\varepsilon}_{s}-\tilde{\mathbf{B}}(L)\pmb{\varepsilon}_{t}+\tilde{\mathbf{B}}(L)\pmb{\varepsilon}_{0} + \tilde{\mathbf{A}}(L)\pmb{\varepsilon}_{t}-\tilde{\mathbf{A}}(L)\pmb{\varepsilon}_{0},
\end{eqnarray*}
where $\bar{\mathbf{A}} = \sum_{\ell=0}^\infty \mathbf{A}^\ell$ and $\bar{\mathbf{B}} = \sum_{\ell=0}^\infty \mathbf{B}_\ell$. Note further that for every $\bar{\mathbf{B}}$, we can find a $\mathbf{A}$ so that $\bar{\mathbf{B}} = (\mathbf{I}-\mathbf{A})^{-1}$. In other words, $\mathbf{A} = \mathbf{I}-\bar{\mathbf{B}}^{-1}$. Therefore, we can now write

\begin{eqnarray*}
\|\bar{\mathbf{B}} -\bar{\mathbf{A}}\|_2 &=& \left\|(\mathbf{I}-\mathbf{A})^{-1} - \sum_{\ell=0}^L \mathbf{A}^\ell \right\|_2+ \left\|\sum_{\ell=L+1}^\infty \mathbf{A}^\ell\right\|_2\notag \\
&\le &O(1)\left(\|\mathbf{A}\|_2^{L+1} + \sum_{\ell=L+1}^\infty \|\mathbf{A} \|_2^\ell\right)\notag\\
&=&O(1)\|\mathbf{A}\|_2^{L+1}\to 0.
\end{eqnarray*}

We then investigate

\begin{eqnarray*}
    \max_{t\in [T]}\|\tilde{\mathbf{B}}(L)\pmb{\varepsilon}_{t}\|_2 & \le & \|\tilde{\mathbf{B}}(L)\|_2\cdot \max_{t\in [T]}\|\pmb{\varepsilon}_{t}\|_2 \le  \|\tilde{\mathbf{B}}(L)\|_2\cdot \left(\max_{t\in [T]}\sum_{i=1}^N \varepsilon_{it}^2 \right)^{1/2}.
\end{eqnarray*}
Let $z_{it}\coloneqq \varepsilon_{it}^2-1$. Using Lemma \ref{LMB3}, we obtain that

\begin{eqnarray*}
    \Pr\left(\max_{t\in [T]}\left|\frac{1}{\sqrt{\sum_{i=1}^N E[z_{it}^2]}}\sum_{i=1}^Nz_{it}\right|\ge \epsilon\right)&\le &\sum_{t=1}^T \Pr\left(\left|\frac{1}{\sqrt{\sum_{i=1}^N E[z_{it}^2]}}\sum_{i=1}^Nz_{it}\right|\ge \epsilon\right)\notag \\
    &\le & 2T \exp(-\epsilon^2/4)= 2T \exp(\log(1/T^c))\notag \\
    &=&2T/T^c\to 0.
\end{eqnarray*}
where the first equality follows from letting $\epsilon^2/4 = \log(T^c)$  (i.e., $\epsilon =\sqrt{4c\log(T)}$) with $c>1$. In connection with the fact that $\sum_{i=1}^N E[z_{it}^2]\asymp N$, we conclude that 

\begin{eqnarray*}
    \max_{t\in [T]}  \sum_{i=1}^Nz_{it} =O_P (\sqrt{N \log(T)} ).
\end{eqnarray*}
Thus, $ \max_{t\in [T]}\|\tilde{\mathbf{B}}(L)\pmb{\varepsilon}_{t}\|_2\le O_P(1)\sqrt{N}\|\tilde{\mathbf{B}}(L)\|_2\cdot$. Consequently, in connection with Lemma \ref{LMB4}, we conclude that

\begin{eqnarray*}
    \max_{t\in [T]}\|\tilde{\mathbf{B}}(L)\pmb{\varepsilon}_{t}-\tilde{\mathbf{B}}(L)\pmb{\varepsilon}_{0} - \tilde{\mathbf{A}}(L)\pmb{\varepsilon}_{t}+\tilde{\mathbf{A}}(L)\pmb{\varepsilon}_{0}\|_2=O_P (\sqrt{N}).
\end{eqnarray*}

In what follows, we just need to focus on $(\bar{\mathbf{B}} - \bar{\mathbf{A}}) \sum_{s=1}^t\pmb{\varepsilon}_{s}$. Simple algebra shows that

\begin{eqnarray*}
    \max_{t\in [T]}\left\|(\bar{\mathbf{B}} - \bar{\mathbf{A}}) \sum_{s=1}^t\pmb{\varepsilon}_{s}\right\|_2 &\le & \|\bar{\mathbf{B}} - \bar{\mathbf{A}}\|_2 \cdot \sqrt{N}\max_{i\in [N]}\sum_{t=1}^T |\varepsilon_{it}| \notag \\
    &=&O(1)\|\mathbf{A}\|_2^{L+1} \cdot \sqrt{N}\max_{i\in [N]}\sum_{t=1}^T |\varepsilon_{it}|\notag \\
    & =&o_P(\sqrt{N}),
\end{eqnarray*}
where the last line follows from the fact that $\|\mathbf{A}\|_2^{L+1}$ shrinks to 0 exponentially.

Putting the above results together, the proof is completed.
\end{proof}

\medskip

\begin{proof}[Proof of Proposition \ref{PP.Main1}]
\item 

By \eqref{def.xth}, we now have

\begin{eqnarray*}
    \textbf{IR}(h,j) &=&E[\mathbf{x}_{t+h}\mid  \pmb{\varepsilon}_{t\mid j}(1), \pmb{\varepsilon}_{t-1},\ldots]  - E[\mathbf{x}_{t+h}\mid  \pmb{\varepsilon}_{t\mid j}(0), \pmb{\varepsilon}_{t-1},\ldots]\notag \\
    &=& \mathbf{A}_1  (E[\mathbf{x}_{t} \mid  \pmb{\varepsilon}_{t\mid j}(1), \pmb{\varepsilon}_{t-1},\ldots] - E[\mathbf{x}_{t}\mid  \pmb{\varepsilon}_{t\mid j}(0), \pmb{\varepsilon}_{t-1},\ldots] ).
\end{eqnarray*}
Since $\mathbf{B}_0 =\mathbf{I}$, we obtain that

\begin{eqnarray}\label{def.Ahj}
    \textbf{IR}(h,j)  = \mathbf{A}_{1,j},
\end{eqnarray}
where $\mathbf{A}_1=(\mathbf{A}_{1,1},\ldots, \mathbf{A}_{1,N})$. 

According to \eqref{def.Bh} and \eqref{def.Ahj}, we conclude  $\mathbf{A}_{1,j} = \mathbf{B}_{h,j}.$ The proof is completed.
\end{proof}

\medskip

\begin{proof}[Proof of Theorem \ref{TM.Main1}]
\item 

Without loss of generality, we consider $\sum_{t=1}^{T-h}u_{it,h} x_{jt}$ only for example. Before proceeding, we define a few notations which will be repeatedly used throughout the proof.  Let $\{\lambda_1,\ldots, \lambda_T\}$ be a positive sequence such that $\sum_{t=h}^T\lambda_t\le 1$, where the first $h-1$ periods are discarded in the summation for some obvious reason in view of the following development. We note that $c_\ell$ with $\ell$ being fixed positive integer always stands for positive constant, and may vary at each appearance. 

By DGP, we write

\begin{eqnarray*}
u_{it,h} x_{jt} =\mathbf{e}_i^\top\mathbf{g}(\pmb{\varepsilon}_{t+h},\ldots, \pmb{\varepsilon}_{t+1}) \sum_{\ell=0}^\infty \pmb{\beta}_{\ell, j}^\top \pmb{\varepsilon}_{t-\ell} \eqqcolon U (t,h),
\end{eqnarray*}
where we have suppressed $(i,j)$ when defining $U (t,h)$ for notational simplicity. Let further 

\begin{eqnarray*}
U (t,h,m)\coloneqq E[U (t,h)\mid \mathscr{F}_{t+h,t+h-m}],
\end{eqnarray*}
where $\mathscr{F}_{t,s}=\sigma(\pmb{\varepsilon}_t, \pmb{\varepsilon}_{t-1},\ldots, \pmb{\varepsilon}_s)$ for $t\ge s$. By the definition, $U(t,h,m)$ and $U(s,h,m)$ are independent of each other as long as $|t-s|>m$. 

For $t\in[T-h]$, we decompose $U(t,h)$ as follows:

\begin{eqnarray*}
U(t,h) = U(t,h)- U(t,h,T)+\sum_{m=1}^T[U (t,h,m)-U(t,h,m-1)]+U(t,h,0),
\end{eqnarray*}
wherein, for $m<h$,

\begin{eqnarray*}
U (t,h,m) &=& E\big[\mathbf{e}_i^\top\mathbf{g}(\pmb{\varepsilon}_{t+h},\ldots, \pmb{\varepsilon}_{t+1})  \sum_{\ell=0}^\infty\pmb{\beta}_{\ell, j}^\top \pmb{\varepsilon}_{t-\ell} \mid \mathscr{F}_{t+h,t+h-m} \big]  =0,
\end{eqnarray*}
and for $m\ge h$,

\begin{eqnarray}\label{U_mom1}
&&\|U (t,h,m)-U (t,h,m-1)\|_J \notag\\
&=&\|\mathbf{e}_i^\top\mathbf{g}(\pmb{\varepsilon}_{t+h},\ldots, \pmb{\varepsilon}_{t+1})  \pmb{\beta}_{m-h, j}^\top \pmb{\varepsilon}_{t-(m-h)} \|_J \notag \\
&\le &O(1) \| \pmb{\beta}_{m-h, j}^\top \pmb{\varepsilon}_{t-(m-h)} \|_J\notag \\
&\le & O(1) \big\| \big\{\sum_{i=1}^N  (b_{m-h,ji}\varepsilon_{i,t-(m-h)} ) ^2\big\}^{\frac{1}{2}}\big\|_J\notag \\
&\le &O(1)\big\{\sum_{i=1}^N   (E | b_{m-h,ji}\varepsilon_{i,t-(m-h)}|^J )^{\frac{2}{J}}\big\}^{\frac{1}{2}} \notag \\
&=&O(1) \|\pmb{\beta}_{m-h, j} \|_2,
\end{eqnarray}
where the first inequality follows from Assumption \ref{AS2}.1, the second inequality follows from Burkholder's inequality, the third inequality follows from Minkowski inequality, and the last line follows from the existence of $E|\varepsilon_{it}|^J$ as mentioned in Assumption \ref{AS1}. It is noteworthy that the right hand side of \eqref{U_mom1} is independent of $t$.

\medskip

Based on the above setting, for $S(t)\coloneqq \sum_{s=1}^{t}u_{is,h} x_{js}$  with $t\in[T-h]$, we can write

\begin{eqnarray*}
S(t) &=& \sum_{s=1}^{t}[U (s,h)- U (s,h,T)] +\sum_{s=1}^{t}\sum_{m=h}^{T}[U (s,h,m)-U (s,h,m-1)] \notag \\
&\eqqcolon& S_1(t)+\sum_{m=h}^TS_{2,m}(t) \eqqcolon S_1(t)+S_2(t).
\end{eqnarray*}
Again, for notational simplicity, we have suppressed $(i,j,h)$ in $S(t)$, $S_1(t)$, and $S_2(t)$. In what follows, we consider the following probability:

\begin{eqnarray*}
\Pr\big(\max_{t\in [T-h]} | S(t) |\ge 4\delta \big) \le \Pr\big(\max_{t\in [T-h]} \left| S_1(t)\right|\ge  \delta\big)+\Pr\big(\max_{t\in [T-h]} \left| S_2(t)\right|\ge  3\delta\big).
\end{eqnarray*}

Firstly, we consider $S_2(t)$. To proceed, we let  $b =\lfloor \frac{T}{m} \rfloor+1$, and for $l\in [b]$ let

\begin{eqnarray*}
\nu_{m,l}\coloneqq \sum_{s=1+(l-1)m}^{lm \wedge T} [U (s,h,m)-U (s,h,m-1)].
\end{eqnarray*}
By the construction, we obtain that

\begin{eqnarray}\label{nu_ml}
\|\nu_{m,l}\|_J  &\le & O(1) \big\| \big\{\sum_{s=1+(l-1)m}^{lm \wedge T}  [U (s,h,m)-U (s,h,m-1)] ^2\big\}^{\frac{1}{2}}\big\|_J\notag \\
&\le &O(1)\big\{\sum_{s=1+(l-1)m}^{lm \wedge T}   (E | U (s,h,m)-U (s,h,m-1)|^J )^{\frac{2}{J}}\big\}^{\frac{1}{2}} \notag \\
&\le &O(\sqrt{m } \|\pmb{\beta}_{m-h, j} \|_2),
\end{eqnarray}
where the first inequality follows from Burkholder's inequality,  the second inequality follows from Minkowski inequality, and the third inequality follows from \eqref{U_mom1}. 

We then denote that 

\begin{eqnarray*}
V_{1,s,m}\coloneqq \sum_{l=1}^s\frac{1-(-1)^l}{2}\nu_{m,l}  \quad\text{and}\quad V_{2,s,m}\coloneqq \sum_{l=1}^s\frac{1+(-1)^l}{2}\nu_{m,l},
\end{eqnarray*}
where $V_{1,s,m}$ and $V_{2,s,m}$ include odd and even indexed terms respectively. Additionally, note that 

\begin{eqnarray*}
S_{2,m}(t) =\sum_{l=1}^{\lfloor t/m \rfloor }\nu_{m,l} + \sum_{s=\lfloor t/m \rfloor m + 1}^{t} [U (s,h,m)-U (s,h,m-1)],
\end{eqnarray*}
where 

\begin{eqnarray*}
\sum_{s=\lfloor t/m \rfloor m + 1}^{t} [U (s,h,m)-U (s,h,m-1)]=_D \sum_{s=  1}^{t - \lfloor t/m \rfloor m} [U (s,h,m)-U (s,h,m-1)].
\end{eqnarray*}
Thus, we further write

\begin{eqnarray}\label{prob.decom}
&&\Pr\left(\max_{t\in [T-h]} \left| S_2(t)\right|\ge  3\delta\right)\notag \\
&\le &\sum_{m=h}^T\Pr\left(\max_{t\in [T-h]} \left| S_{2,m}(t)\right|\ge 3\lambda_m \delta\right)\notag \\
&\le &\sum_{m=h}^T\Big[ \Pr\left(\max_{s\in [b]}|V_{1,s,m}| \ge \lambda_m \delta\right)+ \Pr\left(\max_{s\in [b]}|V_{2,s,m}| \ge \lambda_m \delta\right)\notag \\
&&+ \Pr\left(\max_{1\le t\le m}\left|\sum_{s=1}^t [U (s,h,m)-U (s,h,m-1)]\right|\ge \lambda_m \delta\right)\Big].
\end{eqnarray}

\medskip

We then study the terms on the right hand side of \eqref{prob.decom} one by one. Since $\nu_{1,l}, \nu_{3,l},\ldots$ are independent, we invoke Nagaev inequality of Lemma \ref{LMB1} and obtain that

\begin{eqnarray*}
&&\Pr\left(\max_{s\in [b]}|V_{1,s,m}| \ge \lambda_m \delta\right)\notag \\
&\le &c_1 \frac{\sum_{l=1}^{b}\|\nu_{m,l}\|_J^J}{(\lambda_m \delta)^J} + 2\exp\left(  \frac{- c_2 (\lambda_m \delta)^2}{\sum_{l=1}^{b}\|\nu_{m,l}\|_2^2}\right)\notag \\
&\le &c_1 \frac{b m^{\frac{J}{2}} \|\pmb{\beta}_{m-h, j} \|_2^J}{(\lambda_m \delta)^J} + 2\exp\left(  \frac{- c_2 (\lambda_m \delta)^2}{b m \|\pmb{\beta}_{m-h, j} \|_2^2}\right) \notag \\
&= &c_1 \frac{T m^{\frac{J}{2}-1} \|\pmb{\beta}_{m-h, j} \|_2^J}{(\lambda_m \delta)^J} + 2\exp\left(  \frac{- c_2 (\lambda_m \delta)^2}{T \|\pmb{\beta}_{m-h, j} \|_2^2}\right),
\end{eqnarray*}
where the second inequality follows from \eqref{nu_ml}. Similarly, we obtain that 

\begin{eqnarray*}
\Pr\left(\max_{s\in [b]}|V_{2,s,m}| \ge \lambda_m \delta\right)\le  c_1 \frac{T m^{\frac{J}{2}-1} \|\pmb{\beta}_{m-h, j} \|_2^J}{(\lambda_m \delta)^J} + 2\exp\left(  \frac{- c_2 (\lambda_m \delta)^2}{T \|\pmb{\beta}_{m-h, j} \|_2^2}\right).
\end{eqnarray*}

In order to study the third term on the right hand side of \eqref{prob.decom}, we note that  

\begin{eqnarray*}
    \{U (T-h-t,h,m)-U (T-h-t,h,m-1)\}_{t=0}^{T-h-1}
\end{eqnarray*}
are martingale differences with respect to $\mathcal{F}_{t}\coloneqq\sigma(\pmb{\varepsilon}_{T-t-m}, \pmb{\varepsilon}_{T-t-m+1},\ldots)$, thus,

\begin{eqnarray*}
\left\{\left| \sum_{s=0}^{t} [U (T-h-s,h,m)-U (T-h-s,h,m-1)]\right| \right\}_{t=0}^{T-h-1}
\end{eqnarray*}
a positive submartingale with respect to $\mathcal{F}_{t}$. Thus, by Doob $L_p$ maximal inequality, we write

\begin{eqnarray}\label{rate.maxt}
&&E\left[\max_{1\le t\le m}\left|\sum_{s=1}^t [U (s,h,m)-U (s,h,m-1)]\right|^J \right]\notag \\
&=&E\left[\max_{1\le t\le m}\left|\sum_{s=t}^m [U (s,h,m)-U (s,h,m-1)]\right|^J \right]\notag \\
&\le &\left(\frac{J}{J-1}\right)^J E\left|\sum_{s=1}^m [U (s,h,m)-U (s,h,m-1)]\right|^J\notag \\
&\le &O(1)m^{\frac{J}{2}} \|\pmb{\beta}_{m-h, j} \|_2^J,
\end{eqnarray}
where the last step follows from \eqref{nu_ml}. Thus, 

\begin{eqnarray*}
\Pr\left(\max_{1\le t\le m}\left|\sum_{s=1}^t [U (s,h,m)-U (s,h,m-1)]\right|\ge \lambda_m \delta\right)\le  c_1 \frac{m^{\frac{J}{2}} \|\pmb{\beta}_{m-h, j} \|_2^J}{(\lambda_m \delta)^J}.
\end{eqnarray*}
 
Up to this point, we have investigated all the terms on the right hand side of \eqref{prob.decom}. In order to put everything together, we let $\lambda_m = \mu_m /\bar{\mu}$ in which $\bar{\mu}=\sum_{m=h}^{\infty}\mu_m$ and $\mu_m = (m^{\frac{J}{2}-1} \|\pmb{\beta}_{m-h, j} \|_2^J)^{\frac{1}{J+1}}$, and write further

\begin{eqnarray*}
&&\Pr\left(\max_{t\in [T-h]} \left| S_2(t)\right|\ge  3\delta\right) \notag \\
&\le &c_1 \frac{T}{\delta^J}\sum_{m=h}^{T}\frac{ m^{\frac{J}{2}-1} \|\pmb{\beta}_{m-h, j} \|_2^J}{\lambda_m ^J} +\sum_{m=h}^{T} 4\exp\left(  \frac{- c_2 (\lambda_m \delta)^2}{T \|\pmb{\beta}_{m-h, j} \|_2^2}\right) \notag \\
&= &c_1 \frac{T}{\delta^J} \bar{\mu}^{J+1}  +\sum_{m=h}^{T} 4\exp\left(  \frac{- c_2 (\mu_m \delta)^2}{T\bar{\mu}^2 \|\pmb{\beta}_{m-h, j} \|_2^2}\right),
\end{eqnarray*}
where the last step follows from the direct calculation. It is noteworthy that we have suppressed the index $h$ in $\bar{\mu}$ for notational simplicity.

\medskip

We now consider $S_1(t)$, and write

\begin{eqnarray*}
\big\| \max_{t\in [T-h]} |S_1(t)|\big\|_J &=&\left\| \max_{t\in [T-h]} \left|\sum_{s=1}^{t}  \sum_{\ell=T-h+1}^\infty \mathbf{e}_i^\top\mathbf{u}_{s,h}\pmb{\beta}_{\ell, j}^\top \pmb{\varepsilon}_{s-\ell}\right| \right\|_J\notag \\
&\le &O(1)\left\| \sum_{s=1}^{T-h}  \sum_{\ell=T-h+1}^\infty \mathbf{e}_i^\top\mathbf{u}_{s,h}\pmb{\beta}_{\ell, j}^\top \pmb{\varepsilon}_{s-\ell} \right\|_J\notag \\
&\le &O(1)\sqrt{T-h}\left\| \sum_{\ell=T-h+1}^\infty \mathbf{e}_i^\top\mathbf{u}_{s,h}\pmb{\beta}_{\ell, j}^\top \pmb{\varepsilon}_{s-\ell}\right\|_J\notag \\
&\le &O(1)\sqrt{T-h}\sum_{\ell=T-h+1}^\infty  \| \pmb{\beta}_{\ell, j}^\top \pmb{\varepsilon}_{s-\ell} \|_J\notag \\
&\le &O(1) \sqrt{T-h}\sum_{\ell=T-h+1}^\infty\|\pmb{\beta}_{\ell, j}\|_2,
\end{eqnarray*}
where the first inequality follows from a development similar to \eqref{rate.maxt}, the second inequality follows from a development similar to  \eqref{nu_ml}, and the last line follows from the development of \eqref{U_mom1}. Additionally, we note that

\begin{eqnarray}\label{rate.betaT}
\left(\sum_{\ell=T-h+1}^\infty\|\pmb{\beta}_{\ell, j}\|_2\right)^J &\le&  \left(\sum_{\ell=T-h+1}^\infty\|\pmb{\beta}_{\ell, j}\|_2^{J/(J+1)}\right)^{J+1}\notag \\
&= &\left(\sum_{\ell=T-h+1}^\infty\frac{(\ell+h)^{(J/2-1)/(J+1)}}{(\ell+h)^{(J/2-1)/(J+1)}}\|\pmb{\beta}_{\ell, j}\|_2^{J/(J+1)}\right)^{J+1}\notag \\
&\le  &\frac{1}{(T+1)^{J/2-1}}\left(\sum_{\ell=T-h+1}^\infty\mu_{\ell+h} \right)^{J+1}\notag \\
&=  &\frac{1}{(T+1)^{J/2-1}}\left(\sum_{\ell=T+1}^\infty\mu_\ell \right)^{J+1},
\end{eqnarray}
where the first inequality follows from the norm monotonicity. Thus, we write

\begin{eqnarray*}
\Pr\left(\max_{t\in [T-h]} |S_1(t)|\ge \delta\right)&\le &c_1 \frac{(\sqrt{T-h}\sum_{\ell=T-h+1}^\infty\|\pmb{\beta}_{\ell, j}\|_2)^J}{\delta^J}\le c_1\frac{T}{\delta^J} \left(\sum_{\ell=T+1}^\infty\mu_\ell \right)^{J+1},
\end{eqnarray*}
where the second inequality follows from \eqref{rate.betaT}.

Based on the development for $S_1(t)$ and $S_2(t)$, we can conclude that 

\begin{eqnarray*}
&&\Pr\left(\max_{t\in [T-h]} \left| S(t)\right|\ge  \delta\right) \le c_1 \frac{T}{\delta^J}  \bar{\mu}^{J+1}+\sum_{m=h}^{T} 4\exp\left(  \frac{- c_2 (\mu_m \delta)^2}{T\bar{\mu}^2 \|\pmb{\beta}_{m-h, j} \|_2^2}\right).
\end{eqnarray*}
Letting  $\delta =\sqrt{T}\bar{\mu}^{1+1/J} \zeta$ with $\zeta>0$ yields that

\begin{eqnarray*}
\frac{  (\mu_m \delta)^2}{T\bar{\mu}^2 \|\pmb{\beta}_{m-h, j} \|_2^2} = \frac{  (\mu_m \sqrt{T}\bar{\mu}^{1+1/J} \zeta)^2}{T\bar{\mu}^2 \|\pmb{\beta}_{m-h, j} \|_2^2} =\zeta^2 \bar{\mu}^{\frac{2}{J}}\cdot \frac{m^{\frac{J-2}{J+1}}  }{\|\pmb{\beta}_{m-h, j} \|_2^{\frac{2}{J+1}}}=\zeta^2 m^{1-\frac{2}{J}} \frac{\bar{\mu}^{\frac{2}{J}}}{ \mu_m^{\frac{2}{J}}} \ge \zeta^2 m^{1-\frac{2}{J}}.
\end{eqnarray*}
Thus,

\begin{eqnarray*}
\sum_{m=h}^{T} \exp\left(  \frac{- c_2 (\mu_m \delta)^2}{T\bar{\mu}^2 \|\pmb{\beta}_{m-h, j} \|_2^2}\right)&\le & \sum_{m=1}^{\infty}\exp (-\zeta^2 m^{1-\frac{2}{J}} )\notag \\
&= &\frac{\sum_{m=1}^{\infty}\exp (-\zeta^2  m^{1-\frac{2}{J}}  )\exp(\zeta^2)}{\exp(\zeta^2)}\notag \\
&\le &\frac{\sum_{m=1}^{\infty}\exp (-m^{1-\frac{2}{J}}  )\exp(1)}{\exp(\zeta^2)}\notag \\
&=& c\exp(-c^*\zeta^2),
\end{eqnarray*}
where $c$ and $c^*$ are two constants, and the second inequality follows from the fact that, for any $s>0$,

\begin{eqnarray*}
\sup_{\zeta\ge 1} \sum_{m=1}^{\infty}\exp (-m^s \zeta^2)\exp(\zeta^2) = \sum_{m=1}^{\infty}\exp (-m^s )\exp(1),
\end{eqnarray*}
which can be easily verified by showing that $\sum_{m=1}^{\infty}\exp (-m^s \zeta^2)\exp(\zeta^2)$ is monotonically decreasing on $\zeta\in [1,\infty)$. Finally, in view of the fact that $\zeta=\delta /(\sqrt{T}\bar{\mu}^{1+1/J})$, we obtain that 

\begin{eqnarray*}
    \Pr\left(\max_{t\in [T-h]}\left|\sum_{s=1}^{t}u_{is,h} x_{j,s-\ell}\right|\ge \delta\right) \le c_1 \frac{T}{\delta^J} \mu_h^{J+1}+c_2\exp\left(-\frac{c_3\delta^2}{T \mu_h^{2+2/J}}\right),
\end{eqnarray*}
where $\displaystyle\mu_h\coloneqq \max_{j}\mu_{h,j}$ and $\mu_{h,j}\coloneqq \sum_{\ell=0}^\infty  [(\ell+h)^{\frac{J}{2}-1} \|\pmb{\beta}_{\ell, j} \|_2^J]^{\frac{1}{J+1}}$.

Note further that 

\begin{eqnarray*}
    &&\sum_{\ell=0}^\infty  [(\ell+h)^{\frac{J}{2}-1} \|\pmb{\beta}_{\ell, j} \|_2^J]^{\frac{1}{J+1}} \notag \\
    &=&\sum_{\ell=0}^h  [(\ell+h)^{\frac{J}{2}-1} \|\pmb{\beta}_{\ell, j} \|_2^J]^{\frac{1}{J+1}}+\sum_{\ell=h+1}^{\infty}  [(\ell+h)^{\frac{J}{2}-1} \|\pmb{\beta}_{\ell, j} \|_2^J]^{\frac{1}{J+1}}\notag \\
    &\le &O(1)h^{\frac{J-2}{2(J+1)}}\sum_{\ell=0}^h \|\pmb{\beta}_{\ell, j} \|_2^{\frac{J}{J+1}}+\sum_{\ell=h+1}^{\infty}  [(2\ell)^{\frac{J}{2}-1} \|\pmb{\beta}_{\ell, j} \|_2^J]^{\frac{1}{J+1}}\notag \\
    &\le & O(1)h^{\frac{J-2}{2(J+1)}}\sum_{\ell=0}^{\infty}  (\ell^{\frac{J}{2}-1} \|\pmb{\beta}_{\ell, j} \|_2^J)^{\frac{1}{J+1}}\notag \\
    &=&O(h^{\frac{J-2}{2(J+1)}}),
\end{eqnarray*}
where the last line follows from Lemma \ref{LMB5}. We conclude that $\mu_h=O(h^{\frac{J-2}{2(J+1)}})$. The proof is now completed.
\end{proof}

\medskip

\begin{proof}[Proof of Theorem \ref{TM.Main2}]
\item 
 
(1).  Firstly, by Theorem \ref{TM.Main1}, Lemma \ref{LMB6} and the condition $\frac{N^2 }{T^{J-1}\log N}$ of Assumption \ref{AS3}, the following results hold with probability approaching 1:

\begin{eqnarray}\label{basis.ineq1}
    &&\left\|\frac{1}{T-h}\sum_{t=1}^{T-h}\mathbf{u}_{t,h}^\top (\mathbf{X}_t^\top \otimes \mathbf{I}_N)\right\|_\infty \le\frac{\gamma}{4},\notag \\
    &&\left\|\frac{1}{T-h} \sum_{t=1}^{T-h} \mathbf{X}_t\mathbf{X}_t^\top - \pmb{\Sigma}_{\mathbf{B}} \right\|_{\max}\le d_{\pmb{\beta}}^2\mu_{h}^{-1-1/J}\gamma,
\end{eqnarray}
where $\gamma \asymp \frac{\mu_{h}^{1+1/J}\sqrt{\log (N)}}{\sqrt{T}}$.

\medskip

Let $Q(\tilde{\mathbf{a}})=Q_0(\tilde{\mathbf{a}})+\gamma\|\VEC(\tilde{\mathbf{a}})\|_1$. Since $Q(\hat{\tilde{\mathbf{a}}}) \le Q(\tilde{\mathbf{A}})$ by \eqref{def.argmin}, we obtain that

\begin{eqnarray}\label{eq.lasso0}
    &&\frac{1}{T-h}\sum_{t=1}^{T-h}\|\mathbf{x}_{t+h}-(\mathbf{X}_t^\top \otimes \mathbf{I}_N) \VEC(\hat{\tilde{\mathbf{a}}})\|_2^2+\gamma\|\VEC(\hat{\tilde{\mathbf{a}}})\|_1 \notag \\
    &\le &\frac{1}{T-h}\sum_{t=1}^{T-h}\|\mathbf{x}_{t+h}-(\mathbf{X}_t^\top \otimes \mathbf{I}_N) \VEC(\tilde{\mathbf{A}})\|_2^2+\gamma\|\VEC(\tilde{\mathbf{A}})\|_1\notag \\
    &= &\frac{1}{T-h}\sum_{t=1}^{T-h}\|\mathbf{u}_{t,h}\|_2^2+\gamma\|\VEC(\tilde{\mathbf{A}})\|_1.
\end{eqnarray}
 
Due to the first inequality of \eqref{basis.ineq1},

\begin{eqnarray*}
    &&\frac{1}{T-h}\left|\sum_{t=1}^{T-h}\mathbf{u}_{t,h}^\top (\mathbf{X}_t^\top \otimes \mathbf{I}_N) \VEC(\tilde{\mathbf{A}}-\hat{\tilde{\mathbf{a}}})\right| \notag \\
    &\le &\frac{1}{T-h} \left\|\sum_{t=1}^{T-h}\mathbf{u}_{t,h}^\top (\mathbf{X}_t^\top \otimes \mathbf{I}_N)\right\|_\infty \cdot \|\VEC(\tilde{\mathbf{A}}-\hat{\tilde{\mathbf{a}}})\|_1\le\frac{\gamma}{4} \|\VEC(\tilde{\mathbf{A}}-\hat{\tilde{\mathbf{a}}})\|_1,
\end{eqnarray*}
so \eqref{eq.lasso0} reduces to 

\begin{eqnarray}\label{eq.lasso1}
    &&\frac{1}{T-h} \sum_{t=1}^{T-h}\|(\mathbf{X}_t^\top \otimes \mathbf{I}_N) \VEC(\tilde{\mathbf{A}}-\hat{\tilde{\mathbf{a}}})\|_2^2+\gamma\|\VEC(\hat{\tilde{\mathbf{a}}})\|_1\notag \\
    &\le &\frac{\gamma}{2} \|\VEC(\tilde{\mathbf{A}}-\hat{\tilde{\mathbf{a}}})\|_1  +\gamma\|\VEC(\tilde{\mathbf{A}})\|_1.
\end{eqnarray}

Note that 

\begin{eqnarray*}
    &&\|\VEC(\tilde{\mathbf{A}}-\hat{\tilde{\mathbf{a}}})\|_1+\|\VEC(\tilde{\mathbf{A}})\|_1 -\|\VEC(\hat{\tilde{\mathbf{a}}})\|_1\notag \\
    &=&\|\VEC(\mathbf{S}_{\mathscr{A}}\circ(\tilde{\mathbf{A}}-\hat{\tilde{\mathbf{a}}}))\|_1+\|\VEC(\mathbf{S}_{\mathscr{A}}\circ\tilde{\mathbf{A}})\|_1-\|\VEC(\mathbf{S}_{\mathscr{A}}\circ\hat{\tilde{\mathbf{a}}})\|_1\notag \\
    &\le &2\|\VEC(\mathbf{S}_{\mathscr{A}}\circ(\tilde{\mathbf{A}}-\hat{\tilde{\mathbf{a}}}))\|_1,
\end{eqnarray*}
so we can rewrite \eqref{eq.lasso1} as 

\begin{eqnarray}\label{eq.lasso2}
    &&\frac{1}{T-h} \sum_{t=1}^{T-h}\|(\mathbf{X}_t^\top \otimes \mathbf{I}_N) \VEC(\tilde{\mathbf{A}}-\hat{\tilde{\mathbf{a}}})\|_2^2 +\frac{\gamma}{2} \|\VEC(\tilde{\mathbf{A}}-\hat{\tilde{\mathbf{a}}})\|_1 \notag \\
    &\le &\gamma(\|\VEC(\tilde{\mathbf{A}}-\hat{\tilde{\mathbf{a}}})\|_1+\|\VEC(\tilde{\mathbf{A}})\|_1-\|\VEC(\hat{\tilde{\mathbf{a}}})\|_1)\notag \\
    &\le &2\gamma\|\VEC(\mathbf{S}_{\mathscr{A}}\circ(\tilde{\mathbf{A}}-\hat{\tilde{\mathbf{a}}}))\|_1,
\end{eqnarray}
which further yields that

\begin{eqnarray}\label{eq.lasso2_1}
    &&\frac{1}{T-h} \sum_{t=1}^{T-h}\|(\mathbf{X}_t^\top \otimes \mathbf{I}_N) \VEC(\tilde{\mathbf{A}}-\hat{\tilde{\mathbf{a}}})\|_2^2 +\frac{\gamma}{2} \|\VEC(\mathbf{S}_{\bar{\mathscr{A}}} \circ(\tilde{\mathbf{A}}-\hat{\tilde{\mathbf{a}}}))\|_1 \notag \\
    &\le & \frac{3}{2}\gamma\|\VEC(\mathbf{S}_{\mathscr{A}}\circ(\tilde{\mathbf{A}}-\hat{\tilde{\mathbf{a}}}))\|_1.
\end{eqnarray}

\medskip

We now focus on $\frac{1}{T-h} \sum_{t=1}^{T-h}\|(\mathbf{X}_t^\top \otimes \mathbf{I}_N) \VEC(\tilde{\mathbf{A}}-\hat{\tilde{\mathbf{a}}})\|_2^2$ of \eqref{eq.lasso2_1}. For $\forall \tilde{\mathbf{a}}\in \mathbb{A}(\mathbf{S}_{\mathscr{A}} )$, Assumption \ref{AS3} yields that 

\begin{eqnarray}\label{eq.lasso3}
    \|\VEC(\tilde{\mathbf{a}}) \|_1 &=& \|\VEC(\mathbf{S}_{\mathscr{A}}\circ\tilde{\mathbf{a}}) \|_1 +\|\VEC(\mathbf{S}_{\bar{\mathscr{A}}}\circ\tilde{\mathbf{a}}) \|_1 \notag \\
    &\le & 4\|\VEC(\mathbf{S}_{\mathscr{A}}\circ\tilde{\mathbf{a}}) \|_1\le 4\sqrt{d_{\mathscr{A}}}\|\VEC(\mathbf{S}_{\mathscr{A}}\circ\tilde{\mathbf{a}}) \|_2.
\end{eqnarray}
Thus, with probability approaching 1,

\begin{eqnarray*}
    &&\frac{1}{T-h} \sum_{t=1}^{T-h}\|(\mathbf{X}_t^\top \otimes \mathbf{I}_N) \VEC(\tilde{\mathbf{A}}-\hat{\tilde{\mathbf{a}}})\|_2^2\notag \\
    & \ge &\VEC(\tilde{\mathbf{A}}-\hat{\tilde{\mathbf{a}}})^\top  (\pmb{\Sigma}_{\mathbf{B}}\otimes \mathbf{I}_N ) \VEC(\tilde{\mathbf{A}}-\hat{\tilde{\mathbf{a}}}) \notag \\
    &&-  \left\|\frac{1}{T-h} \sum_{t=1}^{T-h} \mathbf{X}_t\mathbf{X}_t^\top - \pmb{\Sigma}_{\mathbf{B}} \right\|_{\max}\cdot \|\VEC(\tilde{\mathbf{A}}-\hat{\tilde{\mathbf{a}}}) \|_1^2\notag \\
    &\ge &\VEC(\tilde{\mathbf{A}}-\hat{\tilde{\mathbf{a}}})^\top  (\pmb{\Sigma}_{\mathbf{B}} \otimes \mathbf{I}_N ) \VEC(\tilde{\mathbf{A}}-\hat{\tilde{\mathbf{a}}}) \notag \\
    &&-  4 d_{\mathscr{A}}\left\|\frac{1}{T-h} \sum_{t=1}^{T-h} \mathbf{X}_t\mathbf{X}_t^\top - \pmb{\Sigma}_{\mathbf{B}} \right\|_{\max}\cdot\|\VEC(\tilde{\mathbf{A}}-\hat{\tilde{\mathbf{a}}}) \|_2^2\notag \\
    &\ge &\frac{1}{2}\VEC(\tilde{\mathbf{A}}-\hat{\tilde{\mathbf{a}}})^\top  (\pmb{\Sigma}_{\mathbf{B}} \otimes \mathbf{I}_N ) \VEC(\tilde{\mathbf{A}}-\hat{\tilde{\mathbf{a}}}) ,
\end{eqnarray*}
where the second inequality follows from \eqref{eq.lasso3}, and the third inequality follows from the second inequality of \eqref{basis.ineq1} and $d_{\mathscr{A}} d_{\pmb{\beta}}^2\frac{\sqrt{\log N}}{\sqrt{T}}\to 0$.

Thus, by Assumption \ref{AS3},

\begin{eqnarray}\label{eq.lasso4}
    \frac{1}{2}\alpha \|\VEC(\tilde{\mathbf{A}}-\hat{\tilde{\mathbf{a}}}) \|_2^2&\le & \frac{3}{2}\gamma \|\VEC(\mathbf{S}_{\mathscr{A}}\circ(\tilde{\mathbf{A}}-\hat{\tilde{\mathbf{a}}}))\|_1\notag \\
    &\le & \frac{3}{2}\gamma  \sqrt{d_{\mathscr{A}}}\|\VEC(\mathbf{S}_{\mathscr{A}}\circ (\tilde{\mathbf{A}}-\hat{\tilde{\mathbf{a}}})) \|_2\notag \\
    &\le & \frac{3}{2}\gamma  \sqrt{d_{\mathscr{A}}}\|\VEC( \tilde{\mathbf{A}}-\hat{\tilde{\mathbf{a}}}) \|_2
\end{eqnarray}
where the first inequality follows from \eqref{eq.lasso2_1}. We can conclude that with probability approaching 1,

\begin{eqnarray*}
    \|\VEC(\tilde{\mathbf{A}}-\hat{\tilde{\mathbf{a}}}) \|_2\le \frac{3\gamma \sqrt{d_{\mathscr{A}}}}{\alpha}.
\end{eqnarray*}
Additionally, we can obtain that 

\begin{eqnarray*}
    \|\VEC( \tilde{\mathbf{A}}-\hat{\tilde{\mathbf{a}}})\|_1\le 4 \sqrt{d_{\mathscr{A}}}\|\VEC( \mathbf{S}_{\mathscr{A}}\circ(\tilde{\mathbf{A}}-\hat{\tilde{\mathbf{a}}}))\|_2\le  \frac{12\gamma d_{\mathscr{A}}}{\alpha}
\end{eqnarray*}
with probability approaching 1, where the first inequality follows from  \eqref{eq.lasso3}, and the second inequality follows from $\|\VEC(\tilde{\mathbf{A}}-\hat{\tilde{\mathbf{a}}}) \|_2\le \frac{3\gamma \sqrt{d_{\mathscr{A}}}}{\alpha}$. The proof of Theorem \ref{TM.Main2}.(1) is now completed.

\medskip

(2). Firstly, we emphasize that \eqref{def.argmin0} and \eqref{def.argmin} share the exactly same $\gamma$. We are now ready to prove that

\begin{eqnarray}\label{KKT9}
    \sgn(\hat{\tilde{\mathbf{a}}}_{\phi}) =\sgn (\tilde{\mathbf{A}})\quad (\text{i.e.,} \quad \sgn(\hat{\pmb{\beta}}_{\phi}) =\sgn (\pmb{\beta}_0))
\end{eqnarray}
with probability approaching one. Due to the property of convex optimization, $\hat{\tilde{\mathbf{a}}}_{\phi}$ is a solution of \eqref{def.argmin} if and only if there exits a subgradient

\begin{eqnarray*}
    \tilde{\mathbf{g}}\in \left\{\mathbf{z}\in \mathbb{R}^{N^2p} \mid  z_\ell = \sgn (\hat{\beta}_{\phi,\ell} )\phi_\ell  \text{ for } \sgn (\hat{\beta}_{\phi,\ell} )\ne 0,\text{ and }|z_\ell|\le \phi_\ell \text{ elsewhere}, \text{ for } \ell \in[N^2p]\right\}
\end{eqnarray*} 
such that 
\begin{eqnarray*}
    \mathbf{0}=\frac{1}{T-h}\sum_{t=1}^{T-h}\mathbf{Z}_t \mathbf{Z}_t^\top \hat{\pmb{\beta}}_\phi - \frac{1}{T-h}\sum_{t=1}^{T-h}\mathbf{Z}_t\mathbf{x}_{t+h}+\frac{\gamma}{2} \, \tilde{\mathbf{g}}.
\end{eqnarray*}

Note that $\sgn(\hat{\pmb{\beta}}_\phi) =\sgn (\pmb{\beta}_0)$ holds if and only if the following Karush-Kuhn-Tucker conditions hold:

\begin{eqnarray}
    \mathbf{0}=&&\frac{1}{T-h}\sum_{t=1}^{T-h}\mathbf{Z}_{\mathscr{A},t} \mathbf{Z}_{\mathscr{A},t}^\top (\hat{\pmb{\beta}}_{\phi,\mathscr{A}} - \pmb{\beta}_{0,{\mathscr{A}}})-\frac{1}{T-h}\sum_{t=1}^{T-h}\mathbf{Z}_{\mathscr{A},t}\mathbf{u}_{t,h} +\frac{\gamma}{2}\,\tilde{\mathbf{g}}_{\mathscr{A}},\label{KKT1} \\
    \mathbf{0}=&&\frac{1}{T-h}\sum_{t=1}^{T-h}\mathbf{Z}_{\bar{\mathscr{A}},t} \mathbf{Z}_{\mathscr{A},t}^\top (\hat{\pmb{\beta}}_{\phi,\mathscr{A}} - \pmb{\beta}_{0,{\mathscr{A}}})-\frac{1}{T-h}\sum_{t=1}^{T-h}\mathbf{Z}_{\bar{\mathscr{A}},t}\mathbf{u}_{t,h} +\frac{\gamma}{2}\,\tilde{\mathbf{g}}_{\bar{\mathscr{A}}}\label{KKT2},
\end{eqnarray}
and

\begin{eqnarray}\label{KKT3}
    \mathrm{sgn}(\beta_{0,\mathscr{A}, \ell})(\hat{\beta}_{\phi,\mathscr{A},\ell}-\beta_{0,\mathscr{A}, \ell})&>&-|\beta_{0,\mathscr{A}, \ell}|,\quad \text{for}\, \ell \in [d_{\mathscr{A}} ],
\end{eqnarray}
where $\tilde{\mathbf{g}}_{\mathscr{A}}$ and $\tilde{\mathbf{g}}_{\bar{\mathscr{A}}}$ contain respectively the elements of $ \tilde{\mathbf{g}}$ that are associated with the non-zero and zero elements of $\pmb{\beta}_0$.

We now investigate $\frac{1}{T-h}\sum_{t=1}^{T-h}\mathbf{Z}_{\mathscr{A},t} \mathbf{Z}_{\mathscr{A},t}^\top$. By standard eigenvalue perturbation bounds, 

\begin{eqnarray*}
    \left|\lambda_{\rm min}\left(\frac{1}{T-h}\sum_{t=1}^{T-h}\mathbf{Z}_{\mathscr{A},t} \mathbf{Z}_{\mathscr{A},t}^\top\right)-\lambda_{\rm min}\left(\pmb{\Sigma}_{\mathbf{Z}, \mathscr{A}}\right)\right|&\leq&  \left\|\frac{1}{T-h}\sum_{t=1}^{T-h}\mathbf{Z}_{\mathscr{A},t} \mathbf{Z}_{\mathscr{A},t}^\top- \pmb{\Sigma}_{\mathbf{Z}, \mathscr{A}}\right\|_2
    \notag\\
    &\leq&d_\mathscr{A} \left\|\frac{1}{T-h}\sum_{t=1}^{T-h}\mathbf{Z}_{\mathscr{A},t} \mathbf{Z}_{\mathscr{A},t}^\top- \pmb{\Sigma}_{\mathbf{Z}, \mathscr{A}}\right\|_{\max}
    \notag\\
    &\leq& d_\mathscr{A}d_{\pmb{\beta}}^2\mu_{h}^{-1-1/J}\gamma,
\end{eqnarray*}
where the last inequality follows from the second result in \eqref{basis.ineq1} and $\gamma \asymp \frac{\mu_{h}^{1+1/J}\sqrt{\log (N)}}{\sqrt{T}}$. Under the condition $d_{\mathscr{A}} d_{\pmb{\beta}}^2\frac{\sqrt{\log N}}{\sqrt{T}}\rightarrow0$ of Assumption \ref{AS3}, we establish the invertibility of $\frac{1}{T-h}\sum_{t=1}^{T-h}\mathbf{Z}_{\mathscr{A},t} \mathbf{Z}_{\mathscr{A},t}^\top$. Therefore, we can write

\begin{eqnarray}\label{KKT4}
    \left\|\left(\frac{1}{T-h}\sum_{t=1}^{T-h}\mathbf{Z}_{\mathscr{A},t} \mathbf{Z}_{\mathscr{A},t}^\top \right)^{-1}\right\|_{\infty}&\leq& \sqrt{d_\mathscr{A}}\left\|\left(\frac{1}{T-h}\sum_{t=1}^{T-h}\mathbf{Z}_{\mathscr{A},t} \mathbf{Z}_{\mathscr{A},t}^\top \right)^{-1}\right\|_{2}
    \notag\\
    &\leq& \frac{\sqrt{d_\mathscr{A}}}{\lambda_{\rm min}\left(\frac{1}{T-h}\sum_{t=1}^{T-h}\mathbf{Z}_{\mathscr{A},t} \mathbf{Z}_{\mathscr{A},t}^\top\right)}
    \notag\\
    &\leq&\frac{2\sqrt{d_{\mathscr{A}}}}{\alpha},
\end{eqnarray}
where $\alpha$ is defined in Assumption \ref{AS3} already.
  
From \eqref{KKT1}, we obtain the following expression for $\hat{\pmb{\beta}}_{\phi,\mathscr{A}} - \pmb{\beta}_{0,{\mathscr{A}}}$:

\begin{eqnarray}\label{KKT5}
    \hat{\pmb{\beta}}_{\phi,\mathscr{A}} - \pmb{\beta}_{0,{\mathscr{A}}}&=&\left(\frac{1}{T-h}\sum_{t=1}^{T-h}\mathbf{Z}_{\mathscr{A},t} \mathbf{Z}_{\mathscr{A},t}^\top\right)^{-1}\left(\frac{1}{T-h}\sum_{t=1}^{T-h}\mathbf{Z}_{\mathscr{A},t}\mathbf{u}_{t,h} -\frac{\gamma}{2}\,\tilde{\mathbf{g}}_{\mathscr{A}}\right).
\end{eqnarray}
In light of \eqref{KKT3}, to establish both \eqref{KKT1} and \eqref{KKT3}, it suffices to verify that 
\begin{eqnarray*}
    \left\|\left(\frac{1}{T-h}\sum_{t=1}^{T-h}\mathbf{Z}_{\mathscr{A},t} \mathbf{Z}_{\mathscr{A},t}^\top\right)^{-1}\left(\frac{1}{T-h}\sum_{t=1}^{T-h}\mathbf{Z}_{\mathscr{A},t}\mathbf{u}_{t,h} -\frac{\gamma}{2}\,\mathbf{g}_{\mathscr{A}}\right)\right\|_{\infty}<\min_{\ell \in [d_\mathscr{A}]}|\beta_{0,\mathscr{A},\ell}|.
\end{eqnarray*}
Using \eqref{KKT4} together with the first bound in \eqref{basis.ineq1}, we obtain
\begin{eqnarray}\label{KKT7}
    &&\left\|\left(\frac{1}{T-h}\sum_{t=1}^{T-h}\mathbf{Z}_{\mathscr{A},t} \mathbf{Z}_{\mathscr{A},t}^\top\right)^{-1}\left(\frac{1}{T-h}\sum_{t=1}^{T-h}\mathbf{Z}_{\mathscr{A},t}\mathbf{u}_{t,h} -\frac{\gamma}{2}\,\mathbf{g}_{\mathscr{A}}\right)\right\|_{\infty}
    \notag\\
    &\leq&\left\|\left(\frac{1}{T-h}\sum_{t=1}^{T-h}\mathbf{Z}_{\mathscr{A},t} \mathbf{Z}_{\mathscr{A},t}^\top\right)^{-1}\right\|_{\infty}\cdot \left(\left\|\frac{1}{T-h}\sum_{t=1}^{T-h}\mathbf{Z}_{\mathscr{A},t}\mathbf{u}_{t,h}\right\|_{\infty} +\frac{\gamma}{2}\max_{\ell\in[d_{\mathscr{A}}]}\phi_{\mathscr{A},\ell}\right)
     \notag\\
    &\leq&\frac{\sqrt{d_{\mathscr{A}}}\gamma}{2\alpha}\left(1+2 \max_{\ell\in[d_{\mathscr{A}}]}\phi_{\mathscr{A},\ell}\right).
\end{eqnarray}
Therefore,  \eqref{KKT1} and \eqref{KKT3} can be satisfied provided

\begin{eqnarray*}
    \min_{\ell \in [d_\mathscr{A}]}|\beta_{0,\mathscr{A},\ell}|>\frac{\sqrt{d_{\mathscr{A}}}\gamma}{2\alpha}\left(1+2 \max_{\ell\in[d_{\mathscr{A}}]}\phi_{\mathscr{A},\ell}\right),
\end{eqnarray*}
which is stated in the body of this lemma by noting that $\gamma \asymp \frac{\mu_{h}^{1+1/J}\sqrt{\log (N)}}{\sqrt{T}}$.

We now turn to \eqref{KKT2}, and let $d_{\bar{\mathscr{A}}} \coloneqq \|\mathbf{S}_{\bar{\mathscr{A}}}\|^2$. Combining \eqref{KKT2} and \eqref{KKT5}, it suffices to show that for any $\ell\in[d_{\bar{\mathscr{A}}}]$
  
\begin{eqnarray}\label{KKT6}
   \frac{\gamma}{2}\,\phi_{\bar{\mathscr{A}},\ell}&\geq& \left|\left(\frac{1}{T-h}\sum_{t=1}^{T-h}\mathbf{Z}_{\bar{\mathscr{A}},t,\ell} \mathbf{Z}_{\mathscr{A},t}^\top\right) \left(\frac{1}{T-h}\sum_{t=1}^{T-h}\mathbf{Z}_{\mathscr{A},t} \mathbf{Z}_{\mathscr{A},t}^\top\right)^{-1}\left(\frac{1}{T-h}\sum_{t=1}^{T-h}\mathbf{Z}_{\mathscr{A},t}\mathbf{u}_{t,h} -\frac{\gamma}{2}\,\mathbf{g}_{\mathscr{A}}\right)\right|
    \notag\\
    &&+\left|\frac{1}{T-h}\sum_{t=1}^{T-h}\mathbf{Z}_{\bar{\mathscr{A}},t,\ell}\mathbf{u}_{t,h} \right|,
\end{eqnarray}
where  $\mathbf{Z}_{\bar{\mathscr{A}},t,\ell}$ is the  $\ell^{th}$ row of $\mathbf{Z}_{\bar{\mathscr{A}},t}$ .

For the first term on the right-hand side of  \eqref{KKT6}, by \eqref{KKT7} and  \eqref{basis.ineq1}, 
\begin{eqnarray}\label{KKT8}
     &&\max_{\ell\in [d_{\bar{\mathscr{A}}}]}\left|\left(\frac{1}{T-h}\sum_{t=1}^{T-h}\mathbf{Z}_{\bar{\mathscr{A}},t,\ell} \mathbf{Z}_{\mathscr{A},t}^\top\right) \left(\frac{1}{T-h}\sum_{t=1}^{T-h}\mathbf{Z}_{\mathscr{A},t} \mathbf{Z}_{\mathscr{A},t}^\top\right)^{-1}\left(\frac{1}{T-h}\sum_{t=1}^{T-h}\mathbf{Z}_{\mathscr{A},t}\mathbf{u}_{t,h} -\frac{\gamma}{2}\,\mathbf{g}_{\mathscr{A}}\right)\right|
      \notag\\
     &\leq&d_{\mathscr{A}}\max_{\ell\in [d_{\bar{\mathscr{A}}}]}\left\|\frac{1}{T-h}\sum_{t=1}^{T-h} \mathbf{Z}_{\mathscr{A},t}\mathbf{Z}_{\bar{\mathscr{A}},t,\ell}^\top\right\|_{\infty} \notag \\
     &&\cdot \left\| \left(\frac{1}{T-h}\sum_{t=1}^{T-h}\mathbf{Z}_{\mathscr{A},t} \mathbf{Z}_{\mathscr{A},t}^\top\right)^{-1}\right\|_2\left\|\frac{1}{T-h}\sum_{t=1}^{T-h}\mathbf{Z}_{\mathscr{A},t}\mathbf{u}_{t,h} -\frac{\gamma}{2}\,\mathbf{g}_{\mathscr{A}}\right\|_\infty
     \notag\\
     &\leq&\frac{d_{\mathscr{A}}\gamma}{2\alpha}\left(\left\|\pmb{\Sigma}_{\mathbf{B}} \right\|_{\rm max}+ d_{\pmb{\beta}}^2\mu_{h}^{-1-1/J}\gamma\right)\left(1+2 \max_{\ell\in[d_{\mathscr{A}}]}\phi_{\mathscr{A},\ell}\right)
     \notag\\
     &\leq&\left(\frac{d_{\mathscr{A}}\gamma}{2\alpha}\left\|\pmb{\Sigma}_{\mathbf{B}} \right\|_{\rm max}+ \frac{\gamma}{4}\right)\left(1+2 \max_{\ell\in[d_{\mathscr{A}}]}\phi_{\mathscr{A},\ell}\right),
\end{eqnarray}
where the first inequality follows from sequentially applying the submultiplicativity of operator norms $\|\mathsf{A}\mathsf{B}\|_2\leq\|\mathsf{A}\|_2\|\mathsf{B}\|_2$ and the fact  $\|\mathsf{A}\|_2\leq \sqrt{{\rm nrows(\mathsf{A})}}\|\mathsf{A}\|_\infty$, where ${\rm nrows(\cdot)}$ computes the number of rows for a matrix. The third inequality holds by the condition $d_\mathscr{A}d_{\pmb{\beta}}^2\mu_{h}^{-1-1/J}\gamma\leq \frac{\alpha}{2}$.

It follows from \eqref{basis.ineq1} that the second term on the right-hand side of \eqref{KKT6} is bounded by $\frac{\gamma}{4}$. Combining this bound with \eqref{KKT8} yields the desired result in \eqref{KKT6}, provided that $ (1+\frac{2d_{\mathscr{A}}}{\alpha}\left\|\pmb{\Sigma}_{\mathbf{B}} \right\|_{\rm max} ) (1+2 \max_{\ell\in[d_{\mathscr{A}}]}\phi_{\mathscr{A},\ell} )\leq 2\min_{\ell\in[d_{\bar{\mathscr{A}}}]}\phi_{\bar{\mathscr{A}},\ell}-1$. Therefore,  \eqref{KKT2} is established, which completes the proof. 
\end{proof}

\medskip

\begin{proof}[Proof of Theorem \ref{TM.Main3}]

\item 
In what follows, we show that $$\text{IC}(\mathfrak{p})>\text{IC}(p) $$ for both cases $\mathfrak{p}>p$ and $\mathfrak{p}<p$ with probability approaching 1, so the proof is then completed.

Recall the notations defined in the proof of Theorem \ref{TM.Main2}, such as $\mathbf{Z}_t\coloneqq \mathbf{X}_t\otimes \mathbf{I}_N$ and $\pmb{\beta}_0\coloneqq\VEC(\tilde{\mathbf{A}})$. Accordingly, we suppose that $\hat{\pmb{\beta}}_{\mathfrak{p}}$ is obtained via \eqref{def.argmin0} with given $\mathfrak{p}$ lags. Therefore, $\text{IC}(\mathfrak{p})$ can be expanded as follows:

\begin{eqnarray*}
    \text{IC}(\mathfrak{p}) &=&\frac{1}{T-h}\sum_{t=1}^{T-h} \| (\mathbf{X}_t^\top \otimes \mathbf{I}_N) \VEC(\tilde{\mathbf{A}}-\hat{\tilde{\mathbf{a}}}_{\mathfrak{p}}) \|_2^2+\frac{1}{T-h}\sum_{t=1}^{T-h} \| \mathbf{u}_{t,h}\|_2^2\notag \\
    &&+\frac{2}{T-h}\sum_{t=1}^{T-h}  \mathbf{u}_{t,h}^\top (\mathbf{X}_t^\top \otimes \mathbf{I}_N) \VEC(\tilde{\mathbf{A}}-\hat{\tilde{\mathbf{a}}}_{\mathfrak{p}}) + \mathfrak{p}\cdot \xi\notag \\
    &=&\frac{1}{T-h}\sum_{t=1}^{T-h} \|\mathbf{Z}_t^\top (\pmb{\beta}_{0} - \hat{\pmb{\beta}}_{\mathfrak{p}}) \|_2^2+\frac{1}{T-h}\sum_{t=1}^{T-h} \| \mathbf{u}_{t,h} \|_2^2 \notag \\
    &&+\frac{2}{T-h}\sum_{t=1}^{T-h} \mathbf{u}_{t,h}^\top \mathbf{Z}_t^\top (\pmb{\beta}_{0} - \hat{\pmb{\beta}}_{\mathfrak{p}}) + \mathfrak{p}\cdot \xi.
\end{eqnarray*}

\medskip

Case 1 --- We consider the case with $\mathfrak{p} > p$. By the proof for (1) of Theorem \ref{TM.Main2}, we know that

\begin{eqnarray*}
    &&\frac{1}{T-h}\sum_{t=1}^{T-h} \| (\mathbf{X}_t^\top \otimes \mathbf{I}_N) \VEC(\tilde{\mathbf{A}}-\hat{\tilde{\mathbf{a}}}_{\mathfrak{p}}) \|_2^2 \le  2\gamma\|\VEC(\mathbf{S}_{\mathscr{A}}\circ(\tilde{\mathbf{A}}-\hat{\tilde{\mathbf{a}}}_{\mathfrak{p}}))\|_1
\end{eqnarray*}
and 

\begin{eqnarray*}
    &&\left|\frac{1}{T-h}\sum_{t=1}^{T-h}  \mathbf{u}_{t,h}^\top (\mathbf{X}_t^\top \otimes \mathbf{I}_N) \VEC(\tilde{\mathbf{A}}-\hat{\tilde{\mathbf{a}}}_{\mathfrak{p}})\right|\le \frac{\gamma}{4}\|\VEC(\tilde{\mathbf{A}}-\hat{\tilde{\mathbf{a}}}_{\mathfrak{p}}) \|_1,
\end{eqnarray*}
where $\gamma \asymp \frac{\mu_{h}^{1+1/J}\sqrt{\log (N)}}{\sqrt{T}}$, the first inequality follows from \eqref{eq.lasso2}, and the second inequality follows from \eqref{basis.ineq1}. Therefore, as long as

\begin{eqnarray}\label{cond.1}
    \xi/  (\gamma\|\VEC(\tilde{\mathbf{A}}-\hat{\tilde{\mathbf{a}}}_{\mathfrak{p}}) \|_1)  \to \infty\quad\text{and}\quad\xi\to 0,
\end{eqnarray}
we have $\text{IC}(\mathfrak{p}) > \text{IC}(p)$ for $\mathfrak{p} > p$. In connection with Lemma \ref{TM.Main2}, simple algebra shows that the condition \eqref{cond.1} is equivalent to requiring 

\begin{eqnarray*}
    \frac{\mu_h^{2+2/J} d_{\mathscr{A}} \log N}{T\xi}\to 0\quad\text{and}\quad\xi\to 0.
\end{eqnarray*}

\medskip

Case 2 --- We now consider the case with $\mathfrak{p}<p$. Without loss of generality, suppose that $\mathfrak{p}=p - 1$. Additionally, we still let $\hat{\pmb{\beta}}_{\mathfrak{p}}$ have the same dimension as $\pmb{\beta}_0$, but let the elements corresponding to $\mathbf{A}_p$ be 0. Accordingly, suppose that $\mathbf{S}_p$ be a selection matrix such that $\mathbf{S}_p \pmb{\beta}_{0}$ yields a column vector including all the non-zeros elements of $\mathbf{A}_p$.

Firstly, we pay attention to the term $\frac{1}{T-h}\sum_{t=1}^{T-h} \|\mathbf{Z}_t^\top (\pmb{\beta}_{0} - \hat{\pmb{\beta}}_{\mathfrak{p}}) \|_2^2$. 

\begin{eqnarray*}
    \frac{1}{T-h}\sum_{t=1}^{T-h} \|\mathbf{Z}_t^\top (\pmb{\beta}_{0} - \hat{\pmb{\beta}}_{\mathfrak{p}}) \|_2^2 &\asymp &(\pmb{\beta}_0-\hat{\pmb{\beta}}_{\mathfrak{p}})^\top \pmb{\Sigma}_{\mathbf{B}}(\pmb{\beta}_0-\hat{\pmb{\beta}}_{\mathfrak{p}})\notag \\
    &\ge  & \alpha \cdot \|\mathbf{S}_p \pmb{\beta}_{0} \|_2^2>c>0,
\end{eqnarray*}
where the first step follows from the proof for (1) of Theorem \ref{TM.Main2}, the second step follows from Assumption \ref{AS3} and the facts that $\mathfrak{p}=p-1$ and the definition of $\hat{\pmb{\beta}}_{\mathfrak{p}}$, and $c$ is a non-negligible positive number. In connection with the proof for (1) of Theorem \ref{TM.Main2}, it is then easy to see $\text{IC}(\mathfrak{p})>\text{IC}(p)$ for $\mathfrak{p}<p$. 

\medskip

Based on the above development, the proof is now completed.
\end{proof}

\medskip

\begin{proof}[Proof of Theorme \ref{TM.Main4}]
\item 

(1). Recall the notations defined in the proof of Theorem \ref{TM.Main1} and in the beginning of Appendix \ref{AP.B}. For example, $\mathscr{F}_{t,s}=\sigma(\pmb{\varepsilon}_t,\ldots,\pmb{\varepsilon}_s)$ for $t\ge s$; we also let $U_{ij} (t,h,m) \coloneqq E[u_{it,h} x_{jt} \mid \mathscr{F}_{t+h,t+h-m}]$. Additionally, we denote the projection operator $\mathscr{P}_{l}(\cdot) = E[\cdot \mid \mathscr{F}_{l}] -E[\cdot \mid \mathscr{F}_{l-1}]$, where $\mathscr{F}_l=\sigma(\pmb{\varepsilon}_l,\pmb{\varepsilon}_{l-1},\ldots)$.

We then write

\begin{eqnarray}\label{mm40}
    \hat{\pmb{\Omega}}_h &=& \frac{1}{T-h}\sum_{|t- k|< h}^{T-h} (\mathbf{X}_t  \otimes \mathbf{I}_N)\hat{\mathbf{u}}_{t,h} \hat{\mathbf{u}}_{k,h}^\top(\mathbf{X}_k^\top \otimes \mathbf{I}_N)\notag \\
    &=& \frac{1}{T-h}\sum_{|t-k|< h}^{T-h}\mathbf{Z}_t\mathbf{u}_{t,h} \mathbf{u}_{k,h}^\top\mathbf{Z}_k^\top + \frac{1}{T-h}\sum_{|t-k|< h}^{T-h} \mathbf{Z}_t\mathbf{Z}_t^\top  (\pmb{\beta}_0-\hat{\pmb{\beta}}_\phi) (\pmb{\beta}_0-\hat{\pmb{\beta}}_\phi)^\top \mathbf{Z}_k\mathbf{Z}_k^\top\notag \\
    &&+ \frac{1}{T-h}\sum_{|t-k|< h}^{T-h} \mathbf{Z}_t\mathbf{Z}_t^\top (\pmb{\beta}_0-\hat{\pmb{\beta}}_\phi) \mathbf{u}_{k,h}^\top \mathbf{Z}_k^\top\notag \\
    &&+ \frac{1}{T-h}\sum_{|t-k|< h}^{T-h} \mathbf{Z}_t \mathbf{u}_{t,h} (\pmb{\beta}_0-\hat{\pmb{\beta}}_\phi)^\top \mathbf{Z}_k\mathbf{Z}_k^\top
    \notag\\
    &\coloneqq&\pmb{\mathcal{B}}_1+\cdots+\pmb{\mathcal{B}}_4,
\end{eqnarray}
where $\mathbf{Z}_t$, $\pmb{\beta}_0$ and $\hat{\pmb{\beta}}_\phi$ are defined in the beginning of Appendix \ref{AP.B}.

To study the order of the first term (i.e., $\pmb{\mathcal{B}}_1$), it is sufficient to consider $ \hat{\omega}-\omega$, where

\begin{eqnarray*}
    \hat{\omega} = \sum_{|t-k|< h}^{T-h}x_{j_1t}u_{i_1t,h} u_{i_2k,h} x_{j_2k} \quad\text{and}\quad
    \omega  = \sum_{|t-k|< h}^{T-h}E[x_{j_1t}u_{i_1t,h} u_{i_2k,h} x_{j_2k}].
\end{eqnarray*}
Note that 

\begin{eqnarray*}
    \hat{\omega}  &=& \sum_{|t-k|< h}^{T-h}x_{j_1t}u_{i_1t,h} u_{i_2k,h} x_{j_2k}\notag \\
    &=& \sum_{t=1}^{T-h} u_{i_1t,h} u_{i_2t,h}x_{j_1t} x_{j_2t} +  \sum_{0<t-k< h } u_{i_1t,h} u_{i_2k,h}x_{j_1t} x_{j_2k} + \sum_{0<k-t< h } u_{i_1t,h} u_{i_2k,h}x_{j_1t} x_{j_2k} \notag\\
    &\eqqcolon&\hat{\omega}_1+\hat{\omega}_2+\hat{\omega}_3.
\end{eqnarray*}
Accordingly, we can define $\omega_1$, $\omega_2$, and $\omega_3$ as the expectation counterparts. 
We now study the term $\hat{\omega}_2-\omega_2$. The investigation about $\hat{\omega}_1-\omega_1$ and $\hat{\omega}_3-\omega_3$ are similar. We drop the indices $(i_1,i_2,j_1, j_2)$ for simplicity whenever possible if no misunderstanding arises.

Note that

\begin{eqnarray*}
    \hat{\omega}_2 =\sum_{t=2}^{T-h}\sum_{t_{h}< k< t } U_{i_1j_1} (t,h)U_{i_2j_2} (k,h) ,
\end{eqnarray*}
where $t_{h}\coloneqq (t-h)\vee 0$ and $U_{ij} (t,h)=u_{it,h} x_{jt}$.  
 Without loss of generality, we split $\{2,\ldots, T-h\}$ into $b_T$ blocks (denoted as $B_n$'s) with sample size $2m_T$ that satisfies $m_T=\lfloor T^{\theta}\rfloor$, where $0<\theta<1$. In other words, $\cup_{n=1}^{b_T}B_n =\{2,\ldots, T-h\}$. Denote

\begin{eqnarray*}
    Q_{T}\coloneqq \sum_{n=1}^{b_T} Q_n \quad \text{and}\quad Q_n\coloneqq\sum_{t\in B_n}\sum_{t_h< k<t } U_{i_1j_1} (t,h,m_T)U_{i_2j_2} (k,h,m_T).
\end{eqnarray*}
Moreover, given $m_T > h$, $Q_1, Q_3,\ldots $ are independent blocks, and $Q_2, Q_4,\ldots $ are also independent blocks. For notational simplicity, we further define

\begin{eqnarray}\label{def.SSQ}
     S_{ij, h}(t)&=&\sum_{t_h< k< t }U_{ij} (k,h),\quad S_{ij, h}(t,m)=\sum_{t_h< k< t }U_{ij} (k,h,m),
     \notag\\
     Q^\ast_{T}&=&\sum_{t=2}^{T-h}\sum_{t_h< k<t } U_{i_1j_1} (t,h)U_{i_2j_2} (k,h,m_T).
\end{eqnarray}

Then we write

\begin{eqnarray}\label{m9}
    &&\Pr\left(|\hat{\omega}_2 - \omega_2|>\epsilon_{NT}\right)\notag \\
    &\le &\Pr\left(|\hat{\omega}_2 - \omega_2-Q_{T}+E[Q_{T}]|>\frac{\epsilon_{NT}}{2}\right) + \Pr\left(|Q_{T} - E[Q_{T}]|>\frac{\epsilon_{NT}}{2}\right),
\end{eqnarray}
where $\epsilon_{NT}=c_\epsilon \sqrt{M Th\log N}$ with constants $c_\epsilon>0$ and a sufficiently large $M$ that satisfies that $M>4$ and $N^4T^{-M}\rightarrow0$. In what follows, we investigate these probabilities on the right hand side of \eqref{m9} one by one. 

For the first term, it is clear to see that 

\begin{eqnarray*}
   \|\hat{\omega}_2 - \omega_2-Q_{T}+E[Q_{T}]\|_{J/2} &\leq& \|\hat{\omega}_2 - \omega_2-Q^\ast_{T}+E[Q^\ast_{T}]\|_{J/2}
   \notag\\
   &&+\|Q_{T} - E[Q_{T}]-Q^\ast_{T}+E[Q^\ast_{T}]\|_{J/2}.
\end{eqnarray*}

Since $\mathscr{P}_l( \hat{\omega}_2-Q^\ast_{T})$ is a MDS, for $J>4$, by applying the Burkholder inequality and Minkowski inequality sequentially as in \eqref{nu_ml}, we obtain

\begin{eqnarray}\label{m8}
    \|\hat{\omega}_2 - \omega_2-Q^\ast_{T}+E[Q^\ast_{T}]\|^2_{J/2}&=&\Big\|\sum_{l=-\infty}^{T} \mathscr{P}_l( \hat{\omega}_2-Q^\ast_{T})\Big\|^2_{J/2}
    \notag\\
    &\leq&O(1)\sum_{l=-\infty}^{T}\| \mathscr{P}_l( \hat{\omega}_2-Q^\ast_{T})\|^2_{J/2}.
\end{eqnarray}

From now on, we suppress the indices $(i_1,i_2,j_1,j_2)$ in $U_{i_1j_1} (t,h)$, $U_{i_1j_1} (t,h, m)$,  $S_{i_2j_2, h}(t)$ and $S_{i_2j_2, h}(t,m_T)$ for notational simplicity. Let $U_l^*(t,h)$, $U_l^*(t,h,m)$,  $ S_{h,l}^*(t)$, and $ S_{h,l}^*(t,m)$ be coupled versions of $U(t,h) $, $U(t,h,m)$, $S_{h}(t)$,  and $S_{h}(t,m)$, respectively,  by replacing $\pmb{\varepsilon}_l$ with an independent copy of $\pmb{\varepsilon}_l^*$. For each $l$,  by Jensen inequality,

\begin{eqnarray*}
    &&\|\mathscr{P}_l( \hat{\omega}_2-Q^\ast_{T})\|_{J/2} \notag \\
    &=& \Big\|\sum_{t=2\vee (l-h)}^{T-h}E\big[\big( U (t,h)\big(S_h (t)-S_h(t,m_T)\big)-U_l^\ast (t,h)\big(S_{h,l}^\ast (t)-S_{h,l}^\ast(t,m_T)\big)\big)\mid \mathscr{F}_{l} \big]\Big\|_{J/2}
    \notag\\
    &\leq&\Big\|\sum_{t=2\vee (l-h)}^{T-h}\big( U (t,h)\big(S_h (t)-S_h(t,m_T)\big)-U_l^\ast (t,h)\big(S_{h,l}^\ast (t)-S_{h,l}^\ast(t,m_T)\big)\big)\Big\|_{J/2}
    \notag\\
    &\leq&\Big\| \sum_{t=2\vee (l-h)}^{T-h}U_l^\ast (t,h)\big(S_h (t)-S_h(t,m_T)-S_{h,l}^\ast (t)+S_{h,l}^\ast(t,m_T)\big) \Big\|_{J/2}
    \notag\\
    &&+\Big\|\sum_{t=2\vee (l-h)}^{T-h} \big(U (t,h)-U_l^\ast (t,h)\big) \big(S_h (t)-S_h(t,m_T)\big) \Big\|_{J/2}
    \notag\\
    &\eqqcolon&\mathcal{A}_{l,1}+\mathcal{A}_{l,2}.
\end{eqnarray*}
For the first term on the right-hand side, by Cauchy-Schwarz inequality and Lemma \ref{LMB7}, 

\begin{eqnarray*}
     \mathcal{A}_{l,1} &\leq& \sum_{k=1\vee (l-h)}^{T-h-1}\Big\|\sum_{t=k+1}^{k+h-1} U_l^\ast (t,h)\Big\|_{J}\cdot \Big\|U(k,h)-U(k,h,m_T)-U_l^\ast (k,h)+U_l^\ast(k,h,m_T)\big\|_{J}
    \notag\\
    &=& O(\sqrt{h})\sum_{k=1\vee (l-h)}^{T-h-1}  \|U(k,h)-U(k,h,m_T) \|_{J} \wedge  \|U(k,h)-U_l^\ast (k,h)\|_{J} 
    \notag\\
    &=& O(\sqrt{h})\sum_{k=1\vee (l-h)}^{T-h-1}  d_{J,i_2j_2}(h,m_T)\wedge c_{J,i_2j_2}(h,k+h-l) .
\end{eqnarray*}
Finally, it yields that 

\begin{eqnarray}\label{m6}
    \sum_{l=-\infty}^{T}  \mathcal{A}_{l,1} ^2&=&O(h)\sum_{l=-\infty}^{T}\left(\sum_{k=1\vee (l-h)}^{T-h-1}  d_{J,i_2j_2}(h,m_T)\wedge c_{J,i_2j_2}(h,k+h-l) \right)^2
    \notag\\
    &=&O(h)\sum_{l=-\infty}^{T}\left(\sum_{k=0\vee(h-l+1)}^{T-l-1}  d_{J,i_2j_2}(h,m_T)\wedge c_{J,i_2j_2}(h,k) \right)^2
    \notag\\
    &=&O(Th)d_{J,i_2j_2}^{\ast2}(h,m_T),
\end{eqnarray}
where $d_{J,ij}^\ast(h,m)=\sum_{k=0}^{\infty} d_{J,ij}(h,m) \wedge c_{J,ij}(h,k)$. 
      
For $\mathcal{A}_{l,2}$, using similar arguments and Lemma \ref{LMB7}, we obtain
\begin{eqnarray*}
    \mathcal{A}_{l,2} &\leq& \sum_{t=2\vee (l-h)}^{T-h} \big\|U (t,h)-U_l^\ast (t,h)\big\|_{J} \big\|S_h (t)-S_h(t,m_T)\big\|_{J} 
    \notag\\
    &=&  O(\sqrt{h})d_{J,i_1j_1}(h,m_T)\sum_{t=2\vee (l-h)}^{T-h} c_{J,i_2j_2}(h,t+h-l),  
\end{eqnarray*}
which immediately gives 
\begin{eqnarray}\label{m7}
   \sum_{l=-\infty}^{T}   \mathcal{A}_{l,2}^2 &\leq&  O(h)d^2_{J,i_1j_1}(h,m_T)\sum_{l=-\infty}^{T}\left(\sum_{t=2\vee (l-h)}^{T-h} c_{J,i_2j_2}(h,,t+h-l)\right)^2
    \notag\\
    &\leq&O(Th)d_{J,i_2j_2}^{\ast2}(h,m_T).
\end{eqnarray}

By \eqref{m8}, \eqref{m6} and \eqref{m7}, we have
\begin{eqnarray*}
    \|\hat{\omega}_2 - \omega_2-Q^\ast_{T}+E[Q^\ast_{T}]\|_{J/2} &\leq&  O(\sqrt{Th})d_{J,i_2j_2}^{\ast}(h,m_T).
\end{eqnarray*}
Using analogous arguments, we can obtain 
 \begin{eqnarray}\label{mm15}
    \|Q_{T} - E[Q_{T}]-Q^\ast_{T}+E[Q^\ast_{T}]\|_{J/2} &\leq&  O(\sqrt{Th})d_{J,i_1j_1}^{\ast}(h,m_T).
\end{eqnarray}
Combining these results together with Assumption \ref{AS4} gives
\begin{eqnarray}\label{mm17}
    \|\hat{\omega}_2 - \omega_2-Q_{T} + E[Q_{T}]\|_{J/2} &\leq&  O(\sqrt{Th})\big(d_{J,i_1j_1}^{\ast}(h,m_T)+d_{J,i_2j_2}^{\ast}(h,m_T)\big)
    \notag\\
    &=&O\left(h^{1/2}T^{1/2-c_0\theta}\right),
\end{eqnarray}
where $c_0 =\big(\frac{2c_b-1}{2}\big)\big(\frac{c_\delta\wedge c_b-1}{c_\delta\wedge c_b}\big)$.

By Markov inequality and \eqref{mm17}, we are then able to obtain that  
\begin{eqnarray}\label{mm27}
    &&\Pr\left(|\hat{\omega}_2 - \omega_2-Q_{T}+E[Q_{T}]|>\frac{\epsilon_{NT}}{2}\right) \notag \\
    &\leq& \big(\frac{\epsilon_{NT}}{2}\big)^{-J/2}\|\hat{\omega}_2 - \omega_2-Q_{T} + E[Q_{T}]\|_{J/2}^{J/2}
    \notag\\
    &\leq&O\left(\epsilon_{NT}^{-J/2}T^{J/4}h^{J/4}\right) \big(d_{J,i_1j_1}^{\ast}(h,m_T)+d_{J,i_2j_2}^{\ast}(h,m_T)\big)^{J/2} 
    \notag\\
    &=&O\left(\epsilon_{NT}^{-J/2}h^{J/4}T^{(1/2-c_0\theta) J/2}\right).
\end{eqnarray}
 
\medskip

We now proceed to study the second probability on the right-hand side of \eqref{m9}.
Without loss of generality, we suppose that $b_T$ is an even number, and by Lemma \ref{LMB2}, we then write

\begin{eqnarray}\label{mm13}
    &&\Pr(|Q_{T}-E[Q_{T}]|\ge \epsilon_x)\notag \\
    &\le &\Pr\left(\left|\sum_{n=1}^{b_T/2}Q_{2n-1}-\sum_{n=1}^{b_T/2}E[Q_{2n-1}]\right|\ge \epsilon_x/2\right)+\Pr\left(\left|\sum_{n=1}^{b_T/2}Q_{2n}-\sum_{n=1}^{b_T/2}E[Q_{2n}]\right|\ge \epsilon_x/2\right)\notag \\
    &\le &\sum_{j=1}^{b_T} \Pr(|Q_n-E[Q_n]|\ge \epsilon_y) \notag \\
    &&+2 \exp\left(\frac{-\alpha^2 (\epsilon_x/2)^2}{2e^{J/2}\sum_{n=1}^{b_T/2}\|Q_{2n-1}-E[Q_{2n-1}]\|_2^2 }\right) \notag \\
    &&+ 2\left(\frac{ \sum_{n=1}^{b_T/2}\|Q_{2n-1}-E[Q_{2n-1}]\|_{J/2}^{J/2}}{\beta (\epsilon_x/2)\epsilon_y^{J/2-1}} \right)^{\beta \epsilon_x/(2\epsilon_y)}\notag \\
    &&+2 \exp\left(\frac{-\alpha^2 (\epsilon_x/2)^2}{2e^{J/2}\sum_{n=1}^{b_T/2}\|Q_{2n}-E[Q_{2n}]\|_2^2 }\right) \notag \\
    &&+ 2\left(\frac{ \sum_{n=1}^{b_T/2}\|Q_{2n}-E[Q_{2n}]\|_{J/2}^{J/2}}{\beta (\epsilon_x/2)\epsilon_y^{J/2-1}} \right)^{\beta \epsilon_x/(2\epsilon_y)}
\end{eqnarray}
for $J\geq 4$, where $\beta = J/(J+4)$, $\alpha = 1-\beta$, $\epsilon_x=\frac{\epsilon_{NT}}{2}$ and $\epsilon_y=\frac{\beta\epsilon_x}{2c_{M}}$ with the constant $c_{M}>\beta$. Without loss of generality, we only consider the first three terms on the right-hand side of this inequality.

Recall that 

\begin{eqnarray*}
    Q_n &=&\sum_{t\in B_n}\sum_{t_h< k< t } U_{i_1j_1} (t,h,m_T)U_{i_2j_2} (k,h,m_T) \notag \\
    &= & \sum_{t\in B_n} U_{i_1j_1} (t,h,m_T) S_{i_2j_2,h}(t,m_T),
\end{eqnarray*}
where the second line follows from the definition of \eqref{def.SSQ}.

From here, we suppress the indices $(i_1,i_2,j_1, j_2)$ again for notational simplicity. Without loss of generality, we consider the order of $\|Q_1 -E[ Q_1 ]\|_J$. As in \eqref{m8}, we know that  

\begin{eqnarray}\label{mm10}
    \|Q_1 -E[ Q_1 ]\|^2_{J/2}=O(1)\sum_{l=1-m_T}^{2m_T}\big\| \mathscr{P}_l(Q_1)\big\|^2_{J/2}.
\end{eqnarray}
For each $l$, we can further write
\begin{eqnarray*}
    \big\| \mathscr{P}_l(Q_1)\big\|_{J/2}&=&\left\|\sum_{t=2}^{2m_T} \mathscr{P}_{l}(U (t,h,m_T) S_{h}(t,m_T) )\right\|_{J/2}  \notag \\
    &= & \left\|\sum_{t=2}^{2m_T}(E[U (t,h,m_T) S_{h}(t,m_T)\mid \mathscr{F}_{l}]-E[U (t,h,m_T) S_{h}(t,m_T)\mid \mathscr{F}_{l-1}])\right\|_{J/2}  \notag \\
    &= & \left\|\sum_{t=2}^{2m_T}E[(U (t,h,m_T) S_{h}(t,m_T)-U_{l}^* (t,h,m_T) S_{h,l}^*(t,m_T))\mid \mathscr{F}_{l}]\right\|_{J/2}  \notag \\
    &\le & \left\|\sum_{t=2}^{2m_T}U_{l}^* (t,h,m_T)( S_{h}(t, m_T)-S_{h,l}^*(t,m_T) )\right\|_{J/2}  \notag \\
    & &+ \left\|\sum_{t=2}^{2m_T}(U (t,h,m_T)-U_{l}^* (t,h,m_T) )S_{h}(t,m_T)\right\|_{J/2}
    \notag\\
    &\eqqcolon&\mathcal{A}_{l,3}+\mathcal{A}_{l,4}.  
\end{eqnarray*}
For $\mathcal{A}_{l,3}$, by Cauchy-Schwarz inequality and Lemma \ref{LMB7},
\begin{eqnarray*}
    \mathcal{A}_{l,3}&\leq&\sum_{k= 1 }^{2m_T-1} \left\|\big(U (k,h,m_T)-U_{l}^\ast (k,h,m_T )\big)\right\|_{J}\Big\|\sum_{t=k+1}^{k+h}U_{l}^* (t,h,m_T)\Big\|_{J}
    \notag\\
    &\leq&O(\sqrt{h})\sum_{k= 1\vee(l-h) }^{2m_T-1}  c_{J,i_2j_2}(h,k+h-l),
\end{eqnarray*}
which gives that 
\begin{eqnarray}\label{mm11}
    \sum_{l=1-m_T}^{2m_T} \mathcal{A}_{l,3}^2\leq O(m_Th)\left(\sum_{k= 0 }^{\infty} c_{J,i_2j_2}(h,k)\right)^2.
\end{eqnarray}
Analogously, we can obtain 
\begin{eqnarray}\label{mm12}
     \sum_{l=1-m_T}^{2m_T} \mathcal{A}_{l,4}^2&\leq&O(m_Th)\left(\sum_{k= 0 }^{\infty} c_{J,i_1j_1}(h,k)\right)^2.
\end{eqnarray}

By \eqref{mm10}, \eqref{mm11}, and \eqref{mm12}, we obtain 
\begin{eqnarray}\label{mm16}
     \|Q_1 -E[ Q_1 ]\|_{J/2}&\leq&O(\sqrt{m_Th})\left( \sum_{k= 0 }^{\infty} c_{J,i_1j_1}(h,k)+\sum_{k= 0 }^{\infty} c_{J,i_2j_2}(h,k) \right) =O(\sqrt{m_Th}).
\end{eqnarray}
Therefore, we have established that 
\begin{eqnarray*}
    &&\sum_{n=1}^{b_T/2}\|Q_{2n-1}-E[Q_{2n-1}]\|_2^2 =O(Th),\notag\\
    &&\sum_{n=1}^{b_T/2}\|Q_{2n-1}-E[Q_{2n-1}]\|_{J/2}^{J/2}=O\left(b_Tm_T^{J/4}h^{J/4}\right).
\end{eqnarray*} 

For the second term on the right-hand side of \eqref{mm13},  
\begin{eqnarray*}
    \exp\left(\frac{-\alpha^2 (\epsilon_x/2)^2}{2e^{J/2}\sum_{n=1}^{b_T/2}\|Q_{2n-1}-E[Q_{2n-1}]\|_2^2 }\right)&=&O\left(\exp\left(-C_{\alpha,q}\frac{\epsilon_x^2}{Th}\right)\right)
    \notag\\
    &=&O\left(\exp\left(-\frac{c_\epsilon^2}{4}C_{\alpha,q}M\log N\right)\right)
    \notag\\
    &=&O\left(N^{-M}\right),
\end{eqnarray*}
where $C_{\alpha,J/2}=\frac{1}{2}\alpha^2e^{-J/2}$ and the last equality holds for some suitable choice of $c_\epsilon$.

For the third term on the right-hand side of \eqref{mm13}, by setting $c_M>M\left(J/4-1\right)^{-1}(1-\theta)^{-1}$, we obtain
\begin{eqnarray*}
    \left(\frac{ \sum_{n=1}^{b_T/2}\|Q_{2n-1}-E[Q_{2n-1}]\|_{J/2}^{J/2}}{\beta (\epsilon_x/2)\epsilon_y^{J/2-1}} \right)^{\beta \epsilon_x/(2\epsilon_y)}&=&O\left(\left(\frac{ b_Tm_T^{J/4}h^{J/4}}{\beta (\epsilon_x/2)\epsilon_y^{J/2-1}} \right)^{\beta \epsilon_x/(2\epsilon_y)}\right)
    \notag\\
    &=&O\left(\left(\frac{ b_Tm_T^{J/4}h^{J/4}}{c_M\epsilon_y^{J/2}} \right)^{c_M}\right)
    \notag\\
    &=&O\left(T^{-\left(J/4-1\right)(1-\theta)c_M}(\log N)^{-J/4}\right)
    \notag\\
    &=&O\left(T^{-M}\right).
\end{eqnarray*}
In summary of these results and applying Lemma \ref{LMB8}.(2) to $\Pr(|Q_n-E[Q_n]|\ge \epsilon_y)$ with $T$ in the lemma replaced by $m_T$, we have
\begin{eqnarray}\label{mm28}
    &&\Pr(|Q_{T}-E[Q_{T}]|\ge \epsilon_x)
    \notag\\
    &\leq &\sum_{j=1}^{b_T} \Pr(|Q_n-E[Q_n]|\ge \epsilon_y) +O\left((N\wedge T)^{-M}\right)
    +O\left((N\wedge T)^{-M}\right)
    \notag\\
    &=& O\left(\epsilon_y^{-J/2} T \log T\left(  h^{J/4}m_T^{(1/2-c_0\theta) J/2-1} 
    + h^{(1-c_0\theta)J/2-1}+1\right)\right)
    \notag\\
    &&+O\left((N\wedge T)^{-M}\right). 
\end{eqnarray}

By \eqref{m9}, \eqref{mm27}, and \eqref{mm28}, 
\begin{eqnarray}
    &&\Pr\left(|\hat{\omega}_2 - \omega_2|>\epsilon_{NT}\right)\notag \\
    &= &O\left(\epsilon_{NT}^{-J/2} T\log T\left(  h^{J/4} m_T^{(1/2-c_0\theta) J/2-1}
    + h^{(1-c_0\theta)J/2-1}+1\right)\right)
    \notag\\
    &&+O\left((N\wedge T)^{-M}\right).
\end{eqnarray}
Using analogous arguments, we can show that $\Pr\left(|\hat{\omega}_1-\omega_1|>\epsilon_{NT}\right)$ and $Pr\left(|\hat{\omega}_3-\omega_3|>\epsilon_{NT}\right)$ are also bounded by the same order. Therefore, we can set $\theta=1/c_0$ and utilize Assumption \ref{AS4} to show that
\begin{eqnarray}
    &&P(\max_{i_1,i_2,j_1,j_2\in[N]}|\hat{\omega} - \omega|>3\epsilon_{NT})
    \notag\\
    &\leq& \sum_{i_1,i_2,j_1,j_2=1}^NP(|\hat{\omega} - \omega|>3\epsilon_{NT})
    \notag\\
    &= &O\left(\epsilon_{NT}^{-J/2} N^4\log T\left( T h^{J/4} m_T^{(1/2-c_0\theta) J/2-1}
    + Th^{(1-c_0\theta)J/2-1}+T\right)\right)+O\left(N^{4}(N\wedge T)^{-M}\right)
    \notag\\
    &=&O\left( N^4T^{1-J/4}(\log N)^{-J/4}\log T\left( m_T^{(1/2-c_0\theta) J/2-1} 
    + h^{(1/2-c_0\theta) J/2-1}+h^{-J/4}\right)\right)\notag \\
    &&+O\left(N^{4}(N\wedge T)^{-M}\right)
    \notag\\
    &=&o(1),
\end{eqnarray}
which immediately yields the following result: 
\begin{eqnarray*}
    \max_{i_1,i_2,j_1,j_2\in[N]}|\hat{\omega} - \omega|=O_P (\sqrt{Th\log N} ).
\end{eqnarray*}
Therefore, we have 
\begin{eqnarray*}
    \max_{i,j\in[N^2p]}| \mathcal{B}_{1,ij}|&=&O_P (\sqrt{h\log N/T} ), 
\end{eqnarray*}
where $ \mathcal{B}_{1,ij}$ is the $(i,j)$-th element of $\pmb{\mathcal{B}}_1$ in \eqref{mm40}. Define $ \mathcal{B}_{2,ij}$, $ \mathcal{B}_{3,ij}$ and $ \mathcal{B}_{4,ij}$ analogously.

For $ \mathcal{B}_{2,ij}$, we can use  Cauchy-Schwarz inequality and Theorem \ref{TM.Main2}  to obtain
\begin{eqnarray*}
   \max_{i,j\in[N^2p]}| \mathcal{B}_{2,ij}|&\leq& O(h)\left\|\pmb{\beta}_0-\hat{\pmb{\beta}}_\phi\right\|_1^2
   \notag\\
   &=&O_P \big( \mu_{h}^{2+2/J}T^{-1}h d_{\mathscr{A}}^2\log N\big)
   \notag\\
   &=&O_P (\sqrt{h\log N/T} ),
\end{eqnarray*} 
under the condition $\mu_{h}^{2+2/J}d_{\mathscr{A}}^2\sqrt{h\log N}/\sqrt{T}\rightarrow 0$.  

Analogously, we can show $\max_{i,j\in[N^2p]}| \mathcal{B}_{3,ij}|$ and  $\max_{i,j\in[N^2p]}| \mathcal{B}_{4,ij}|$ are also bounded in probability in the desired order. It completes the proof of Theorem \ref{TM.Main4}.(1).

\medskip

(2). For $\hat{\pmb{\Omega}}_h$, we write
\begin{eqnarray}\label{mm33}
     \|\mathscr{G}_\eta(\hat{\pmb{\Omega}}_h)-E[\tilde{\pmb{\Omega}}_h] \|_2&\leq& \|\mathscr{G}_\eta(E[\tilde{\pmb{\Omega}}_h])-E[\tilde{\pmb{\Omega}}_h] \|_2+ \|\mathscr{G}_\eta(\hat{\pmb{\Omega}}_h)-\mathscr{G}_\eta(E[\tilde{\pmb{\Omega}}_h]) \|_2.
\end{eqnarray}
For the first term,

\begin{eqnarray}\label{mm39}
    \|\mathscr{G}_\eta(E[\tilde{\pmb{\Omega}}_h])-E[\tilde{\pmb{\Omega}}_h] \|_2&\leq& \|\mathscr{G}_\eta(E[\tilde{\pmb{\Omega}}_h])-E[\tilde{\pmb{\Omega}}_h] \|_\infty
    \notag\\
    &\leq& \max_{i\in[N^2p]} \sum_{j=1}^{N^2p}|\Omega_{ij}|I(|\Omega_{ij}| < \eta)
    \notag\\
    &\leq&\eta^{1-c_a}\max_{i\in[N^2p]} \sum_{j=1}^{N^2p}|\Omega_{ij}|^{c_a}
    \notag\\
    &=&\eta^{1-c_a}c_{N}, 
\end{eqnarray}
where the first inequality follows from the facts that $\|\mathbf{A}\|_2\le \sqrt{\|\mathbf{A}\|_1\|\mathbf{A}\|_\infty}$ and $\mathscr{G}_\eta(E[\tilde{\pmb{\Omega}}_h])-E[\tilde{\pmb{\Omega}}_h]$ is symmetric, $\Omega_{ij}$ is the $(i,j)^{th}$ element of $E[\tilde{\pmb{\Omega}}_h]$, and $c_{N}=\max_{i\in[N^2p]} \sum_{j=1}^{N^2p}|\Omega_{ij}|^{c_a}$ is the same as that defined in \eqref{def.dense}. 

We now study the second term on the right-hand side of \eqref{mm33}. Let $\hat{\Omega}_{ij}$ be the $(i,j)$-th element of $\hat{\pmb{\Omega}}_h$. We write
\begin{eqnarray}\label{mm36}
   \|\mathscr{G}_\eta(\hat{\pmb{\Omega}}_h)-\mathscr{G}_\eta(E[\tilde{\pmb{\Omega}}_h])\|_2 &\leq& \|\mathscr{G}_\eta(\hat{\pmb{\Omega}}_h)-\mathscr{G}_\eta(E[\tilde{\pmb{\Omega}}_h])\|_\infty
   \notag\\
   &\leq&\max_{i\in[N^2p]} \sum_{j=1}^{N^2p}|\hat{\Omega}_{ij}-\Omega_{ij}|I(|\hat{\Omega}_{ij}| \geq  \eta, |\Omega_{ij}| \geq \eta) 
   \notag\\
   &&+\max_{i\in[N^2p]} \sum_{j=1}^{N^2p}|\hat{\Omega}_{ij}|I(|\hat{\Omega}_{ij}| \geq  \eta, |\Omega_{ij}| < \eta)
   \notag\\
   &&+\max_{i\in[N^2p]} \sum_{j=1}^{N^2p}|\Omega_{ij}|I(|\hat{\Omega}_{ij}| <  \eta, |\Omega_{ij}| \geq \eta),
\end{eqnarray}
where the first inequality follows from the facts that $\|\mathbf{A}\|_2\le \sqrt{\|\mathbf{A}\|_1\|\mathbf{A}\|_\infty}$ and $\mathscr{G}_\eta(\hat{\pmb{\Omega}}_h)-\mathscr{G}_\eta(E[\tilde{\pmb{\Omega}}_h])$ is symmetric.

For the first term, since $I(|\hat{\Omega}_{ij}| \geq  \eta, |\Omega_{ij}| \geq \eta)\leq |\Omega_{ij}|^{c_a}\eta^{-c_a}$,
\begin{eqnarray}\label{mm35}
   \max_{i\in[N^2p]} \sum_{j=1}^{N^2p}|\hat{\Omega}_{ij}-\Omega_{ij}|I(|\hat{\Omega}_{ij}| \geq  \eta, |\Omega_{ij}| \geq \eta) &\leq& \max_{i,j\in[N^2p]}|\hat{\Omega}_{ij}-\Omega_{ij}|\sum_{j=1}^{N^2p}|\Omega_{ij}|^{c_a}\eta^{-c_a}
   \notag\\
   &=&O_P(\eta^{1-c_a}c_N),
\end{eqnarray}
where the equality holds by Theorem \ref{TM.Main4}.1.

For the second term in \eqref{mm36},

\begin{eqnarray}\label{mm34}
    &&\max_{i\in[N^2p]} \sum_{j=1}^{N^2p}|\hat{\Omega}_{ij}|I(|\hat{\Omega}_{ij}| \geq  \eta, |\Omega_{ij}| < \eta)
    \notag\\
    &\leq& \max_{i\in[N^2p]} \sum_{j=1}^{N^2p}|\hat{\Omega}_{ij}-\Omega_{ij}|I(|\hat{\Omega}_{ij}| \geq  \eta, |\Omega_{ij}| < \eta)+\eta^{1-c_a}c_{N}
    \notag\\
    &\leq& \max_{i\in[N^2p]} \sum_{j=1}^{N^2p}|\hat{\Omega}_{ij}-\Omega_{ij}|I(|\hat{\Omega}_{ij}| \geq  \eta, |\Omega_{ij}| < c_\eta\eta)
    \notag\\
    &&+ \max_{i\in[N^2p]} \sum_{j=1}^{N^2p}|\hat{\Omega}_{ij}-\Omega_{ij}|I(|\hat{\Omega}_{ij}| \geq  \eta, c_\eta\eta<|\Omega_{ij}| < \eta)+\eta^{1-c_a}c_{N},
\end{eqnarray}
where $0<c_\eta<1$ is a constant.

Additionally, we can further write
\begin{eqnarray*}
   &&\max_{i\in[N^2p]} \sum_{j=1}^{N^2p}|\hat{\Omega}_{ij}-\Omega_{ij}|I(|\hat{\Omega}_{ij}| \geq  \eta, |\Omega_{ij}| < c_\eta\eta) 
   \notag\\
   &\leq&  \max_{i,j\in[N^2p]}|\hat{\Omega}_{ij}-\Omega_{ij}| \max_{i\in[N^2p]} \sum_{j=1}^{N^2p} I(|\hat{\Omega}_{ij}-\Omega_{ij}| \geq  (1-c_\eta)\eta),
\end{eqnarray*}
which implies that this term is asymptotically negligible.

For the second term in \eqref{mm34}, similarly to \eqref{mm35}, 
\begin{eqnarray*}
&&\max_{i\in[N^2p]} \sum_{j=1}^{N^2p}|\hat{\Omega}_{ij}-\Omega_{ij}|I(|\hat{\Omega}_{ij}| \geq  \eta, c_\eta\eta<|\Omega_{ij}| < \eta) \notag \\
&\leq& \max_{i,j\in[N^2p]}|\hat{\Omega}_{ij}-\Omega_{ij}|\sum_{j=1}^{N^2p}|\Omega_{ij}|^{c_a}\eta^{-c_a}c_\eta^{-c_a} =O_P (\eta^{1-c_a}c_N ).
\end{eqnarray*}

Therefore, we can establish the following result for the second term in \eqref{mm36}
\begin{eqnarray}\label{mm37}
    \max_{i\in[N^2p]} \sum_{j=1}^{N^2p}|\hat{\Omega}_{ij}|I(|\hat{\Omega}_{ij}| \geq  \eta, |\Omega_{ij}| < \eta)=O_P (\eta^{1-c_a}c_N ).
\end{eqnarray}

We now proceed to study the third term in \eqref{mm36} and write
\begin{eqnarray}\label{mm38}
   &&\max_{i\in[N^2p]} \sum_{j=1}^{N^2p}|\Omega_{ij}|I(|\hat{\Omega}_{ij}| <  \eta, |\Omega_{ij}| \geq \eta) \notag \\
   &\leq&  \max_{i\in[N^2p]} \sum_{j=1}^{N^2p}|\hat{\Omega}_{ij}-\Omega_{ij}|I(|\hat{\Omega}_{ij}| <  \eta, |\Omega_{ij}| \geq \eta)
   \notag\\
   &&+ \max_{i\in[N^2p]} \sum_{j=1}^{N^2p}|\hat{\Omega}_{ij}|I(|\hat{\Omega}_{ij}| <  \eta, |\Omega_{ij}| \geq \eta)
   \notag\\
   &\leq&  \Big(\max_{i,j\in[N^2p]}|\hat{\Omega}_{ij}-\Omega_{ij}|+\eta \Big) \max_{i\in[N^2p]} \sum_{j=1}^{N^2p}|\Omega_{ij}|^{c_a}\eta^{-c_a}
   \notag\\
   &=&O_P (\eta^{1-c_a}c_N ).
\end{eqnarray}

Combining the results in \eqref{mm35}, \eqref{mm37}, and \eqref{mm38}, we have
\begin{eqnarray*}
   \|\mathscr{G}_\eta(\hat{\pmb{\Omega}}_h)-\mathscr{G}_\eta(E[\tilde{\pmb{\Omega}}_h]) \|_2 &=& O_P\left(\eta^{1-c_a}c_N\right).
\end{eqnarray*}
Together with \eqref{mm39}, it leads to the desired result in Theorem \ref{TM.Main4}.(2).
\end{proof}

\medskip

\begin{proof}[Proof of Theorem \ref{TM.Main5}]
\item 

Let $\pmb{\Omega}_z=\pmb{\Sigma}_{\mathbf{B}}^{-1} \otimes \mathbf{I}_N$. We first show that 
\begin{eqnarray}\label{node8}
    \sqrt{T}\,\pmb{\rho}_{\mathbf{a}}^\top\VEC(\hat{\tilde{\mathbf{a}}}_{\rm bc}-\tilde{\mathbf{A}})&\xrightarrow{~D~}&\mathcal{N}(0,\pmb{\rho}_{\mathbf{a}}^\top\pmb{\Omega}_z\pmb{\Omega}_h\pmb{\Omega}_z\pmb{\rho}_{\mathbf{a}}).
\end{eqnarray}
We can observe that the debiased estimator admits the following expansion

\begin{eqnarray*}
\pmb{\rho}_{\mathbf{a}}^\top\VEC(\hat{\tilde{\mathbf{a}}}_{\rm bc}-\tilde{\mathbf{A}})&=&\pmb{\rho}_{\mathbf{a}}^\top\Big(\mathbf{I}-\frac{1}{T-h}\hat{\pmb{\Omega}}_z\sum_{t=1}^{T-h}\mathbf{Z}_t\mathbf{Z}_t^\top \Big)\VEC(\hat{\tilde{\mathbf{a}}}-\tilde{\mathbf{A}})+\frac{1}{T-h}\pmb{\rho}_{\mathbf{a}}^\top\hat{\pmb{\Omega}}_z\sum_{t=1}^{T-h}\mathbf{Z}_t\mathbf{u}_{t,h}
\notag\\
&=&\mathcal{G}_1+\mathcal{G}_2.
\end{eqnarray*}

For $\mathcal{G}_1$, we first investigate $\big\|\frac{1}{T-h}\sum_{t=1}^{T-h}\mathbf{Z}_t\mathbf{Z}_t^\top\hat{\pmb{\Omega}}_{z,i}-\mathbf{e}_i \big\|_\infty$, where $\mathbf{e}_i$  is the $i$-th column of the identity matrix and  $\hat{\pmb{\Omega}}_{z,i}=\hat{\tau}_i^{-2}\widehat{\mathbf{C}}_{z,i}$ is the $i$-th column of $\hat{\pmb{\Omega}}_{z}^\top$. 

Using simple algebra, we can show that
\begin{eqnarray*}
    \frac{1}{T-h}\sum_{t=1}^{T-h}\mathbf{Z}_{j,t}\mathbf{Z}_t^\top\hat{\pmb{\Omega}}_{z,i}&=&\frac{1}{(T-h)\hat{\tau}_i^2}\sum_{t=1}^{T-h}\mathbf{Z}_{j,t}\big(\mathbf{Z}_{i,t}^\top-\mathbf{Z}_{-i,t}^\top\hat{\mathbf{b}}_i\big),
\end{eqnarray*}
for $j\in[N^2p]$.

For notational simplicity, let $\hat{\mathbf{r}}_{i,t}=\mathbf{Z}_{i,t}-\hat{\mathbf{b}}_i^\top\mathbf{Z}_{-i,t}$. For $j=i$, we write
\begin{eqnarray}\label{node1}
   \frac{1}{T-h} \sum_{t=1}^{T-h}\mathbf{Z}_{i,t}\big(\mathbf{Z}_{i,t}^\top-\mathbf{Z}_{-i,t}^\top\hat{\mathbf{b}}_i\big)&=&
   \frac{1}{T-h}\sum_{t=1}^{T-h}\hat{\mathbf{r}}_{it}\hat{\mathbf{r}}_{it}^\top + \frac{1}{T-h}\sum_{t=1}^{T-h}\hat{\mathbf{b}}_i^\top\mathbf{Z}_{-i,t}\hat{\mathbf{r}}_{it}^\top.
\end{eqnarray}
Noteworthily, $\hat{\mathbf{b}}_i$ is a solution of \eqref{def.node} if and only if there exits a subgradient
\begin{eqnarray*}
    \tilde{\mathbf{g}}_{\mathbf{b},i} &\in &\big\{\mathbf{z}\in \mathbb{R}^{N^2p-1} \mid  z_j = \sgn (\hat{\mathbf{b}}_{i,j} )  \text{ for } \sgn (\hat{\mathbf{b}}_{i,j} )\ne 0,\text{ and }\notag \\
    &&|z_j|\le 1 \text{ elsewhere}, \text{ for } j \in[N^2p-1]\big\}
\end{eqnarray*} 
such that 
\begin{eqnarray}\label{node2}
    \mathbf{0}=\tilde{\gamma}_i \tilde{\mathbf{g}}_{\mathbf{b},i}- \frac{1}{T-h}\sum_{t=1}^{T-h}\mathbf{Z}_{-i,t}\hat{\mathbf{r}}_{i,t}^\top.
\end{eqnarray}
By \eqref{node1} and \eqref{node2}, we obtain
\begin{eqnarray}\label{node3}
   \frac{1}{T-h} \sum_{t=1}^{T-h}\mathbf{Z}_{i,t}\big(\mathbf{Z}_{i,t}^\top-\mathbf{Z}_{-i,t}^\top\hat{\mathbf{b}}_i\big)
   &=&\frac{1}{T-h}\sum_{t=1}^{T-h}\|\hat{\mathbf{r}}_{it}\|^2+\tilde{\gamma}_i \|\hat{\mathbf{b}}_i\|_1
   =\hat{\tau}_i^2,
\end{eqnarray}
which immediately yields that $\frac{1}{T-h}\sum_{t=1}^{T-h}\mathbf{Z}_{i,t}\mathbf{Z}_t^\top\hat{\pmb{\Omega}}_{z,i}=1$. For $j\neq i$, using \eqref{node2},
\begin{eqnarray}\label{node4}
    \left|\frac{1}{T-h}\sum_{t=1}^{T-h}\mathbf{Z}_{j,t}\mathbf{Z}_t^\top\hat{\pmb{\Omega}}_{z,i}\right|&=& \left|\frac{1}{(T-h)\hat{\tau}_i^2}\sum_{t=1}^{T-h}\mathbf{Z}_{j,t}\hat{\mathbf{r}}_{i,t}^\top\right|
    \leq\frac{\tilde{\gamma}_i }{\hat{\tau}_i^2}.
\end{eqnarray}
In summary of \eqref{node3} and \eqref{node4}, we have
\begin{eqnarray*}
    \max_i\left\|\frac{1}{T-h}\sum_{t=1}^{T-h}\mathbf{Z}_t\mathbf{Z}_t^\top\hat{\pmb{\Omega}}_{z,i}-\mathbf{e}_i \right\|_\infty=O_P\left( \frac{\mu_{h}^{1+1/J}\sqrt{\log (N)}}{\sqrt{T}}\right).
\end{eqnarray*}

Together with Theorem \ref{TM.Main2}, it gives 
\begin{eqnarray*}
    |\mathcal{G}_1|&=&\left|\pmb{\rho}_{\mathbf{a}}^\top\Big(\mathbf{I}-\frac{1}{T-h}\hat{\pmb{\Omega}}_z\sum_{t=1}^{T-h}\mathbf{Z}_t\mathbf{Z}_t^\top \Big)\VEC(\hat{\tilde{\mathbf{a}}}-\tilde{\mathbf{A}})\right|
    \notag\\
    &\leq&\|\pmb{\rho}_{\mathbf{a}}\|_1\, \left\|\VEC(\hat{\tilde{\mathbf{a}}}-\tilde{\mathbf{A}})\right\|_1\max_i\left\|\frac{1}{T-h}\sum_{t=1}^{T-h}\mathbf{Z}_t\mathbf{Z}_t^\top\hat{\pmb{\Omega}}_{z,i}-\mathbf{e}_i \right\|_\infty
    \notag\\
    &=&O_P\left( \frac{\mu_{h}^{2+2/J}d_{\mathscr{A}}\log N}{T} \right).
\end{eqnarray*}

\medskip

We now study $\mathcal{G}_2$. We can further write

\begin{eqnarray*}
    \mathcal{G}_2&=&\frac{1}{T-h}\pmb{\rho}_{\mathbf{a}}^\top\pmb{\Omega}_z\sum_{t=1}^{T-h}\mathbf{Z}_t\mathbf{u}_{t,h}+\frac{1}{T-h}\pmb{\rho}_{\mathbf{a}}^\top(\widehat{\pmb{\Omega}}_z-\pmb{\Omega}_z)\sum_{t=1}^{T-h}\mathbf{Z}_t\mathbf{u}_{t,h}
    \notag\\
    &\eqqcolon&\mathcal{G}_{2,1}+\mathcal{G}_{2,2}.
\end{eqnarray*}

For $\mathcal{G}_{2,1}$, it is sufficient to investigate $\frac{1}{T-h}\sum_{t=1}^{T-h}u_{it,h}x_{jt}$ to formulate its leading-order term. Recall that we have defined  $U_{ij} (t,h)=u_{it,h} x_{jt}$ and $\mathscr{P}_{l}(\cdot) = E[\cdot \mid \mathscr{F}_{l}] -E[\cdot \mid \mathscr{F}_{l-1}]$, where $\mathscr{F}_l=\sigma(\pmb{\varepsilon}_l,\pmb{\varepsilon}_{l-1},\ldots)$. For notational simplicity, we suppress the indices in $U_{ij} (t,h)$ and write

\begin{eqnarray}\label{node5}
    \frac{1}{T-h}\sum_{t=1}^{T-h}U(t,h) &=& \frac{1}{T-h}\sum_{t=1}^{T-h}\sum_{l=0}^{m}\mathscr{P}_{t+h-l}(U(t,h))\notag \\
    &&+\frac{1}{T-h}\sum_{t=1}^{T-h}\sum_{l=m+1}^{\infty}\mathscr{P}_{t+h-l}(U(t,h)),
\end{eqnarray}
where $m$ satisfies that $m\rightarrow\infty$ and $m/T\rightarrow0$.

For the second term on the right-hand side of \eqref{node5}, 
\begin{eqnarray*}
    \mathscr{P}_{t+h-l}(U(t,h))&=& E[U(t,h) \mid \mathscr{F}_{t+h-l}] -E[U(t,h) \mid \mathscr{F}_{t+h-l-1}]
    \notag\\
    &=&E[U(t,h) -U_{t+h-l}^*(t,h)\mid \mathscr{F}_{t+h-l}],
\end{eqnarray*}
where $U_l^*(t,h)$ is a coupled versions of $U(t,h) $ by replacing $\pmb{\varepsilon}_l$ with an independent copy of $\pmb{\varepsilon}_l^*$. 

By  Jensen inequality, Burkholder inequality and Lemma \ref{LMB7}, 
\begin{eqnarray}\label{node7}
    &&\left\|\frac{1}{T-h}\sum_{t=1}^{T-h}\sum_{l=m+1}^{\infty}\mathscr{P}_{t+h-l}(U(t,h))\right\|_{2} \notag \\
    &\leq& \frac{1}{T-h}\sum_{l=m+1}^{\infty} \left\|\sum_{t=1}^{T-h}(U(t,h) -U_{t+h-l}^*(t,h))\right\|_{2}
    \notag\\
    &\leq& O\left(\frac{1}{T}\right)\sum_{l=m+1}^{\infty} \left(\sum_{t=1}^{T-h} \|U(t,h)-U_{t+h-l}^\ast (t,h)\|^{2}\right)^{1/2}
    \notag\\
    &=&O\left(T^{-1/2}\right)\sum_{l=m+1}^{\infty}c_{2,ij}(h,l)
    \notag\\
    &=&o\left(T^{-1/2}\right).
\end{eqnarray}

For the first term on the right-hand side of \eqref{node5}, it admits the following decomposition:
\begin{eqnarray}\label{node6}
    &&\frac{1}{T-h}\sum_{t=1}^{T-h}\sum_{l=0}^{m}\mathscr{P}_{t+h-l}(U(t,h))\notag \\
    &=&\frac{1}{T-h}\sum_{l=0}^{m}\sum_{t=1}^{T-h}\mathscr{P}_{t+h}(U(t+l,h))+\frac{1}{T-h}\sum_{l=0}^{m}\sum_{t=1-l}^{0}\mathscr{P}_{t+h}(U(t+l,h))
    \nonumber\\
    &&-\frac{1}{T-h}\sum_{l=0}^{m}\sum_{t=T-h-l+1}^{T-h}\mathscr{P}_{t+h}(U(t+l,h))
    \notag\\
    &=&\frac{1}{T-h}\sum_{l=0}^{m}\sum_{t=1}^{T-h}\mathscr{P}_{t+h}(U(t+l,h))+o_P(T^{-1/2}),
\end{eqnarray} 
where the second equality holds by directly applying  Lemma \ref{LMB7} and Burkholder inequality to the second and third terms on the right-hand side:
\begin{eqnarray*}
    &&\left\|\frac{1}{T-h}\sum_{l=0}^{m}\sum_{t=1-l}^{0}\mathscr{P}_{t+h}(U(t+l,h))\right\|_2\notag \\
    &\leq&O\left(\frac{1}{T}\right)\sum_{l=0}^{m}\left(\sum_{t=1-l}^{0}\|U(t+l,h)-U^\ast_{t+h}(t+l,h)\|^2\right)^{1/2}
    \notag\\
     &=&O\left(T^{-1}m^{1/2}\right)\sum_{l=0}^{m}c_{2,ij}(h,l)
      \notag\\
    &=&o\left(T^{-1/2}\right),
\end{eqnarray*}
and similarly, $\left\|\frac{1}{T-h}\sum_{l=0}^{m}\sum_{t=T-h-l+1}^{T-h}\mathscr{P}_{t+h}(U(t+l,h))\right\|_2=O\left(T^{-1}m^{1/2}\right)=o\left(T^{-1/2}\right)$.

For $ \mathcal{G}_{2,2}$, using basic matrix algebra, we can write $\pmb{\Omega}_{z}=(\pmb{\Omega}_{z,1},\cdots, \pmb{\Omega}_{z,N^2p})^\top$, where $\pmb{\Omega}_{z,i}=\tau_i^{-2}\mathbf{C}_{z,i}$ with \(\mathbf{C}_{z,i}^\top\) being constructed by inserting  1 in the \(i\)-th position and filling the remaining  entries with the elements of \(-\mathbf{b}_i\) in their  order, and  
\begin{eqnarray*}
    \tau_i^2=\Sigma_{z,i,i}-\pmb{\Sigma}_{z,i,-i}\pmb{\Sigma}_{z,-i,-i}^{-1}\pmb{\Sigma}_{z,-i,i}, \,\, \mathbf{b}_{i}=\pmb{\Sigma}_{z,-i,-i}^{-1}\pmb{\Sigma}_{z,-i,i}.
\end{eqnarray*}
Here $\Sigma_{z,i,i}$ is the $i$-th diagonal element of $\pmb{\Sigma}_{\mathbf{B}} \otimes \mathbf{I}_N$, $\pmb{\Sigma}_{z,-i,-i}$ is the submatrix constructed by removing the $i$-th row and column of $\pmb{\Sigma}_{\mathbf{B}} \otimes \mathbf{I}_N$. $\pmb{\Sigma}_{z,i,-i}$ and $\pmb{\Sigma}_{z,-i,i}$ denote the $i$-th row and  $i$-th column, respectively, in which the $i$-th element is removed.  
Using  arguments that are analogous to those in the proof of Theorem \ref{TM.Main2}.(1), we can show that  $\|\hat{\mathbf{b}}_{i}-\mathbf{b}_{i}\|_1=O_P \big( \frac{\mu_{h}^{1+1/J}s_{\mathscr{A},i}\sqrt{\log N}}{\sqrt{T}} \big)$, which implies that 
\begin{eqnarray}\label{node9}
    \left\|\sum_{i=1}^{N^2p}\rho_{i}(\widehat{\pmb{\Omega}}_{z,i}-\pmb{\Omega}_{z,i})^\top\right\|_1=O_P \left( \frac{\mu_{h}^{1+1/J}\sqrt{\log N}}{\sqrt{T}}\sum_{i=1}^{N^2p}\rho_i s_{\mathscr{A},i} \right).
\end{eqnarray} 
Together with \eqref{basis.ineq1}, it leads to the following result:
\begin{eqnarray*}
    |\mathcal{G}_{2,2}|&\leq&\left\|\sum_{i=1}^{N^2p}\rho_{i}(\widehat{\pmb{\Omega}}_{z,i}-\pmb{\Omega}_{z,i})^\top\right\|_1\left\|\frac{1}{T-h}\sum_{t=1}^{T-h}\mathbf{Z}_t\mathbf{u}_{t,h}\right\|_{\infty}
    \notag\\
    &=& o_P(T^{-1/2}).
\end{eqnarray*}
Together with \eqref{node5}, \eqref{node7} and \eqref{node6}, it gives 
\begin{eqnarray*}
    \mathcal{G}_2&=&\frac{1}{T-h}\pmb{\rho}_{\mathbf{a}}^\top\pmb{\Omega}_z\sum_{l=0}^{m}\sum_{t=1}^{T-h}\mathscr{P}_{t+h}(U(t+l,h))+o_P(T^{-1/2}).
\end{eqnarray*}
Since $\{\sum_{l=0}^{m}\mathscr{P}_{t+h}(U(t+l,h))\}_{t=1}^{T-h}$ is a martingale difference sequence adapted to $\mathscr{F}_{t+h}$, the standard CLT can be applied to this term, and it yields the desired result in \eqref{node8}.

By Theorem \ref{TM.Main4}.(2) and \eqref{node9}, we have $\pmb{\rho}_{\mathbf{a}}^\top\hat{\pmb{\Omega}}_z\mathscr{G}_\eta(\hat{\pmb{\Omega}}_h)\hat{\pmb{\Omega}}_z\pmb{\rho}_{\mathbf{a}}=\pmb{\rho}_{\mathbf{a}}^\top\pmb{\Omega}_z\pmb{\Omega}_h\pmb{\Omega}_z\pmb{\rho}_{\mathbf{a}}+o_P(1)$, which completes the proof of Theorem \ref{TM.Main5}. 
\end{proof}

\subsection{Proofs of the Preliminary Lemmas}

\begin{lemma}\label{LMB1}
Let $z_1,\ldots,z_T$ be mean zero independent random variables satisfying that $\max_{t\in [T]}\|z_t\|_J<\infty$ for some $J>2$. Then,

$$ 
\Pr\left(\max_{t\in [T]}\sum_{s=1}^{t}z_s \geq x \right) \leq \left(1 + \frac{2}{J} \right)^{J} \frac{\sum_{t=1}^{T}\|z_t\|_J^J}{x^J} + 2\exp\left(  \frac{- 2e^{-J}(J+2)^{-2} x^2}{\sum_{t=1}^{T}\|z_t\|_2^2}\right).
$$
\end{lemma}

Corollary 1.8 in \cite{nagaev1979large} shows that the above inequality holds for $\sum_{s=1}^{T}z_s$, while \cite{borovkov1973notes} proves that the above inequality also holds for $\max_{t\in [T]}\sum_{s=1}^{t}z_s$.

\begin{lemma}\label{LMB2}
Let $z_1,\ldots,z_T$ be mean zero independent random variables, and let $x > y$ and $y \geq (4\sum_{t=1}^{T}\|z_t\|_J^J)^{1/q}$. Suppose $\max_{t\in [T]}\|z_t\|_J<\infty$ for some $q \geq 2$, $\beta = J/(J+2)$ and $\alpha = 1-\beta$. Then,

\begin{eqnarray*}
    \Pr\left(\sum_{t=1}^{T}z_t \geq x \right) \leq \sum_{t=1}^{T}\Pr\left(z_t \geq y \right) + \exp\left(\frac{-\alpha^2x^2}{2e^{J}\sum_{t=1}^{T}\|z_t\|_2^2 }\right)+ \left(\frac{ \sum_{t=1}^{T}\|z_t\|_J^J}{\beta xy^{J-1}} \right)^{\beta x/y}.
\end{eqnarray*}
\end{lemma}

Lemma \ref{LMB2} is Corollary 1.7 of \cite{nagaev1979large}.

\begin{lemma}\label{LMB3}
Let $\{z_1,\ldots, z_n\}$ be independent zero-mean random variables. Suppose that for some positive real $L$ and every integer $k\ge 2$, $E[|z_i^k|]\le \frac{1}{2}E[z_i^2] L^{k-2} k!$. Then for $0\le \epsilon\le \frac{1}{2L}\sqrt{\sum_{i=1}^nE[z_i^2] }$,

\begin{eqnarray*}
\Pr\left(\frac{\sum_{i=1}^nz_i}{\sqrt{\sum_{i=1}^nE[z_i^2] }}\ge 2\epsilon\right)  \le \exp(-\epsilon^2)  .
\end{eqnarray*}

\end{lemma}

This lemma is from Corollary 2.1 of \cite{JSW_2023}.

\medskip

The Beveridge and Nelson (BN) decomposition under the HD setting is as follows:

\begin{eqnarray}\label{def.BN1}
\mathbf{B}(L) = \bar{\mathbf{B}} -(1-L)\tilde{\mathbf{B}}(L) ,
\end{eqnarray}
where $\tilde{\mathbf{B}}(L)\coloneqq\sum_{\ell=0}^{\infty}\tilde{\mathbf{B}}_\ell L^\ell$ with $\tilde{\mathbf{B}}_\ell\coloneqq \sum_{k=\ell+1}^{\infty}\mathbf{B}_k$. The expression of \eqref{def.BN1} is the so-called BN decomposition from \cite{PS1992}.

\begin{lemma}\label{LMB4}
\item Under Assumptions \ref{AS1} and \ref{AS.A1},
\begin{enumerate}[leftmargin=24pt, parsep=2pt, topsep=2pt]
\item $\sum_{s=1}^t \mathbf{x}_{s} = \bar{\mathbf{B}} \sum_{s=1}^t\pmb{\varepsilon}_{s}-\tilde{\mathbf{B}}(L)\pmb{\varepsilon}_{t}+\tilde{\mathbf{B}}(L)\pmb{\varepsilon}_{0}$, where $\|\tilde{\mathbf{B}}(1)\|_2<\infty$ and $\|\tilde{\mathbf{B}}(L)\pmb{\varepsilon}_{t}\|_2 =O_P(\sqrt{N})$ for $\forall t\ge 1$.
\item $\left\|\sum_{t=1}^T \pmb{\varepsilon}_t\right\|_2=O_P(\sqrt{NT})$.
\end{enumerate}
\end{lemma}

\begin{lemma}\label{LMB5}
    Suppose that $\max_{j\in [N]}\sum_{\ell = 0}^{\infty} \ell^{\frac{1}{2} - \frac{1}{2(J+1)}} \|\pmb{\beta}_{\ell, j} \|_2^{\frac{J}{J+1}}<\infty$ (i.e., Assumption \ref{AS2}.2). Then the following results hold:
    \begin{enumerate}[leftmargin=24pt, parsep=2pt, topsep=2pt]
    \item $\max_{j\in [N]}\sum_{\ell=0}^{\infty} (\ell^{\frac{J}{2}-1} \|\pmb{\beta}_{\ell, j} \|_2^J)^{\frac{1}{J+1}}<\infty$,
    \item  $\max_{j\in [N]} \sum_{\ell = 0}^{\infty}\ell \|\pmb{\beta}_{\ell, j} \|_2^2<\infty$.
    \end{enumerate}
\end{lemma}

\begin{lemma}\label{LMB6}
Under Assumptions \ref{AS1}-\ref{AS2}, suppose that $\frac{N^2 }{T^{J-1}\log N}\to 0$. The following results hold:

\begin{enumerate}[leftmargin=24pt, parsep=2pt, topsep=2pt]
    \item $\left\|\frac{1}{T-h} \sum_{t=1}^{T-h} \sum_{\ell = 0}^{\infty}  \mathbf{B}_\ell \pmb{\varepsilon}_{t-\ell}\pmb{\varepsilon}_{t-\ell}^\top \mathbf{B}_\ell^\top  -\sum_{\ell = 0}^{\infty}  \mathbf{B}_\ell \mathbf{B}_\ell^\top \right\|_{\max} =O_P\left(\frac{d_{\pmb{\beta}}^2\sqrt{\log N}}{\sqrt{T}}\right)$;
    
    \item $\left\|\frac{1}{T-h}\sum_{\ell=0}^{\infty}\sum_{j=1}^\infty \mathbf{e}_m^\top\mathbf{B}_\ell\pmb{\varepsilon}_{t-\ell}\pmb{\varepsilon}_{t-\ell-j}^\top \mathbf{B}_{\ell+j}^\top \mathbf{e}_n\right\|_{\max} =O_P (\frac{d_{\pmb{\beta}}^2\sqrt{\log N}}{\sqrt{T}})$;
    
    \item $\left\|\frac{1}{T-h} \sum_{t=1}^{T-h} \mathbf{x}_t\mathbf{x}_t^\top - \sum_{\ell = 0}^{\infty}  \mathbf{B}_\ell \mathbf{B}_\ell^\top \right\|_{\max} =O_P (\frac{d_{\pmb{\beta}}^2\sqrt{\log N}}{\sqrt{T}})$;
    
    \item $\left\|\frac{1}{T-h} \sum_{t=1}^{T-h} \mathbf{x}_t\mathbf{x}_{t-j}^\top - \sum_{\ell = 0}^{\infty}  \mathbf{B}_{\ell+j} \mathbf{B}_\ell^\top \right\|_{\max} =O_P (\frac{d_{\pmb{\beta}}^2\sqrt{\log N}}{\sqrt{T}})$ for $\forall j\in [p]$.
\end{enumerate}
\end{lemma}

\medskip

Recall that we have defined the following notations: 

\begin{eqnarray*}
    &&U_{ij} (t,h) \coloneqq u_{it,h} x_{jt},\notag\\
    && U_{ij} (t,h,m)\coloneqq E[u_{it,h} x_{jt} \mid \mathscr{F}_{t+h,t+h-m}], \notag\\
    &&S_{ij, h}(t,m)\coloneqq \sum_{t_h< k< t }U_{ij} (k,h,m).\notag
\end{eqnarray*}
We let $U_{ij,l}^*(t,h)$, $U_{ij,l}^*(t,h,m)$,  $ S_{ij,h,l}^*(t)$, and $ S_{ij,h,l}^*(t,m)$ be coupled versions of $U_{ij}(t,h) $, $U_{ij}(t,h,m)$, $S_{ij,h}(t)$,  and $S_{ij,h}(t,m)$, respectively,  by replacing $\pmb{\varepsilon}_l$ with an independent copy of $\pmb{\varepsilon}_l^*$.

\begin{lemma}\label{LMB7}
Under Assumptions \ref{AS1}-\ref{AS4}, the following results hold for $\forall h>1$

\begin{enumerate}[leftmargin=24pt, parsep=2pt, topsep=2pt]
    \item $\|U_{ij} (t,h)-U_{ij, t+h-l}^\ast (t,h)\|_J \lesssim  c_{J,ij}(h,l) $, where 
    \begin{eqnarray*}
        c_{J,ij}(h,l)=\left\{\begin{array}{ll}
            \delta_{J,i}(h,l)\left(\sum_{i=1}^N\sum_{\ell=0}^{\infty} |b_{\ell,ji}|^{2}\right)^{\frac{1}{2}} & \text{for }l =1,\ldots, h-1 \\
             \left(\sum_{i=1}^N |b_{l-h,ji}|^{2}\right)^{\frac{1}{2}}&\text{for } l\ge h
        \end{array} \right.;
    \end{eqnarray*} 
    \item    $\|U_{ij} (t,h)-U_{ij}(t,h,m)\|_J\lesssim d_{J,ji}(h,m)$, where $d_{J,ij}(h,m)=\left(\sum_{l=m+1}^{\infty}  c^{2}_{J,ij}(h,l)\right)^{\frac{1}{2}}$;
    \item $\|\sum_{t=1}^{T-h}(U_{ij} (t,h)-U_{ij}(t,h,m))\|_J \lesssim T^{1/2} a_{J,ji}(h,m)$ and $\|\sum_{t=1}^{T-h}U_{ij} (t,h)\|_J \lesssim T^{1/2} $, where  $a_{J,ji}(h,m) =\sum_{l=m+1}^{\infty}  c_{J,ij}(h,l)$.
\end{enumerate}
\end{lemma}

Define 
\begin{eqnarray*}
&&P_{h}(T, m,\epsilon)
\notag\\
&=&\sup_{|a_{s,t}|<1}\Pr\Big(\Big|\sum_{t=2}^{T-h}\sum_{1\leq s< t} a_{s,t} \big(U_{i_1j_1} (t,h,m) U_{i_2j_2}(s,h,m)-E[U_{i_1j_1} (t,h,m)  U_{i_2j_2}(s,h,m)]\big)\Big|\ge \epsilon \Big).  
\end{eqnarray*}

\begin{lemma}\label{LMB8}
Under Assumptions \ref{AS1}-\ref{AS4},  the following results hold.

\begin{enumerate}[leftmargin=24pt, parsep=2pt, topsep=2pt]
    \item  Suppose that $\max_{s,t}|a_{s,t}|=O(1)$, $m=\lfloor T^\theta \rfloor$,  $T^{1+c_\epsilon}/\epsilon_{NT}\rightarrow 0$,  $l_{T,\theta}=\lceil-\log\log T/\log\theta \rceil$, and $J>4$, where $0<\theta<1$ and $c_\epsilon>0$, we have  
    \[P_{h}(T, m,\epsilon_{NT})=O\left(\epsilon_{NT}^{-J}T\left(m^{(1-c_0\theta)J/2-1}+l_{T,\theta}^{J}\right)\right).\]
    
    \item  Suppose that $J>4$ and there exist constants $0<\theta<1$, $c_\epsilon>0$ and $c_h>0$ such that $T^{\theta^{c_h}}/h\rightarrow 0$ and $T^{\theta^{c_h}(1+c_\epsilon)}/\epsilon_{NT}\rightarrow 0$,  we have
\begin{eqnarray*}
  &&\Pr\Big(\Big|\sum_{t=2}^{T-h}\big(U_{i_1j_1} (t,h) S_{i_2j_2,h}(t)-E[U_{i_1j_1} (t,h)  S_{i_2j_2,h}(t)]\big)\Big|\ge \epsilon_{NT} \Big)
  \notag\\
&=&  O\left(\epsilon_{NT}^{-J/2}\log T\left(  T^{(1/2-c_0\theta) J/2}h^{J/4}
+ Th^{(1-c_0\theta)J/2-1}+T\right)\right).
\end{eqnarray*}
\end{enumerate}
\end{lemma}

\begin{proof}[Proof of Lemma \ref{LMB4}]
\item 

(1). Using \eqref{def.BN1}, we write

\begin{eqnarray*}
\sum_{s=1}^t \mathbf{x}_{s} &=&  \sum_{s=1}^t [\bar{\mathbf{B}} -(1-L)\tilde{\mathbf{B}}(L)] \pmb{\varepsilon}_{s} \notag \\
&=&\bar{\mathbf{B}} \sum_{s=1}^t\pmb{\varepsilon}_{s} - \sum_{s=1}^t\tilde{\mathbf{B}}(L)\pmb{\varepsilon}_{s}+ \sum_{s=1}^t \tilde{\mathbf{B}}(L)\pmb{\varepsilon}_{s-1}\notag \\
&=&\bar{\mathbf{B}} \sum_{s=1}^t\pmb{\varepsilon}_{s}-\tilde{\mathbf{B}}(L)\pmb{\varepsilon}_{t}+\tilde{\mathbf{B}}(L)\pmb{\varepsilon}_{0}.
\end{eqnarray*}

By Assumption \ref{AS.A1},

\begin{eqnarray*} 
\sum_{\ell=0}^{\infty} \|\tilde{\mathbf{B}}_\ell\|_2\le \sum_{\ell=0}^{\infty}\sum_{k=\ell+1}^{\infty}\| \mathbf{B}_k\|_2=\sum_{\ell=1}^{\infty}\ell \| \mathbf{B}_\ell\|_2<\infty.
\end{eqnarray*}
Thus,

\begin{eqnarray*}
\|\tilde{\mathbf{B}}(L)\pmb{\varepsilon}_{t}\|_2 =O_P(\sqrt{N}).
\end{eqnarray*}

\medskip

(2). Write

\begin{eqnarray*}
    E\left\|\sum_{t=1}^T \pmb{\varepsilon}_t\right\|^2 = \sum_{i=1}^N \sum_{t,s=1}^T E[\varepsilon_{it}\varepsilon_{is}]= \sum_{i=1}^N \sum_{t=1}^T E[\varepsilon_{it}^2] =NT,
\end{eqnarray*}
where the last line follows from Assumption \ref{AS1}.
\end{proof}

\medskip

\begin{proof}[Proof of Lemma \ref{LMB5}]
\item 

(1). It is easy to see that 

\begin{eqnarray*}
    &&\max_{j\in [N]}\sum_{\ell=0}^{\infty} (\ell^{\frac{J}{2}-1} \|\pmb{\beta}_{\ell, j} \|_2^J)^{\frac{1}{J+1}} =\sum_{\ell = 0}^{\infty} \ell^{\frac{1}{2} - \frac{3}{2(J+1)}} \|\pmb{\beta}_{\ell, j} \|_2^{\frac{J}{J+1}}\le \sum_{\ell = 0}^{\infty} \ell^{\frac{1}{2} - \frac{1}{2(J+1)}} \|\pmb{\beta}_{\ell, j} \|_2^{\frac{J}{J+1}}<\infty.
\end{eqnarray*}

\medskip

\noindent (2). For the second results, we write

\begin{eqnarray*}
    \sum_{\ell = 0}^{\infty}\ell \|\pmb{\beta}_{\ell, j} \|_2^2\le \{\sum_{\ell = 0}^{\infty}(\ell \|\pmb{\beta}_{\ell, j} \|_2^2)^{\frac{J}{2(J+1)}}\}^{\frac{2(J+1)}{J}} = \{\sum_{\ell = 0}^{\infty} \ell^{\frac{1}{2} - \frac{1}{2(J+1)}} \|\pmb{\beta}_{\ell, j} \|_2^{\frac{J}{J+1}}\}^{\frac{2(J+1)}{J}} <\infty.
\end{eqnarray*}
The proof is now completed.
\end{proof}

\medskip

\begin{proof}[Proof of Lemma \ref{LMB6}]
\item 

Before proceeding, we define a few notations. Recall that $\mathbf{e}_n$ is a selection vector, and decompose $\mathbf{e}_m^\top\mathbf{x}_t\mathbf{x}_t^\top \mathbf{e}_n$ as follows:

\begin{eqnarray*} 
    \mathbf{e}_m^\top\mathbf{x}_t\mathbf{x}_t^\top \mathbf{e}_n &=& \sum_{\ell = 0}^{\infty} \mathbf{e}_m^\top\mathbf{B}_\ell \pmb{\varepsilon}_{t-\ell}\pmb{\varepsilon}_{t-\ell}^\top \mathbf{B}_\ell^\top \mathbf{e}_n + \sum_{\ell=0}^{\infty}\sum_{j=1}^\infty \mathbf{e}_m^\top\mathbf{B}_\ell\pmb{\varepsilon}_{t-\ell}\pmb{\varepsilon}_{t-\ell-j}^\top \mathbf{B}_{\ell+j}^\top \mathbf{e}_n \notag \\
    &&+ \sum_{\ell=1}^{\infty}\sum_{j=0}^\infty \mathbf{e}_m^\top\mathbf{B}_{\ell+j}\pmb{\varepsilon}_{t-j-\ell}\pmb{\varepsilon}_{t-j}^\top \mathbf{B}_j^\top \mathbf{e}_n\notag \\
    &=& \sum_{\ell = 0}^{\infty} [ (\mathbf{e}_n^\top\mathbf{B}_\ell) \otimes(\mathbf{e}_m^\top \mathbf{B}_\ell) ]\text{vec}(\pmb{\varepsilon}_{t-\ell}\pmb{\varepsilon}_{t-\ell}^\top )\notag \\
    &&+\sum_{\ell=0}^{\infty}\sum_{j=1}^\infty  [ (\mathbf{e}_n^\top\mathbf{B}_{\ell+j}) \otimes (\mathbf{e}_m^\top\mathbf{B}_\ell) ]\text{vec}(\pmb{\varepsilon}_{t-\ell}\pmb{\varepsilon}_{t-\ell-j}^\top)\notag \\
    &&+ \sum_{\ell=1}^{\infty}\sum_{j=0}^\infty [ (\mathbf{e}_n^\top\mathbf{B}_j) \otimes (\mathbf{e}_m^\top\mathbf{B}_{\ell+j}) ]\text{vec}(\pmb{\varepsilon}_{t-j-\ell}\pmb{\varepsilon}_{t-j}^\top)\notag \\
    &\eqqcolon& \mathbf{U}(L)\tilde{\pmb{\varepsilon}}_t + \sum_{j=1}^{\infty}\mathbf{U}_{j\centerdot}(L)\tilde{\pmb{\varepsilon}}_{t, \centerdot j} + \sum_{\ell=1}^{\infty}\mathbf{U}_{\centerdot\ell}(L) \tilde{\pmb{\varepsilon}}_{t,\ell\centerdot},
\end{eqnarray*}
where the definitions of the variables on the right hand side are obvious, and the indices $m$ and $n$ have been suppressed for notational simplicity. 

In what follows, we often use the following result:

\begin{eqnarray}\label{ineq.dbeta}
    \left(\sum_{\ell=0}^{\infty}\|\mathbf{B}_{\ell}\|_\infty^2\right)^{1/2}\le \sum_{\ell=0}^{\infty}\|\mathbf{B}_{\ell}\|_\infty= d_{\pmb{\beta}},
\end{eqnarray}
where the inequality follows from the norm monotonicity, and $d_{\pmb{\beta}}$ is defined in Assumption \ref{AS2} already. We are now ready to start the proof.

\medskip

(1). Similar to \eqref{def.BN1}, we have

\begin{eqnarray*}
    \mathbf{U}(L) = \bar{\mathbf{U}} -(1-L)\tilde{\mathbf{U}}(L)
\end{eqnarray*}
where $\bar{\mathbf{U}}\coloneqq \sum_{\ell = 0}^{\infty}\mathbf{U}_{\ell}$, and $\tilde{\mathbf{U}}(L)\coloneqq\sum_{\ell=0}^{\infty}\tilde{\mathbf{U}}_\ell L^\ell$ with $\tilde{\mathbf{U}}_\ell\coloneqq \sum_{k=\ell+1}^{\infty}\mathbf{U}_k$. Thus,

\begin{eqnarray*}
    \sum_{t=1}^{T-h}\mathbf{U}(L)\tilde{\pmb{\varepsilon}}_t = \bar{\mathbf{U}} \sum_{t=1}^{T-h}\tilde{\pmb{\varepsilon}}_t -\tilde{\mathbf{U}}(L)\tilde{\pmb{\varepsilon}}_{T-h} +\tilde{\mathbf{U}}(L)\tilde{\pmb{\varepsilon}}_0.
\end{eqnarray*}

Firstly, we study $\tilde{\mathbf{U}}(L)\tilde{\pmb{\varepsilon}}_{T-h}$ and $\tilde{\mathbf{U}}(L)\tilde{\pmb{\varepsilon}}_0$, and note that for $\forall t$

\begin{eqnarray*}
    &&\max_{m,n}\frac{1}{T-h}|\tilde{\mathbf{U}}(L)\tilde{\pmb{\varepsilon}}_t | \notag \\
    &\le &\max_{m,n}\frac{1}{T-h}\sum_{\ell = 0}^{\infty}\sum_{k=\ell+1}^{\infty}| \mathbf{e}_m^\top\mathbf{B}_k \pmb{\varepsilon}_{t-\ell}\pmb{\varepsilon}_{t-\ell}^\top \mathbf{B}_k^\top \mathbf{e}_n |\notag \\
    &\le & \frac{1}{T-h}\max_{m} \sum_{\ell = 0}^{\infty}\sum_{k=\ell+1}^{\infty} |\pmb{\beta}_{k, m}^\top \pmb{\varepsilon}_{t-\ell}|^2  \notag \\
    &\lesssim &\frac{1}{T-h}\max_{m} \sum_{\ell = 0}^{\infty}\ell \|\pmb{\beta}_{\ell, m} \|_2^2\notag \\
    &=&\frac{1}{T-h}, 
\end{eqnarray*}
where the second inequality follows from  Cauchy-Schwarz inequality, and  the last line follows from Lemma \ref{LMB5}. Therefore,

\begin{eqnarray}\label{rate.sum1}
    \max_{m,n}\frac{1}{T-h}\left|\tilde{\mathbf{U}}(L)\tilde{\pmb{\varepsilon}}_{T-h} -\tilde{\mathbf{U}}(L)\tilde{\pmb{\varepsilon}}_0 \right| =O_P(1)\frac{1}{T-h}.
\end{eqnarray}

\medskip

We now focus on 
 
\begin{eqnarray*}
    \frac{1}{T-h} \sum_{t=1}^{T-h} \bar{\mathbf{U}}\tilde{\pmb{\varepsilon}}_t  =\sum_{\ell=0}^\infty \pmb{\beta}_{\ell, m}^\top \hat{\pmb{\Sigma}} \pmb{\beta}_{\ell, n} ,
\end{eqnarray*}
where $\hat{\pmb{\Sigma}}=\{ \hat{\sigma}_{ij} \}\coloneqq\frac{1}{T-h} \sum_{t=1}^{T-h}  \pmb{\varepsilon}_{t} \pmb{\varepsilon}_{t}^\top .$ Note that 

\begin{eqnarray}\label{rate.sum2}
    \left|\sum_{\ell=0}^\infty \pmb{\beta}_{\ell, m}^\top (\hat{\pmb{\Sigma}}-\mathbf{I}_N) \pmb{\beta}_{\ell, n} \right| &\le &\|\hat{\pmb{\Sigma}}-\mathbf{I}_N\|_{\max}\sum_{\ell=0}^\infty \|\pmb{\beta}_{\ell, m}\|_1 \|\pmb{\beta}_{\ell, n}\|_1\notag \\
    &\le &\|\hat{\pmb{\Sigma}}-\mathbf{I}_N\|_{\max} \sum_{\ell=0}^\infty \|\mathbf{B}_\ell\|_{\infty}^2\notag\\
    &\le & \|\hat{\pmb{\Sigma}}-\mathbf{I}_N\|_{\max}d_{\pmb{\beta}}^2,
\end{eqnarray}
where the second inequality follows from Cauchy-Schwarz inequality, and the third inequality follows from \eqref{ineq.dbeta}. We study $\|\hat{\pmb{\Sigma}}-\mathbf{I}_N\|_{\max}$ in what follows. 

For notational simplicity, we let $1_{ij}\coloneqq I(i=j)$ and $\xi_{ij,s}\coloneqq \varepsilon_{is}\varepsilon_{js} - 1_{ij}$, and then write for $\forall i,j$

\begin{eqnarray*}
    &&\Pr \left(\max_{t\in [T-h]}  \sum_{s=1}^{t} \xi_{ij,s}\ge \delta \right) \notag \\
    &\leq &\left(1 + \frac{2}{J} \right)^{J} \frac{\sum_{t=1}^{T-h}\|\xi_{ij,s}\|_J^J}{\delta^J} + 2\exp\left(  \frac{- 2e^{-J}(J+2)^{-2} \delta^2}{\sum_{t=1}^{T-h}\|\xi_{ij,s}\|_2^2}\right)\notag \\
    &\leq &c_1\frac{T}{\delta^J} + 2\exp\left(  \frac{- c_2 \delta^2}{T}\right),
\end{eqnarray*}
where the first inequality follows from Lemma \ref{LMB1}, and we require $J$ of Assumption \ref{AS1} to satisfy that $J\ge 4$ and $2q =J$. Therefore,

\begin{eqnarray}\label{eq.prob1}
    \Pr(\max_{i,j}(T-h)|\hat{\sigma}_{ij} - 1_{ij}|\ge \delta)\le c_1\frac{N^2  T}{\delta^J} + 2N^2 \exp\left(  \frac{- c_2 \delta^2}{T}\right).
\end{eqnarray}
Consequently, it yields that 

\begin{eqnarray*}
    \Pr((T-h)\|\hat{\pmb{\Sigma}}-\mathbf{I}_N\|_{\max}\ge \delta)&\le &c_1\frac{N^2  T}{\delta^{2J}} + 2N^2 \exp\left(  \frac{- c_2 \delta^2}{T}\right)\notag \\
    &=& \frac{N^2 }{T^{J-1}\log N^c}+ \frac{N^2}{\exp\left( \log N^{cc_2}\right)} \to  0,
\end{eqnarray*}
where the first equality follows by letting $\delta= (T\log N^c)^{1/2}$, and the last step follows as long as $cc_2>2$. 

In connection with \eqref{rate.sum1} and \eqref{rate.sum2}, we now can conclude that 

\begin{eqnarray*}
    \max_{m,n}\left|\frac{1}{T-h} \sum_{t=1}^{T-h} \bar{\mathbf{U}}\tilde{\pmb{\varepsilon}}_t - \sum_{\ell=0}^\infty \pmb{\beta}_{\ell, m}^\top \pmb{\beta}_{\ell, n}  \right|=O_P\left(\frac{d_{\pmb{\beta}}^2\sqrt{\log N}}{\sqrt{T}} \right).
\end{eqnarray*}
The proof of the first result is completed.

\medskip

(2). Next, we study 

\begin{eqnarray*}
    \left|\sum_{t=1}^{T-h}\sum_{j=1}^{\infty}\mathbf{U}_{j\centerdot}(L)\tilde{\pmb{\varepsilon}}_{t, \centerdot j} \right| &=& \left|\sum_{t=1}^{T-h}\sum_{j=1}^\infty \sum_{\ell=0}^{\infty}\mathbf{e}_m^\top\mathbf{B}_\ell \pmb{\varepsilon}_{t-\ell}\pmb{\varepsilon}_{t-\ell-j}^\top \mathbf{B}_{\ell+j}^\top \mathbf{e}_n\right| \notag \\
    &\le & \sum_{j=1}^\infty \|\pmb{\beta}_{j, n}\|_1\sum_{\ell=0}^{\infty}\|\pmb{\beta}_{\ell, m}\|_1\left\|\sum_{t=1}^{T-h}\pmb{\varepsilon}_{t-\ell}\pmb{\varepsilon}_{t-\ell-j}^\top\right\|_{\max}.
\end{eqnarray*}

First, note that for $\forall j$, the elements of $\pmb{\varepsilon}_{t}\pmb{\varepsilon}_{t-j}^\top$ are martingale difference sequence adapt to the filtration $\mathscr{F}_t =\{\pmb{\varepsilon}_{t}, \pmb{\varepsilon}_{t-1},\ldots \}$. Therefore, in view of the proof of the first result and the discussion under Eq. (1.8) of \cite{Rio2017}, we know that, for $\forall \ell\ge 0, \forall j\ge 1$,

\begin{eqnarray*}
    \left\|\frac{1}{T-h}\sum_{t=1}^{T-h}\pmb{\varepsilon}_{t-\ell}\pmb{\varepsilon}_{t-\ell-j}^\top\right\|_{\max} =O_P\left(\frac{\sqrt{\log N}}{\sqrt{T}}\right).
\end{eqnarray*}
Thus, we can write further 

\begin{eqnarray*}
    \left|\sum_{t=1}^{T-h}\sum_{j=1}^{\infty}\mathbf{U}_{j\centerdot}(L)\tilde{\pmb{\varepsilon}}_{t, \centerdot j} \right| &\lesssim & \frac{\sqrt{\log N}}{\sqrt{T}}\left(\sum_{\ell=0}^{\infty}\|\mathbf{B}_{\ell}\|_\infty\right)^2 = \frac{d_{\pmb{\beta}}^2\sqrt{\log N}}{\sqrt{T}}.
\end{eqnarray*}
The proof of the second result is now completed.

\medskip

(3). By the first two results of Lemma \ref{LMB6}, we obtain that 

\begin{eqnarray*}
    \left\|\frac{1}{T-h} \sum_{t=1}^{T-h} \mathbf{x}_t\mathbf{x}_{t}^\top - \sum_{\ell = 0}^{\infty}  \mathbf{B}_{\ell} \mathbf{B}_\ell^\top \right\|_{\max} =O_P \left(\frac{d_{\pmb{\beta}}^2\sqrt{\log N}}{\sqrt{T}} \right).
\end{eqnarray*}

\medskip

(4). Similar to the third result, 

\begin{eqnarray*}
    \left\|\frac{1}{T-h} \sum_{t=1}^{T-h} \mathbf{x}_t\mathbf{x}_{t-j}^\top - \sum_{\ell = 0}^{\infty}  \mathbf{B}_{\ell+j} \mathbf{B}_\ell^\top \right\|_{\max} =O_P \left(\frac{d_{\pmb{\beta}}^2\sqrt{\log N}}{\sqrt{T}} \right)
\end{eqnarray*}
for $j=1,\ldots, p-1$. The proof is now complete.
\end{proof}

\medskip

\begin{proof}[Proof of Lemma \ref{LMB7}]
\item 
  
(1). For notational simplicity, we suppress the indices $(i,j)$ in $U_{ij} (t,h)$ and $U_{ij,l}^\ast (t,h)$. 

\medskip

\noindent Case 1. We replace $\pmb{\varepsilon}_{t+h-l}$ with $\pmb{\varepsilon}_{t+h-l}^*$ for $l=0,\ldots, h-1$. It is straightforward to obtain  
\begin{eqnarray}\label{mm3}
   &&\| U(t,h)-U_{t+h-l}^\ast (t,h)\|_J \notag \\
   &\leq&\|g_{i}(\pmb{\varepsilon}_{t+h}, \ldots,\pmb{\varepsilon}_{t+h-l} \ldots,\pmb{\varepsilon}_{t+1})-g_{i}(\pmb{\varepsilon}_{t+h}, \ldots, \pmb{\varepsilon}_{t+h-l}^\ast, \ldots,\pmb{\varepsilon}_{t+1})\|_J\Big\|\sum_{\ell=0}^{\infty}\pmb{\beta}_{\ell,j}^\top\pmb{\varepsilon}_{t-\ell}\Big\|_J
   \notag\\
   &\leq&\delta_{J,i}(h,l)\Big\|\sum_{\ell=0}^{\infty}\pmb{\beta}_{\ell,j}^\top\pmb{\varepsilon}_{t-\ell}\Big\|_J
   \notag\\
   &=&\delta_{J,i}(h,l)\Big\|\sum_{i=1}^N\sum_{\ell=0}^{\infty}b_{\ell,ji}\varepsilon_{i, t-\ell}\Big\|_J.
\end{eqnarray} 
For $J\geq2$,  by applying  Burkholder inequality and Minkowski inequality sequentially, 

\begin{eqnarray}\label{mm1}
    \Big\|\sum_{i=1}^N\sum_{\ell=0}^{\infty}b_{\ell,ji}\varepsilon_{i, t-\ell}\Big\|^2_J&\leq& O(1)\left\|\sum_{i=1}^N\sum_{\ell=0}^{\infty}b_{\ell,ji}^2\varepsilon_{i, t-\ell}^2\right\|_{\frac{J}{2}}
    \leq O(1)\sum_{i=1}^N\sum_{\ell=0}^{\infty}\|b_{\ell,ji}\varepsilon_{i, t-\ell}\|_J^{2}
    \notag\\
    &=&O(1)\sum_{i=1}^N\sum_{\ell=0}^{\infty}|b_{\ell,ji}|^2.
\end{eqnarray}
Thus, by \eqref{mm1}, 
\begin{eqnarray*}
    \Big\|\sum_{i=1}^N\sum_{\ell=0}^{\infty}b_{\ell,ji}\varepsilon_{i, t-\ell}\Big\|_J\leq O(1)\left(\sum_{i=1}^N\sum_{\ell=0}^{\infty} |b_{\ell,ji}|^{2}\right)^{\frac{1}{2}}.
\end{eqnarray*}
Together with \eqref{mm3}, it gives 
\begin{eqnarray*}
   \| U(t,h)-U_{l}^\ast (t,h)\|_J&\leq&O(1)\delta_{J,i}(h,l)\left(\sum_{i=1}^N\sum_{\ell=0}^{\infty} |b_{\ell,ji}|^{2}\right)^{\frac{1}{2}},
\end{eqnarray*}
for $l\in [h]$. 


\noindent Case 2. We replace $\pmb{\varepsilon}_{t+h-l}$ with $\pmb{\varepsilon}_{t+h-l}^*$ for $l=h,h+1,\ldots$ Using arguments that are analogous to \eqref{mm1},  we obtain
\begin{eqnarray*}
     \|U(t,h)-U_{t+h-l}^\ast (t,h)\|_J&=& \|g_{i}(\pmb{\varepsilon}_{t+h}, \ldots, \pmb{\varepsilon}_{t+1})\pmb{\beta}_{l-h,j}^\top(\pmb{\varepsilon}_{l}-\pmb{\varepsilon}_{l}^\ast) \|_J
     \notag\\
     &\leq&O(1)\left\|\sum_{i=1}^N b_{l-h,ji}\varepsilon_{i, t-\ell}\right\|_J
     \leq O(1)\left(\sum_{i=1}^N |b_{l-h,ji}|^{2}\right)^{\frac{1}{2}},
\end{eqnarray*}
which completes the proof of Lemma \ref{LMB7}.(1). 

\medskip

(2). For $U(t,h)-U(t,h,m)$, it admits the following MDS decomposition:
\begin{eqnarray*}
    U (t,h)-U(t,h,m)&=&\sum_{l=m+1}^{\infty}(U(t,h,l)-U(t,h,l-1))\eqqcolon \sum_{l=m+1}^{\infty}D_{t,h,l},\notag
\end{eqnarray*}
where $D_{t,h,l}$ is a MDS adapted to  $\mathscr{F}_{t+h,t+h-l}$ for each given $t$. 

For $J\geq2$, by Burkholder inequality,  Minkowski inequality and Jensen inequality, 

\begin{eqnarray}\label{mm4}
   \|U (t,h)-U(t,h,m)\|^2_J
   &\leq& O(1)\left\|\sum_{l=m+1}^{\infty}D_{t,h,l}^2\right\|_{\frac{J}{2}}
   \notag\\
   &\leq& O(1) \sum_{l=m+1}^{\infty}\|D_{t,h,l}\|_J^2
   \notag\\
   &=& O(1) \sum_{l=m+1}^{\infty}\|U(t,h,l)-U_{t+h-l}^\ast(t,h,l)\|_J^2
   \notag\\
   &= & O(1) \sum_{l=m+1}^{\infty} \|U(t,h)-U_{t+h-l}^\ast (t,h)\|_J^2.
\end{eqnarray}

In connection with Lemma \ref{LMB7}.(1), it suffices to establish Lemma \ref{LMB7}.(2).


\medskip

(3). Notably, $D_{t,h,l}$ is also a (reverse) MDS adapted to $\sigma(\pmb{\varepsilon}_t,\ldots,\pmb{\varepsilon}_{\infty})$ for each given $l$. Thus,  similarly to \eqref{mm4}, we can obtain 

\begin{eqnarray*}
     \left\|\sum_{t=1}^{T-h}D_{t,h,l}\right\|_J^{2}&\leq& O(1) \sum_{t=1}^{T-h} \|U(t,h)-U_{t+h-l}^\ast (t,h)\|^{2}_J.
\end{eqnarray*}
which yields
\begin{eqnarray}\label{mm41}
    \left\| \sum_{t=1}^{T-h}(U(t,h)-U (t,h,m))\right\|_J&\leq&  \sum_{l=m+1}^{\infty} \left\|\sum_{t=1}^{T-h}D_{t,h,l}\right\|_J
    \notag\\
    &\leq& O(1)  \sum_{l=m+1}^{\infty} \left(\sum_{t=1}^{T-h} \|U(t,h)-U_{t+h-l}^\ast (t,h)\|^{2}_J\right)^{\frac{1}{2}}
    \notag\\
    &=& O(1)T^{1/2} \sum_{l=m+1}^{\infty}\|U(t,h)-U_{t+h-l}^\ast (t,h)\|_J.
\end{eqnarray}
This completes the proof of the first result in Lemma \ref{LMB7}.(3). 

For the second result in Lemma~\ref{LMB7}.(3), observe that \(U(t,h)\) admits the expansion   $U(t,h)-U(t,h,0)=\sum_{l=1}^{\infty} D_{t,h,l}$, where \(U(t,h,0)=E[u_{it,h}x_{jt}\mid \mathscr{F}_{t+h,t+h}]\) is a MDS.  
Following similar steps as in \eqref{mm41} and applying Burkholder inequality, Minkowski inequality, and Lemma~\ref{LMB7}.(1) yields
\begin{eqnarray}\label{mm42}
    \left\| \sum_{t=1}^{T-h}(U(t,h)-U (t,h,0))\right\|_J
   = O\left(T^{1/2}\right) \sum_{l=1}^{\infty} \|U(t,h)-U^{\ast}_{t+h-l}(t,h)\|_J
   = O\left(T^{1/2}\right).
\end{eqnarray}
For \(U(t,h,0)\), standard arguments for MDS imply that $ \|\sum_{t=1}^{T-h} U(t,h,0)\|_J = O(T^{1/2}),$ which in connection with \eqref{mm42} establishes the second claim in Lemma \ref{LMB7}.(3).
\end{proof}

\medskip
   
\begin{proof}[Proof of Lemma \ref{LMB8}]

\item 

(1). Recall $U_{ij} (t,h,m)= E[u_{it,h} x_{jt} \mid \mathscr{F}_{t+h,t+h-m}]$. Let $V_{ij}(t,h, m)=\sum_{1 \le k< t }U_{ij} (k,h,m)$. We consider the following decomposition for $V_{ij}(t,h, m)$: 
\begin{eqnarray*}
    V_{ij}(t,h, m)&=&\sum_{t_m \le k< t } a_{k,t}U_{ij} (k,h,m)+\sum_{1\le k< t_m-1 } a_{k,t} U_{ij} (k,h,m)
    \notag\\
    &\eqqcolon &V_{ij, 1}(t,h,m)+V_{ij, 2}(t,h,m),
\end{eqnarray*}
where $t_m=t-3m$. 

We further define 

\begin{eqnarray*}
    Q_{i_1i_2j_1j_2}(T, h,m) &=& \sum_{t=2}^{T-h}  U_{i_1j_1} (t,h,m)V_{i_2j_2}(t,h,m),\notag \\
    \widetilde{U}_{ij} (t,h,m)&=& U_{ij} (t,h,\lfloor m^\theta \rfloor).
\end{eqnarray*}
By replacing $V_{i_2j_2}(t,h,m)$ in $Q_{i_1i_2j_1j_2}(T, h,m)$ with $V_{i_2j_2,1}(t,h,m)$ and $V_{i_2j_2, 2}(t,h, m)$, we define $Q_{i_1i_2j_1j_2,1}(T, h,m)$ and $Q_{i_1i_2j_1j_2,2}(T, h,m)$, respectively. Additionally, we define 

\begin{eqnarray*}
    &&\widetilde{V}_{ij}(t,h,m),\quad \widetilde{V}_{ij, 1}(t,h,m), \quad \widetilde{V}_{ij, 2}(t,h,m), \notag \\
    &&\widetilde{Q}_{i_1i_2j_1j_2}(T, h,m), \quad \widetilde{Q}_{i_1i_2j_1j_2,1}(T, h,m),\quad \widetilde{Q}_{i_1i_2j_1j_2,2}(T, h,m)
\end{eqnarray*}
by replacing $U_{ij} (t,h,m)$ with $\widetilde{U}_{ij} (t,h,m)$ in

\begin{eqnarray*}
    && V_{ij}(t,h,m),\quad V_{ij, 1}(t,h,m),\quad V_{ij, 2}(t,h,m), \notag \\
    && Q_{i_1i_2j_1j_2}(T, h,m),\quad Q_{i_1i_2j_1j_2,1}(T, h,m),\quad Q_{i_1i_2j_1j_2,2}(T, h,m),
\end{eqnarray*}
respectively. For notational simplicity, we suppress the indices $i_1,i_2,j_1$ and $j_2$ whenever possible. 

Using these notations, we write

\begin{eqnarray}\label{mm18}
    &&P_{h}(T, m,\epsilon_{NT})\notag \\
    &=&\Pr\big(|Q(T, h,m)-E[Q(T, h,m)]|\ge \epsilon_{NT} \big)
    \notag\\
    &\leq&\Pr\big(|Q_1(T, h,m)-\widetilde{Q}_1(T, h,m)-E[Q_1(T, h,m)]+E[\widetilde{Q}_1(T, h,m)]|\ge \epsilon_x\big)
    \notag\\
    &&+ \Pr\big(|Q_2(T, h,m)-\widetilde{Q}_2(T, h,m)-E[Q_2(T, h,m)]+E[\widetilde{Q}_2(T, h,m)]|\ge\epsilon_x\big)
    \notag\\
    &&+ P_{h}(T, \lfloor m^\theta\rfloor,\epsilon_{NT}-2\epsilon_x),
\end{eqnarray}
where $\epsilon_x=\epsilon_{NT}/(2l_{T,\theta})$ with $l_{T,\theta}=\lceil-\log\log T/\log\theta \rceil$. We study the first two probabilities on the right hand side of \eqref{mm18} one by one.

For the term with $Q_1(T,h,m)$ and $\widetilde{Q}_1(T,h,m)$, partition \([T]\) into consecutive blocks \(B_1,\dots,B_{b_T}\) of length \(4m\).
For each block set
\[
Q_{1,n}(T, h,m):=\sum_{t\in B_n}\big(U(t,h,m)V_{  1}(t,h,m)-\widetilde{U}(t,h,m)\widetilde{V}_{ 1}(t,h,m)\big).
\]
By construction, blocks that are not adjacent involve disjoint time windows for the underlying innovations, hence \(Q_{1,n}\) and \(Q_{1,n^\ast}\) are independent whenever \(|n-n^\ast|>1\).

Using similar arguments to those in the proof of \eqref{mm17}, we have
\begin{eqnarray*}
    \|Q_{1,n}(T, h,m)-E[Q_{1,n}(T, h,m)]\|_{J/2}&=&O (m^{1-c_0\theta} ),
\end{eqnarray*}
where $c_0 =\big(\frac{2c_b-1}{2}\big)\big(\frac{c_\delta\wedge c_b-1}{c_\delta\wedge c_b}\big)$. Together with Lemma \ref{LMB2} and the Markov inequality, it yields that there exists a constant $M>1$ such that 
\begin{eqnarray}\label{mm19}
    &&\Pr(|Q_1(T, h,m)-\widetilde{Q}_1(T, h,m)-E[Q_1(T, h,m)]+E[\widetilde{Q}_1(T, h,m)]|\ge \epsilon_x)\notag \\
    &\le &\sum_{n=1}^{b_T} \Pr(|Q_{1,n}(T, h,m)-E[Q_{1,n}(T, h,m)]|\ge \epsilon_y) +O\left((N\wedge T)^{-M}\right)
    \notag \\
    &\le &\sum_{n=1}^{b_T} \epsilon_y^{-J/2}   \|Q_{1,n}(T, h,m)-E[Q_{1,n}(T, h,m)]\|^{J/2}_{J/2} +O\left((N\wedge T)^{-M}\right)
    \notag \\
    &=& O\left(\epsilon_y^{-J/2}Tm^{(1-c_0\theta)J/2-1}\right) +O\left((N\wedge T)^{-M}\right),
\end{eqnarray}
where $\epsilon_y=\frac{\beta\epsilon_x}{2c_{M}}$,  with $\beta = J/(J+4)$, $\alpha = 1-\beta$, and a constant $c_{M}>\beta$.

We then turn to the investigation of the second probability on the right-hand side of \eqref{mm18}.  We partition \([T]\) into $b_T^*$ blocks \(B_1^*,\dots,B_{b_T^*}^*\) of length \(m\). For each block, we define
\[
Q_{2,n}(T, h,m):=\sum_{t\in B_n}\big(U(t,h,m)V_{ 2}(t,h,m)-\widetilde{U}(t,h,m)\widetilde{V}_{  2}(t,h,m)\big).
\]
Let \(l_n=\max B_n^*\) and the sigma field $\mathscr{F}^\ast_n:=\mathscr{F}_{l_n}$. We can then easily check that for every \(t\in B_n^*\) the corresponding $V_{  2}(t,h,m)$ and $\widetilde{V}_{  2}(t,h,m)$ are \(\mathscr{F}^\ast_{n-2}\)-measurable, and moreover, $U(t,h,m)$ and $\widetilde{U}(t,h,m)$ are  \(\mathscr{F}^\ast_{n-2}\) -conditionally mean zero:
\begin{eqnarray*}
 E\big[Q_{2,n}(T, h,m)\mid\mathscr{F}^\ast_{n-2}\big]=0.   
\end{eqnarray*}
Therefore,  $\{Q_{2,n}(T, h,m)\}_{n\ \text{odd}}$ is a MDS adapted to $\{\mathscr{F}^\ast_n\}_{n\ \text{odd}}$ and similarly for the even subsequence.  Additionally, we can observe that $U(t,h,m)$ and $\widetilde{U}(t,h,m)$ are independent with $V_{  2}(t,h,m)$ and $\widetilde{V}_{  2}(t,h,m)$ by definition.  Then,  we apply Lemma 1 of \cite{haeusler1984exact} to the odd and even subsequences to show that for any $M>1$, there exists a constant $C^\ast_{M}$  such that
\begin{eqnarray}\label{eq:Haeusler}
&&\Pr\Big(\Big|\sum_{n=1}^{b_T^*}Q_{2,n}(T, h,m)\Big|\ge \epsilon_x\Big)\notag \\
&\le& 4\Pr\Big(\Big|\sum_{n=1}^{b_T^*} E\big[Q^2_{2,n}(T, h,m)\mid\mathscr{F}^\ast_{n-2}\big]\Big|> \epsilon_x^2/(\log T)^{3/2}\Big)\notag\\
&&+\sum_{n=1}^{b_T^*}\Pr\big(|Q_{2,n}(T, h,m)|\ge \epsilon_x/\log T\big)+C^\ast_{M}\,\epsilon_x^{-M}.
\end{eqnarray}
For the first term on the right-hand side of \eqref{eq:Haeusler}, simple algebra gives

\begin{eqnarray*}
    &&\sum_{n=1}^{b_T^*} E\big[Q^2_{2,n}(T, h,m)\mid\mathscr{F}^\ast_{n-2}\big]\notag \\
    &\leq&2\sum_{n=1}^{b_T^*} \sum_{t,s\in B_n}a_{s,t} E\big[U(t,h,m)U(s,h,m)\big]V_{ 2}(t,h,m)V_{ 2}(s,h,m)
    \notag\\
    &&+2\sum_{n=1}^{b_T^*}\sum_{t,s\in B_n} a_{s,t}E\big[\widetilde{U}(t,h,m)\widetilde{U}(s,h,m)\big]\widetilde{V}_{ 2}(t,h,m)\widetilde{V}_{ 2}(s,h,m)
    \notag\\
    &\coloneqq&\mathcal{D}_1+\mathcal{D}_2.
\end{eqnarray*}
Let $b_{s,t} = E\big[U(t,h,m)U(s,h,m)\big]$ by reorganizing the summations, we have 
\begin{eqnarray*}
    \mathcal{D}_1&=&2\sum_{n=1}^{b_T^*} \sum_{t,s\in B_n}a_{s,t}b_{s,t}V_{ 2}(t,h,m)V_{2}(s,h,m) 
    \notag\\
   &=&2\sum_{n=1}^{b_T^*} \sum_{t,s\in B_n}a_{s,t}b_{s,t} \sum_{1\le k_1< t_m-1 }U (k_1,h,m)\sum_{1\le k_2< s_m-1 }U (k_2,h,m)
   \notag\\
   &=&\sum_{k_1=2}^{T-h}\sum_{k_2=2}^{T-h}a^\ast_{k_1,k_2}U (k_1,h,m)U (k_2,h,m),
\end{eqnarray*} 
where $a^\ast_{k_1,k_2}=2\sum_{n=1}^{b_T^*}\sum_{t\in B^\ast(n,k_1)}\sum_{s\in B^\ast(n,k_2)}a_{s,t}b_{s,t}$ with  $B^\ast(n,k)=\{t\in B_n: t>k+3m+1 \}.$ By Lemma \ref{LMB7}, we have  $\|\sum_{t=2}^{T-h}U (t,h,m)\|_2\leq O\left(T^{1/2}\right)$, which implies that $|a^\ast_{k_1,k_2}|\leq C_aT$, for some positive constant $C_a$. Similarly, we have
\begin{eqnarray*}
   \mathcal{D}_2=2\sum_{n=1}^{b_T^*} \sum_{t,s\in B_n}a_{s,t}\widetilde{b}_{s,t}\widetilde{V}_{2}(t,h,m)\widetilde{V}_{ 2}(s,h,m) 
   &=&\sum_{k_1=2}^{T-h}\sum_{k_2=2}^{T-h}\widetilde{a}^\ast_{k_1,k_2}\widetilde{U} (k_1,h,m)\widetilde{U} (k_2,h,m),
\end{eqnarray*} 
where $\widetilde{b}_{s,t} = E\big[\widetilde{U}(t,h,m)\widetilde{U}(s,h,m)\big]$ the coefficients also satisfy $|\widetilde{a}^\ast_{k_1,k_2}|\leq C_a T$. Since $E[\mathcal{D}_1]\epsilon_x^{-2}(\log T)^{2}$ and $E[\mathcal{D}_2]\epsilon_x^{-2}(\log T)^{2}$ are negligible for a sufficiently large $T$, in connection with the definition of $P_{h}(T, m,\epsilon)$, 
\begin{eqnarray}\label{mm20}
    &&4\Pr\Big(\Big|\sum_{n=1}^{b_T^*} E\big[Q^2_{2,n}(T, h,m)\mid\mathscr{F}^\ast_{n-2}\big]\Big|> \epsilon_x^2/(\log T)^{3/2}\Big)
    \notag\\
    &\leq&4\Pr\Big(|\mathcal{D}_1-E[\mathcal{D}_1]|\ge  \epsilon_x^2/(\log T)^{2}\Big)+4\Pr\Big(|\mathcal{D}_2-E[\mathcal{D}_2]|\ge  \epsilon_x^2/(\log T)^{2}\Big)
    \notag\\
    &\leq&O(1)(P_{h}(T, m,\epsilon_x^2T^{-1}(\log T)^{-2})+P_{h}(T, \lfloor m\rfloor^\theta,\epsilon_x^2T^{-1}(\log T)^{-2})).
\end{eqnarray}

For the second term on the right-hand side of \eqref{eq:Haeusler}, by Markov inequality and applying  arguments similar to those in the proof of \eqref{mm17}, we  obtain 
\begin{eqnarray}\label{mm21}
    \sum_{n=1}^{b_T^*}\Pr\big(|Q_{2,n}(T, h,m)|\ge \epsilon_x/\log T\big)&\leq& \sum_{n=1}^{b_T^*}\epsilon_x^{-J}(\log T)^J\|Q_{2,n}(T, h,m)\|_J^J
    \notag\\
    &=&O\left(\epsilon_x^{-J}(\log T)^JT^{J/2+1}m^{(1/2-c_0\theta)J-1}\right)
    .
\end{eqnarray}
Together with \eqref{mm20}, it gives 
\begin{eqnarray}\label{mm22}
   &&\Pr\Big(\Big|\sum_{n=1}^{b_T^*}Q_{2,n}(T, h,m)\Big|\ge \epsilon_x\Big) \notag \\
   &\leq&O(1)(P_{h}(T, m,\epsilon_x^2T^{-1}(\log T)^{-2})+P_{h}(T, \lfloor m\rfloor^\theta,\epsilon_x^2T^{-1}(\log T)^{-2}))
   \notag\\
   &&+O\left(\epsilon_x^{-J}(\log T)^JT^{J/2+1}m^{(1/2-c_0\theta)J-1}\right)+O\left(\epsilon_x^{-M}\right).
\end{eqnarray}

By \eqref{mm18}, \eqref{mm19}, and   \eqref{mm22}
\begin{eqnarray}\label{mm23}
     P_{h}(T, m,\epsilon_{NT})&\leq&P_{h}(T, \lfloor m^\theta\rfloor,\epsilon_{NT}-2\epsilon_x)+O(1)P_{h}(T, m,\epsilon_x^2T^{-1}(\log T)^{-2})
     \notag\\
     &&+O(1)P_{h}(T, \lfloor m\rfloor^\theta,\epsilon_x^2T^{-1}(\log T)^{-2}) 
     +O\left(\epsilon_x^{-J/2}Tm^{(1-c_0\theta)J/2-1}\right) 
    \notag\\
     && +O\left((N\wedge T)^{-M}\right)+O\left(\epsilon_x^{-M}\right).
\end{eqnarray}
Since $\epsilon_x^2T^{-1}(\log T)^{-2}>\epsilon_{NT} T^{c_{\delta}/2}$ for a sufficiently large $T$,  for the term $$P_{h}(T, m,\epsilon_x^2T^{-1}(\log T)^{-2}),$$ by repeating the expansion in \eqref{mm23} for a sufficiently large number of times, we can show that it is bounded by the rest of terms.  Therefore, 
\begin{eqnarray}\label{mm24}
     P_{h}(T, m,\epsilon_{NT})&\leq&O(1)P_{h}(T, \lfloor m^\theta\rfloor,\epsilon_{NT}-2\epsilon_x)+O\left(\epsilon_x^{-J}(\log T)^JT^{J/2+1}m^{(1/2-c_0\theta)J-1}\right)
     \notag\\
     &&
     +O\left((N\wedge T)^{-M}\right)+O\left(\epsilon_x^{-M}\right).
\end{eqnarray}

Starting from a sufficiently small $m$ (e.g., $m=1,2$), we can apply similar arguments and obtain
\begin{eqnarray}\label{mm25}
     P_{h}(T, m,2\epsilon_x)&=O\left(l_{T,\theta}^{J/2}\epsilon_x^{-J/2}T\right)
      +O\left((N\wedge T)^{-M}\right)+O\left(l_{T,\theta}^{M}\epsilon_x^{-M}\right). 
\end{eqnarray}
Since   $ T^{\theta^{l_{T,\theta}}}<e$,  iterating \eqref{mm24} for $l_{T,\theta}-1$ times from $ P_{h}(T, m,\epsilon_{NT})$ to $P_{h}\left(T, \lfloor T^{\theta^{l_{T,\theta}}}\rfloor ,2\epsilon_x\right)$, we obtain 
 \begin{eqnarray}\label{mm26}
    P_{h}(T, m,\epsilon_{NT})&=&O\left(\epsilon_x^{-J/2}T\left(m^{(1-c_0\theta)J/2-1}+l_{T,\theta}^{J/2}\right)\right)\notag \\
    &&+O\left((N\wedge T)^{-M}\right)+O\left(l_{T,\theta}^{M}\epsilon_x^{-M}\right). 
 \end{eqnarray}  
Substituting \eqref{mm26} back to  \eqref{mm20}, we can update the result in \eqref{mm23} as follows:
 \begin{eqnarray*}
     P_{h}(T, m,\epsilon_{NT})&\leq&P_{h}(T, \lfloor m^\theta\rfloor,\epsilon_{NT}-2\epsilon_x) +O\left(\epsilon_{NT}^{-J}T\left(m^{(1-c_0\theta)J/2-1}+l_{T,\theta}^{J}\right)\right),
\end{eqnarray*} 
and iterating it for   $l_{T,\theta}-1$ times, we can finally establish desired result in Lemma \ref{LMB8}.(1).

(2). We now use  Lemma \ref{LMB8}.(1)   to establish  Lemma \ref{LMB8}.(2).  Define $m_{l}=\lfloor T^{\theta^l}\rfloor$, for $l=1,\ldots, l_{T,\theta}$, where $l_{T,\theta}=\lceil-\log\log T/\log\theta \rceil$. Let $l_{T,\theta}^\ast$ be the smallest $l$ such that $m_l\leq h/4$. Define 
\begin{eqnarray*}
    Q_{NT,0}=\sum_{t=2}^{T-h}U_{i_1j_1} (t,h) S_{i_2j_2,h}(t),
\end{eqnarray*}
and 
\begin{eqnarray*}
    Q_{NT,l}=\sum_{t=2}^{T-h}U_{i_1j_1} (t,h,m_l) S_{i_2j_2,h}(t,m_l),
\end{eqnarray*}
for $l=1,\ldots, l^\ast_{T,\theta}$.

We write
\begin{eqnarray*}
&&\Pr\Big(\Big|\sum_{t=2}^{T-h}\big(U_{i_1j_1} (t,h) S_{i_2j_2,h}(t)-E[U_{i_1j_1} (t,h)  S_{i_2j_2,h}(t)]\big)\Big|\ge \epsilon_{NT} \Big)
\notag\\
&\leq&\sum_{l=0}^{l^\ast_{T,\theta}-1}\Pr\Big(\Big|Q_{NT,l}-Q_{NT,l+1}-E[(Q_{NT,l}-Q_{NT,l+1})]\Big|\ge \epsilon_{NT}/(l_{T,\theta}+1) \Big)
\notag\\
&&+\Pr\Big(\Big|Q_{NT,l^\ast_{T,\theta}}-E[Q_{NT,l^\ast_{T,\theta}}]\Big|\ge \epsilon_{NT}/(l_{T,\theta}+1) \Big).
\end{eqnarray*}
Using similar arguments to \eqref{mm27}, we have
\begin{eqnarray}\label{mm32}
   &&\Pr\Big(\Big|Q_{NT,0}-Q_{NT,1}-E[(Q_{NT,0}-Q_{NT,1})]\Big|\ge \epsilon_{NT}/(l_{T,\theta}+1) \Big) \notag \\
   &= &O\left(l_{T,\theta}^{J/2}\epsilon_{NT}^{-J/2}T^{\left(1/2-c_0\theta\right)J/2}h^{J/4}\right).
\end{eqnarray}

For $1\leq l\leq l_{T,\theta}^\ast-1$,  partition \([T]\) into $b_T^l$ consecutive blocks \(B^l_1,\dots,B^l_{b_T^l}\) of length \(m_l+h\). We further write
\begin{eqnarray*}
    Q_{n,l}=\sum_{t\in B_n^l} U_{i_1j_1} (t,h,m_l)S_{i_2j_2,h} (t,m_l),\,\, Q^\ast_{n,l}=\sum_{t\in B_n^l} U_{i_1j_1} (t,h,m_{l+1})S_{i_2j_2,h} (t,m_{l+1}).
\end{eqnarray*}
Similarly to \eqref{mm27},  for $\epsilon_y<\epsilon_{NT}/(l_{T,\theta}+1)$, we can obtain  
\begin{eqnarray}\label{mm29}
   \Pr\Big(\big|Q_{n,l}-Q_{n,l}^\ast-E[(Q_{n,l}-Q_{n,l}^\ast)]\big|\geq \epsilon_y\Big)
   &= &O\left(l_{T,\theta}^{J/2}\epsilon_{NT}^{-J/2}m_l^{J/4}h^{J/4} {m_{l+1}}^{-c_0\theta J/2}\right).
\end{eqnarray}
By \eqref{mm29}, Lemma \ref{LMB2} and using arguments that are analogous to \eqref{mm13}, we can show that  there exists a constant $M>1$ such that 
\begin{eqnarray}\label{mm31}
    &&\Pr\Big(\big|Q_{NT,l}-Q_{NT,l+1}-E[(Q_{NT,l}-Q_{NT,l+1})]\big|\ge \epsilon_{NT}/(l_{T,\theta}+1) \Big) 
    \notag\\
    &\leq&\sum_{n=1}^{b_T^l} \Pr\Big(\big|Q_{n,l}-Q_{n,l}^\ast-E[(Q_{n,l}-Q_{n,l}^\ast)]\big|\geq \epsilon_y\Big)+O\left((N\wedge T)^{-M}\right)
    \notag\\
    &=&O\left(l_{T,\theta}^{J/2}\epsilon_{NT}^{-J/2}Th^{J/4} {m_{l}}^{(1/2-c_0\theta)J/2-1}\right)+O\left((N\wedge T)^{-M}\right)
    \notag\\
    &=&O\left(l_{T,\theta}^{J/2}\epsilon_{NT}^{-J/2}T^{J/4} {h}^{(1-c_0\theta)J/2-1}\right)+O\left((N\wedge T)^{-M}\right).
\end{eqnarray}

Finally, we only need to study $Q_{NT,l^\ast_{T,\theta}}-E[Q_{NT,l^\ast_{T,\theta}}]$. We re-partition  \([T]\) into $\widetilde{b}_T$ blocks \(\widetilde{B}_1,\dots,\widetilde{B}_{b_T}\) of length \(2h\), and define $\widetilde{Q}_{n,l}=\sum_{t\in \widetilde{B}_n} U_{i_1j_1} (t,h,m_{l^\ast_{T,\theta}})S_{i_2j_2,h} (t,m_{l^\ast_{T,\theta}})$. By  Lemma \ref{LMB2} and Lemma \ref{LMB8}.(1), 
\begin{eqnarray}\label{mm30}
    &&\Pr\Big(\big|Q_{NT,l_{T,\theta}^\ast}-E[Q_{NT,l_{T,\theta}^\ast}]\big|\ge \epsilon_{NT}/(l_{T,\theta}+1) \Big) 
    \notag\\
    &\leq&\sum_{n=1}^{\widetilde{b}_T} \Pr\Big(\big|\widetilde{Q}_{n,l}-E[\widetilde{Q}_{n,l}]\big|\geq \epsilon_y\Big)+O\left((N\wedge T)^{-M}\right)
    \notag\\
    &=&O\left(l_{T,\theta}^J\epsilon_{NT}^{-J}T\left(h^{(1-c_0\theta)J/2-1}+l_{h,\theta}^{J}\right)\right)+O\left((N\wedge T)^{-M}\right).
\end{eqnarray}

By \eqref{mm32}, \eqref{mm31}, and \eqref{mm30}, we can readily establish the result in Lemma \ref{LMB8}.(2).
\end{proof}
\end{appendices}

}

\end{document}